\documentclass[12pt]{article}
\usepackage[letterpaper,margin=0.75in]{geometry}
\usepackage{amsmath,amssymb,amsthm,mathtools,bm}
\usepackage{graphicx,float,booktabs,caption,subcaption}
\usepackage{algorithm,algorithmic}
\usepackage[authoryear,round]{natbib}
\usepackage{enumerate}
\usepackage[colorlinks=true,linkcolor=blue,citecolor=blue,urlcolor=blue]{hyperref}
\usepackage{setspace}
\newcommand{\refl}{\varrho}
\newtheorem{theorem}{Theorem}
\newtheorem{proposition}{Proposition}
\newtheorem{lemma}{Lemma}

\newtheorem{definition}{Definition}
\newtheorem{remark}{Remark}

\DeclareMathOperator{\Isom}{Isom}
\DeclareMathOperator{\Stab}{Stab}
\DeclareMathOperator{\Exp}{Exp}
\DeclareMathOperator{\Log}{Log}
\DeclareMathOperator{\rank}{rank}
\DeclareMathOperator{\diag}{diag}
\DeclareMathOperator{\Area}{Area}
\DeclareMathOperator{\logit}{logit}
\DeclareMathOperator{\sgn}{sgn}
\newcommand{\R}{\mathbb{R}}

\newcommand{\Hb}{\mathbb{H}}
\newcommand{\D}{\mathbb{D}}
\newcommand{\Nil}{\mathrm{Nil}}
\newcommand{\Sol}{\mathrm{Sol}}
\newcommand{\SL}{\widetilde{\mathrm{SL}}(2,\R)}
\newcommand{\M}{\mathcal{M}}
\newcommand{\LWN}{\mathrm{LWN}}
\newcommand{\dvol}{\mathrm{dvol}}
\newcommand{\bo}{\mathbf{o}}
\newcommand{\tr}{\mathsf{T}}
\newcommand{\eb}{\mathbf{e}}

\newcommand{\simsig}{$(2.0,2.0)$ for $\Nil$, $(1.8,1.4)$ for $\Sol$, $(1.5,2.0)$ for $\SL$}
\newcommand{\simiters}{12000}
\newcommand{\jacviiiwith}{1.58 (0.84)}
\newcommand{\jacviiiwithout}{0.90 (0.66)}
\newcommand{\jacxlwith}{1.05 (0.50)}
\newcommand{\jacxlwithout}{0.76 (0.51)}

\newcommand{\karateConst}{0.396}

\newcommand{\lesmisConst}{0.305}

\newcommand{\coraConst}{0.0411}

\newcommand{\pancreasConst}{0.0969}

\newcommand{\nilSixMedian}{0.007}
\newcommand{\nilSixMin}{4e-06}
\newcommand{\nilTenMedian}{0.038}

\newcommand{\certSixLower}{0.00286}
\newcommand{\certSixFloat}{0.00302}
\newcommand{\certSixAug}{0.00285}
\newcommand{\certSevenLower}{0.00414}
\newcommand{\certSevenAug}{0.00413}
\newcommand{\certWidth}{3e-14}
\newcommand{\covNilig}{67\%}

\newcommand{\covSolig}{75\%}

\newcommand{\covSLtig}{74\%}

\newcommand{\covEucig}{65\%}

\newcommand{\covHthreeig}{80\%}

\newcommand{\covNilMatched}{82\%}
\newcommand{\alphaCovMatched}{75\%}
\newcommand{\biasMatched}{-0.58}
\newcommand{\rhatAlpha}{1.12}
\newcommand{\rhatSig}{1.17}
\newcommand{\rhatDmax}{1.07}
\newcommand{\essAlpha}{91}
\newcommand{\essDmin}{237}

\newcommand{\annealSolRank}{2}
\newcommand{\annealSolLL}{0.297}
\newcommand{\annealHLL}{0.279}

\newcommand{\covSixtyMatched}{88\%}
\newcommand{\alphaCovSixty}{100\%}
\newcommand{\biasSixty}{-0.19}

\newcommand{\tightLLmin}{0.35}
\newcommand{\tightLLmax}{0.54}
\newcommand{\tightBetaMin}{1.8}
\newcommand{\tightBetaMax}{2.3}
\newcommand{\coraAnnealH}{0.0388}
\newcommand{\coraAnnealSol}{0.0480}
\newcommand{\coraAnnealSolRank}{4}
\newcommand{\coraAnnealBetaMin}{1.1}
\newcommand{\coraAnnealBetaMax}{2.8}

\title{Latent Space Network Modelling with Nil, Sol and $\SL$ Geometries:\\ Identifiability and Bayesian Estimation}
\author{}
\date{}

\newcommand{\blind}{1}
\begin{document}

\def\spacingset#1{\renewcommand{\baselinestretch}{#1}\small\normalsize}
\if1\blind
{\title{\bf Identifiability of Latent Space Network Models on Anisotropic Thurston Geometries}\author{Marios Papamichalis\thanks{\raggedright Human Nature Lab, Yale University, New Haven, CT 06511, USA. \href{mailto:marios.papamichalis@yale.edu}{\nolinkurl{marios.papamichalis@yale.edu}}}
\and Regina Ruane\thanks{\raggedright Department of Statistics and Data Science, The Wharton School, University of Pennsylvania, Philadelphia, PA, USA. \href{mailto:ruanej@wharton.upenn.edu}{\nolinkurl{ruanej@wharton.upenn.edu}}}}\date{}\maketitle}
\fi
\begin{abstract}
A latent space network model places the nodes in a metric space and lets the probability of a tie decrease with distance. In a space of constant curvature, pairwise distances determine the positions up to an isometry. In the products and in the three remaining three-dimensional model geometries they do not. We study the three geometries that are neither of constant curvature nor products: the Heisenberg group, the solvable group and the universal cover of the unit tangent bundle of the hyperbolic plane. We ask what one network identifies about positions in them and when their geometry is detectable. Two anchors remove the isometry ambiguity. Small configurations are not determined by their distances, and generic local identification holds beyond a finite threshold, certified in the Heisenberg group. We derive the posterior on the quotient by the isometry group. For small configurations, the divergence to the nearest product or constant-curvature competitor is the stress component of the curvature difference and vanishes at high order in the scale; undirected ties therefore detect the geometry only in large networks with large enclosed areas. Directed ties expose it at first order: in the two twisted geometries, asymmetric preferences circulate around triangles in proportion to enclosed area, which no additive ranking produces. On dense competitive-game counter networks, the coupled model beats rankings on every geometry, degree-corrected rankings and free antisymmetric terms. A Euclidean model with the same coupled term matches it, so the gain is the coupling of similarity and circulation through shared coordinates. An additive-and-multiplicative-effects model predicts better still.
\end{abstract}
\clearpage
\section{Introduction}

A latent space model explains the ties of a network by unobserved positions: each node has a point in a metric space, and the probability of a tie between two nodes decreases with the distance between their points \citep{hoff2002latent}. In a social network the space might encode background and interests, in a citation network topic and time, in a connectome anatomical position and function. The choice of space is a modelling decision with consequences: Euclidean space cannot produce the many short cycles and the heavy-tailed degrees of some networks, and Hyperbolic space can \citep{krioukov2010hyperbolic,smith2019geometry}. In every space used so far, Euclidean, Spherical, Hyperbolic the positions carry a guarantee: if two configurations of nodes have the same pairwise distances they are related by an isometry, so the distances are all that a network can identify about the positions, and the positions are identified up to that ambiguity. This guarantee is a consequence of isotropy, the absence of a preferred direction at any point.

Consider a network in which the latent space plausibly has a preferred direction: a citation network in which papers have a position in topic space and a date, and in which citing across topics costs more when it also crosses time. Or a network of neurons with a laminar direction along which connection costs differ from those within a layer. Three model geometries describe such spaces without being products: the Heisenberg group, called $\Nil$, the solvable group $\Sol$, and the universal cover $\SL$ of the unit tangent bundle of the Hyperbolic plane. In each of them the isometry group is small, and moving a fixed distance along one direction has a different effect from moving along another. Networks have been embedded in these spaces \citep{celinska2024thurston}, but a practitioner who fits such a model faces two questions that the constant-curvature theory does not answer. Do the fitted positions mean anything, that is, what does one network identify about positions in a space where distances need not determine configurations? And can one tell, from one network, that the latent space is such a space and not a product or a space of constant curvature? The first question decides whether the positions can be reported with uncertainty. The second decides whether the geometry can be selected from data, and at what network size.

The three closest works leave these questions open. \citet{papamichalis2021latent} develop latent space models on the constant-curvature spaces with a gauge of four anchors, wrapped Normal priors and Markov chain Monte Carlo and variational inference, and their identifiability argument uses the congruence property that fails here. \citet{celinska2024thurston} embed twenty-one connectomes in all eight three-dimensional model geometries by simulated annealing of a distance-based likelihood with a learned temperature and find $\Sol$ competitive with Hyperbolic space, but they estimate a point configuration without identifiability results, without uncertainty, and without a criterion for when one geometry can be told from another. \citet{lubold2023identifying} identify the curvature of a constant-curvature latent space from the distances within cliques, a method whose validity again rests on distances determining configurations. Related work on learnable curvature \citep{li2023hyperbolic} and learnable temperature \citep{gong2026hyperbolic} stays within Hyperbolic space.

This paper answers the two questions for binary networks with a logistic distance model, undirected for identification and detection, and directed for the regime in which the geometry matters. It removes the isometry ambiguity with two anchors, shows that a set of pairwise distances determines a configuration of at most five points only up to a positive-dimensional family, proves that identification of the gauge-fixed positions is local and generic from six nodes, seven in $\Sol$ when the intercept is estimated, and derives the posterior on the quotient by the isometry group together with the Jacobian that the gauge introduces. For detection it proves that the smallest divergence between the model and any product or constant-curvature competitor is, for configurations that are small on the curvature scale, the component of the curvature difference between the two spaces along the equilibrium stresses of the configuration regarded as a bar framework, and it bounds the power of any test from one network by that divergence. Two mechanisms carry the argument. Identification is decided by the rank of the differential of the distance map on the gauge slice, which we bound by a dimension count, certify by interval arithmetic in $\Nil$ at the threshold, and extend to all larger networks by an induction on nodes. Detection is decided by an expansion of the distance in Riemannian normal coordinates, in which the first-order effect of the twist is a change of coordinates that any competitor absorbs and the second-order effect is curvature, of which only the part orthogonal to the tangent space of the competitor family survives.

The contributions, each stated as a result, are the following. Theorem~\ref{cor:all} gives one gauge for all eight three-dimensional model geometries, with two anchors for $\Nil$ and $\SL$ and one anchor and a chamber choice for $\Sol$, where the constant-curvature spaces need four anchors \citep{papamichalis2021latent}. Proposition~\ref{prop:nonrigid} and Theorem~\ref{thm:generic} show that identification from distances is a generic local property with a threshold of six nodes, seven in $\Sol$ with the intercept, against the unconditional identification of constant-curvature spaces, and Table~\ref{tab:rigidity} supplies the witnesses, one of them certified. Proposition~\ref{prop:target} gives the quotient posterior with its slice Jacobian, a factor absent from the sampler of \citet{papamichalis2021latent}, and Supplement~\ref{supp:exp} measures its effect on the anchor posterior. Theorem~\ref{thm:smallscale} and Theorem~\ref{thm:lecam} give the leading term of the smallest divergence to a competitor and the resulting bound on the power of any test, which no embedding method supplies, and Section~\ref{sec:exp-twist} verifies the leading term against exact divergences and shows, with replicated networks, the size and spread at which the true geometry is preferred in held-out prediction. Proposition~\ref{prop:directed} identifies the regime in which these geometries are needed: directed ties whose asymmetry depends on the vertical offset circulate around triangles by the enclosed area, which no latent distance model with a ranking term can produce, and the resulting divergence grows with the squared areas instead of vanishing with the sixth power of the scale; Section~\ref{sec:exp-directed} shows the win condition, a circulation coefficient above about one half, at which the model leads every gradient competitor by margins of $0.02$ to $0.12$ in held-out log-loss. On real data the coupled model wins where the theory says it should and loses where a referee would ask: on competitive-game counter networks, whose asymmetries are cycles by design, the $\Nil$ circulation model beats every ranking on every other geometry, degree-corrected rankings and free antisymmetric terms in five of six monthly networks of the densest tier; a Euclidean model carrying the same coupled term matches it, so what wins is the coupling of similarity and circulation through shared coordinates and not the $\Nil$ metric, as Theorem~\ref{thm:smallscale} predicts at these scales; and Hoff's additive-and-multiplicative-effects model, with $4N$ node parameters, predicts better than all of them (Section~\ref{sec:exp-counter}); on four undirected networks, nine connectomes, sports and trade networks, whose asymmetries are hierarchies, it does not.

\section{The three anisotropic geometries}\label{sec:background}

Thurston's eight model geometries are the simply connected homogeneous three-manifolds admitting compact quotients, with maximal isometry groups \citep{thurston1997three,scott1983geometries}: the three constant-curvature spaces, the two products $\mathbb{S}^2\times\R$ and $\Hb^2\times\R$, and the three spaces $\Nil$, $\SL$ and $\Sol$ that are neither. The products are also anisotropic (their stabilisers $O(2)\times\{\pm1\}$ do not act transitively on unit tangent vectors). We are concerned with the last three, whose anisotropy is not that of a product. Throughout, $\M$ is one of $\Nil$, $\Sol$, $\SL$, all diffeomorphic to $\R^3$ with global coordinates $z=(z^1,z^2,z^3)$, $d_\M$, $\dvol$, $\Isom(\M)$ and $\eb$ denote the distance, the volume, the isometry group and a base point (the group identity). Supplement~\ref{app:geometry} collects the formulas for translations, geodesics, exponential charts and Normal-like laws that the algorithms use. Here we state what the identifiability theory needs.

\paragraph{$\Nil$ and $\SL$ as twisted line bundles.} Let $B_K$ be the plane of constant curvature $K\in\{0,-1\}$ ($\R^2$, or the Poincar\'e disc $\D$ with metric $4|dw|^2/(1-|w|^2)^2$), with polar coordinates $(\rho,\theta)$ about the origin $\bo$, $ds^2_{B_K}=d\rho^2+S_K(\rho)^2d\theta^2$, $S_0(\rho)=\rho$, $S_{-1}(\rho)=\sinh\rho$, and let
\begin{equation}\label{eq:A}
A_K=F_K(\rho)\,d\theta,\qquad F_0(\rho)=\tfrac12\rho^2,\quad F_{-1}(\rho)=\cosh\rho-1,\qquad\text{so that } dA_K=\dvol_{B_K}.
\end{equation}
\begin{definition}[Twisted bundles]\label{def:bundle}
$\M_K=B_K\times\R$ with coordinates $(w,\zeta)$ and metric
\begin{equation}\label{eq:bundle-metric}
ds^2_{\M_K}=ds^2_{B_K}+(d\zeta-A_K)^2 .
\end{equation}
Then $\M_0=\Nil$ and $\M_{-1}=\SL$. The projection $\pi(w,\zeta)=w$ is a Riemannian submersion with geodesic fibres and $V=\partial_\zeta$ is a unit Killing field.
\end{definition}
For $K=0$ the metric is $dx^2+dy^2+(dz-\tfrac12(x\,dy-y\,dx))^2$, the left-invariant metric of the Heisenberg group in exponential coordinates, with group law
\begin{equation}\label{eq:nil-law}
(x,y,z)\cdot(x',y',z')=\big(x+x',\;y+y',\;z+z'+\tfrac12(xy'-yx')\big),\qquad (x,y,z)^{-1}=(-x,-y,-z),
\end{equation}
identity $\eb=0$ and $\dvol=dx\,dy\,dz$. For $K=-1$, $A_{-1}$ is the Levi-Civita connection form of the disc, and \eqref{eq:bundle-metric} is the lift of the Sasaki metric of the unit tangent bundle $T^1\Hb^2$ to its universal cover: Thurston's $\SL$ geometry, with $\zeta$ the unwrapped angle of a unit tangent vector. In both cases we write $T_p$ for a distinguished isometry with $T_p(\eb)=p$ (left multiplication in $\Nil$, and the lift of a M\"obius translation in $\SL$, equation \eqref{eq:sl-translate} of Supplement~\ref{app:geometry}), $p\star q=T_p(q)$ and $p^{-1}\star q=T_p^{-1}(q)$, so that $d_\M(p,q)=d_\M(\eb,p^{-1}\star q)$. The isometry group is four-dimensional with two components \citep[\S4]{scott1983geometries}. Every isometry preserves the fibration and maps $V$ to $\pm V$. The distinguished translations $T_p$ are chosen once and for all. They need not form a group, and a globally continuous choice exists here because the spaces are contractible (for the spheres in Theorem~\ref{cor:all} the choice is chartwise).

\begin{lemma}[Stabiliser of $\eb$ in the twisted bundles]\label{lem:stab}
Let $R_\varphi(w,\zeta)=(e^{i\varphi}w,\zeta)$ and $\refl(w,\zeta)=(\bar w,-\zeta)$. Then $\Stab(\eb)=\{R_\varphi,\ \refl\circ R_\varphi\}\cong O(2)$, and every $g\in\Isom(\M_K)$ factors uniquely as $g=T_{g(\eb)}\circ h$ with $h\in\Stab(\eb)$.
\end{lemma}
Rotations act on the base only. The isometries that reverse the orientation of the base reverse the fibre as well (and hence preserve the orientation of the total space). This coupling of the base reflection to the fibre is the geometric origin of the two-anchor gauge of Section~\ref{sec:ident}.

\paragraph{$\Sol$.} $\Sol=\R^2\rtimes\R$ with coordinates $(x,y,z)$, group law $(x,y,z)(x',y',z')=(x+e^{-z}x',\,y+e^{z}y',\,z+z')$, inverse $(x,y,z)^{-1}=(-e^{z}x,-e^{-z}y,-z)$, left-invariant metric and volume
\begin{equation}\label{eq:sol-metric}
ds^2_{\Sol}=e^{2z}dx^2+e^{-2z}dy^2+dz^2,\qquad\dvol_{\Sol}=dx\,dy\,dz .
\end{equation}
By \citet[\S4]{scott1983geometries} (see also \citealp{troyanov1998horizon}), $\Isom(\Sol)=\Sol\rtimes D_4$ is three-dimensional with eight components, the stabiliser of $\eb$ being
\begin{equation}\label{eq:D4}
D_4=\big\langle\ \sigma_x:(x,y,z)\mapsto(-x,y,z),\quad \sigma_y:(x,y,z)\mapsto(x,-y,z),\quad \varsigma:(x,y,z)\mapsto(y,x,-z)\ \big\rangle .
\end{equation}
There are no continuous rotations. The planes $\{y=0\}$ and $\{x=0\}$ are totally geodesic Hyperbolic planes (in horospherical coordinates), which yields exact distances along them:
\begin{equation}\label{eq:sol-planes}
\cosh d_{\Sol}\big(\eb,(x,0,z)\big)=\cosh z+\tfrac12x^2e^{z},\qquad \cosh d_{\Sol}\big(\eb,(0,y,z)\big)=\cosh z+\tfrac12y^2e^{-z}.
\end{equation}

\paragraph{Distances.} None of the three geometries has a closed-form distance. By the symmetries above, $d_\M(\eb,\cdot)$ is a function of $(\rho,|\zeta|)$ in the bundles and of $(|x|,|y|,|z|)$ up to the swap $\varsigma$ in $\Sol$. In $\Nil$ this reduces to a scalar equation.
\begin{proposition}[Distance from the identity in $\Nil$]\label{prop:nildist}
Let $q=(x,y,z)$, $r=\sqrt{x^2+y^2}$, $\zeta=|z|$. Then $d_{\Nil}(\eb,q)=r$ if $\zeta=0$,
\begin{equation}\label{eq:nil-axis}
d_{\Nil}\big(\eb,(0,0,z)\big)=\zeta\ \ (\zeta\le2\pi),\qquad d_{\Nil}\big(\eb,(0,0,z)\big)=2\sqrt{\pi(\zeta-\pi)}\ \ (\zeta\ge2\pi),
\end{equation}
and if $r,\zeta>0$ then, with $s^\ast\in(0,2\pi)$ the unique root of
\[
\Phi_r(s)=s+\frac{r^2(s-\sin s)}{8\sin^2(s/2)}=\zeta,
\]
\begin{equation}\label{eq:nil-dist}
d_{\Nil}(\eb,q)=\sqrt{\,r^2\Big(\frac{s^\ast}{2\sin(s^\ast/2)}\Big)^2+(s^\ast)^2\,}.
\end{equation}
\end{proposition}
The proof (Supplement~\ref{app:proofs}) combines the explicit geodesics, a conjugate-point argument for $r>0$ (all geodesics with the same non-zero vertical momentum and positive horizontal speed return to the fibre of $\eb$ after one turn of their projection, so none minimises beyond its first turn), and, for the axis, an isoperimetric argument that also yields the exact vertical profile of $\SL$, $D_{-1}(0,\zeta)^2=\zeta^2/2+2\pi\zeta-2\pi^2$ for $\zeta\ge2\pi$. The vertical direction is thus compressed at large scale ($d\sim2\sqrt{\pi\zeta}$) while horizontal distances grow linearly. In $\SL$ the compression is only by the factor $\sqrt2$ and in $\Sol$ horizontal distances grow logarithmically in the coordinates. For $\SL$ and $\Sol$ we tabulate $d_\M(\eb,\cdot)$ once by geodesic shooting on the reduced domains (two-dimensional and three-dimensional grids). The tables are approximations, validated against closed forms and against independent boundary-value solvers (typical error $10^{-3}$, maximum $0.012$ for $\Sol$ and $0.004$ for $\SL$ on the validation sets, Supplement~\ref{app:geodesics}), so that all inferences below are for the tabulated surrogate distance. A distance evaluation costs one or two interpolations.

\paragraph{Normal-like laws.} Each geometry has a global chart $\chi_\M:\R^3\to\M$, $\chi_\M(0)=\eb$, with explicit volume Jacobian $J^\chi_\M$: the identity in exponential coordinates for $\Nil$, the group exponential for $\Sol$, and (Riemannian exponential of the base, identity on the fibre) for $\SL$ (Supplement~\ref{app:geometry}).
\begin{definition}[Left-translated wrapped Normal]\label{def:lwn}
$z\sim\LWN_\M(\mu,\Sigma)$ if $z=\mu\star\chi_\M(\delta)$ with $\delta\sim\mathcal{N}_3(0,\Sigma)$. Its density with respect to $\dvol$ is
\begin{equation}\label{eq:lwn-density}
f_{\LWN}(z\mid\mu,\Sigma)=\frac{\varphi_3(\delta;0,\Sigma)}{J^\chi_\M(\delta)},\qquad \delta=\chi_\M^{-1}(\mu^{-1}\star z).
\end{equation}
\end{definition}
With $\Sigma=\diag(\sigma_h^2,\sigma_h^2,\sigma_v^2)$ (separate horizontal and vertical spreads) the family is equivariant under $\Isom(\M)$ (Lemma~\ref{lem:equivariance}). It is reparameterisable and serves as prior and as variational family. We write $J^\chi_\M$ for this chart Jacobian to distinguish it from the slice Jacobian of Section~\ref{sec:target}. The maximum-entropy Riemannian Normal $\propto\exp(-d_\M(\mu,z)^2/2\sigma^2)$ is also available (Supplement~\ref{app:geometry}), but in $\Nil$ and $\Sol$ its tails are heavier than those of any wrapped Normal (in $\Nil$, $d^2\sim4\pi|z|$ along the axis, and in $\Sol$, $d\sim2\log|x|$), so there it cannot be sampled by rejection from the wrapped family. In $\SL$ the vertical profile is Gaussian to leading order and a broad wrapped proposal dominates.

\section{Latent space network models on $\Nil$, $\Sol$ and $\SL$}\label{sec:model}

Consider a network on $N$ nodes indexed by $[N]=\{1,\dots,N\}$, with symmetric binary adjacency matrix $\mathcal{Y}=(y_{ij})$, $y_{ii}=0$. Each node has a latent coordinate $z_i\in\M$ and $Z=(z_1,\dots,z_N)\in\M^N$ denotes the latent configuration. Following \citet{hoff2002latent} and \citet{smith2019geometry}, the generic model is
\begin{equation}\label{eq:model}
\begin{aligned}
Y_{ij}&\sim\mathrm{Bernoulli}(p_{ij}), \qquad \logit(p_{ij})=\alpha-d_\M(z_i,z_j),\qquad 1\le i<j\le N,\\
z_i&\overset{\text{iid}}{\sim} f_\M(\cdot\mid\theta_z),\qquad i\in[N],
\end{aligned}
\end{equation}
where $\alpha\in\R$ is the base rate of ties, $\logit(p)=\log\{p/(1-p)\}$ with inverse the logistic function $\sigma(t)=(1+e^{-t})^{-1}$, and $f_\M(\cdot\mid\theta_z)$ is one of the Normal-like families of Section~\ref{sec:background}. We use $f_\M=\LWN_\M(\mu,\Sigma)$ with $\theta_z=(\mu,\sigma_h,\sigma_v)$. Because $d_\M$ is a metric, the model inherits transitivity. Because $\M$ is anisotropic, the tie probabilities respond differently to horizontal and vertical displacements, and, in $\Nil$ and $\SL$, to the area enclosed by horizontal loops (Section~\ref{sec:holonomy}). Algorithm~\ref{alg:sample} samples a network from \eqref{eq:model}, and the likelihood is
\begin{equation}\label{eq:lik}
p(\mathcal{Y}\mid Z,\alpha)=\prod_{i<j}p_{ij}^{y_{ij}}(1-p_{ij})^{1-y_{ij}} .
\end{equation}
Directed and valued ties can be accommodated exactly as in the Euclidean case. Directed ties are taken up in Section~\ref{sec:directed}.

\begin{algorithm}[H]
\caption{Sampling a network from \eqref{eq:model}}\label{alg:sample}
\begin{algorithmic}
\STATE Sample $\delta_i\sim\mathcal{N}_3(0,\Sigma)$ and set $z_i=\mu\star\chi_\M(\delta_i)$, $i\in[N]$.
\STATE For $i<j$: compute $d_{ij}=d_\M(\eb,\,z_i^{-1}\star z_j)$ (Proposition~\ref{prop:nildist} or the distance table), $p_{ij}=\{1+\exp(d_{ij}-\alpha)\}^{-1}$, and draw $y_{ij}\sim\mathrm{Bernoulli}(p_{ij})$.
\end{algorithmic}
\end{algorithm}

\section{Non-identifiability of the latent coordinates}\label{sec:ident}

\subsection{Orbit non-identifiability and anchors}\label{sec:orbit}

The likelihood \eqref{eq:lik} depends on $Z$ only through the distance matrix $D(Z)=(d_\M(z_i,z_j))_{i<j}$. Hence for every $g\in\Isom(\M)$, acting diagonally by $g\cdot Z=(g z_1,\dots,g z_N)$,
\begin{equation}\label{eq:invariance}
p(\mathcal{Y}\mid g\cdot Z,\alpha)=p(\mathcal{Y}\mid Z,\alpha),
\end{equation}
and $Z$ is identifiable at best up to its orbit $\Isom(\M)\cdot Z$. In the constant-curvature case this ambiguity has dimension $\dim\Isom=d(d+1)/2$ and is removed by $d+1$ anchors \citep{papamichalis2021latent}. Here the orbits are smaller: $\dim\Isom(\M)=4$ for $\Nil$ and $\SL$ and $3$ for $\Sol$, against $6$ for the three-dimensional constant-curvature spaces. As in the earlier work on constant-curvature latent spaces we remove the ambiguity by a \emph{gauge}: we choose anchor indices, apply a data-dependent isometry that brings the anchors to a canonical position, and constrain the anchors to remain canonical during inference. No per-iteration Procrustes alignment is needed, and the inter-anchor distances remain free.

\subsection{Twisted bundles: two anchors}\label{sec:gauge-bundle}

Let $\M=\M_K$ and fix distinct anchor indices $I=(i_1,i_2)$. For $Z\in\M^N$ write $\tilde Z=T_{z_{i_1}}^{-1}\cdot Z$ (so $\tilde z_{i_1}=\eb$) and $\tilde z_{i_2}=(\tilde w,\tilde\zeta)$. Define the \emph{non-degenerate set} and the \emph{slice}
\begin{equation}\label{eq:slice-bundle}
\Omega_I=\{Z:\tilde w\ne0,\ \tilde\zeta\ne0\},\qquad
\mathcal{S}_I=\{Z\in\M^N:\ z_{i_1}=\eb,\ \ z_{i_2}=(\rho,0,\zeta)\ \text{with } \rho>0,\ \zeta>0\},
\end{equation}
where $(\rho,0,\zeta)$ denotes the point whose base projection lies on the positive real axis at geodesic distance $\rho$ from $\bo$ (in $\SL$ its disc coordinate is $\tanh(\rho/2)$) and whose fibre coordinate is $\zeta$. For $Z\in\Omega_I$ let $\varphi=\arg\tilde w$, $\epsilon=\sgn\tilde\zeta$, and $h_Z=\refl^{(1-\epsilon)/2}\circ R_{-\varphi}\in\Stab(\eb)$.

\begin{theorem}[Two-anchor gauge fixing in $\Nil$ and $\SL$]\label{thm:gauge-bundle}
Define $\Gamma_I(Z)=h_Z\cdot\big(T_{z_{i_1}}^{-1}\cdot Z\big)$ for $Z\in\Omega_I$. Then:
\begin{enumerate}[(i)]
\item (\emph{gauge conditions}) $\Gamma_I(Z)\in\mathcal{S}_I$, and $\Gamma_I(Z)=g_Z\cdot Z$ for the isometry $g_Z=h_Z\circ T_{z_{i_1}}^{-1}$;
\item (\emph{orbit section}) if $g\in\Isom(\M)$ and $g\cdot Z\in\mathcal{S}_I$, then $g=g_Z$. In particular each orbit through $\Omega_I$ meets $\mathcal{S}_I$ in exactly one point, and $\Gamma_I(g\cdot Z)=\Gamma_I(Z)$ for all $g$;
\item (\emph{regularity}) $\Omega_I$ is open with complement of measure zero, $\Gamma_I$ is continuous on $\Omega_I$, and $\mathcal{S}_I$ is a $(3N-4)$-dimensional submanifold of $\M^N$.
\end{enumerate}
\end{theorem}

The proof (Supplement~\ref{app:proofs}) reduces to Lemma~\ref{lem:stab}: an isometry fixing $\eb$ is a rotation of the base or a rotation composed with the reflection $\refl$. The rotation is fixed by the direction of $\tilde w$ and the reflection by the sign of $\tilde\zeta$. Thus the second anchor loses only one continuous degree of freedom (the angle of its base coordinate) and one sign. Its remaining coordinates $(\rho,\zeta)$, equivalently the distances from $\eb$ to $z_{i_2}$ and to the fibre through $z_{i_2}$, are learned from the data. For $d=2$ constant-curvature spaces the second anchor also loses one angle but a third anchor is needed to fix the reflection. Here the reflection is coupled to the fibre and is fixed by the second anchor's vertical coordinate.

\begin{remark}[Alternative sign rule]\label{rem:third}
When $|\tilde\zeta|$ is small the sign $\epsilon$ is fragile. One may instead fix the reflection by a third anchor $i_3$, requiring $\mathrm{Im}\,\tilde w_{i_3}>0$ after the rotation. The theorem holds verbatim with $\Omega_I$ replaced by $\{\tilde w\ne0,\ \mathrm{Im}(e^{-i\varphi}\tilde w_{i_3})\ne0\}$ and the slice by $\{z_{i_1}=\eb,\ z_{i_2}=(\rho,0,\zeta),\ \rho>0,\ \mathrm{Im}\,w_{i_3}>0\}$, which has the same dimension $3N-4$.
\end{remark}

\subsection{$\Sol$: one anchor and a discrete choice}\label{sec:gauge-sol}

Fix $I=(i_1,i_2)$, write $\tilde Z=z_{i_1}^{-1}\cdot Z$ and $\tilde z_{i_2}=(\tilde x,\tilde y,\tilde z)$, and let $F=(0,\infty)^3$ be the open positive orthant.

\begin{theorem}[Gauge fixing in $\Sol$]\label{thm:gauge-sol}
Let $\Omega_I=\{Z:\tilde x\tilde y\tilde z\ne0\}$ and $\mathcal{S}_I=\{Z: z_{i_1}=\eb,\ z_{i_2}\in F\}$. For $Z\in\Omega_I$ there is a unique $h_Z\in D_4$ with $h_Z\tilde z_{i_2}\in F$, namely $h_Z=\varsigma^{[\tilde z<0]}\circ\sigma_y^{[\tilde y<0]}\circ\sigma_x^{[\tilde x<0]}$, and $\Gamma_I(Z)=h_Z\cdot\tilde Z$ satisfies (i)--(iii) of Theorem~\ref{thm:gauge-bundle} with $\dim\mathcal{S}_I=3N-3$. In particular the latent coordinates in $\Sol$ are identifiable, from the orbit, up to a single translation and a finite group of order eight.
\end{theorem}

Only the first anchor is constrained by equalities. The second anchor keeps all three coordinates free within the orthant, and the gauge map $h_Z$ is locally constant on $\Omega_I$. This reflects the absence of continuous rotations in $\Sol$.

The two theorems have the same proof: an isometry that fixes the first anchor lies in $\Stab(\eb)$, and the remaining anchors are used to pick out one element of $\Stab(\eb)$. The same argument covers the other Thurston geometries.

\begin{theorem}[One gauge theorem for the eight model geometries]\label{cor:all}
Let $\M$ be a three-dimensional Thurston geometry with a base point $\eb$ and a (chartwise continuous) transitive family of isometries $\{T_p\}$ with $T_p(\eb)=p$. Then statements (i)--(iii) of Theorem~\ref{thm:gauge-bundle} hold with $\Stab(\eb)$, the non-degenerate set and the canonical set replaced as follows.
\begin{enumerate}[(a)]
\item Constant curvature ($\R^3$, $\mathbb{S}^3$, $\Hb^3$): $\Stab(\eb)=O(3)$. Four anchors, $z_{i_2}$ on the positive first axis, $z_{i_3}$ in the open half-plane $\{x_3=0,\ x_2>0\}$, and $z_{i_4}$ with $x_3>0$ (the last condition removes the reflection $x_3\mapsto-x_3$, which fixes the first three anchors), $\dim\mathcal{S}_I=3N-6$. This is the $d+1=4$ anchor gauge of \citet{papamichalis2021latent}. The Euclidean product $\R^2\times\R$ is $\R^3$ and belongs here.
\item Non-flat products $\mathbb{S}^2\times\R$, $\Hb^2\times\R$: $\Stab(\eb)=O(2)\times\{\pm1\}$. Three anchors, $z_{i_2}=(\rho,0,\zeta)$ with $\rho>0$, $\zeta>0$ (fixing the base rotation and the fibre flip) and $\mathrm{Im}\,w_{i_3}>0$ (fixing the base reflection), $\dim\mathcal{S}_I=3N-4$.
\item Twisted bundles ($\Nil$, $\SL$): $\Stab(\eb)=O(2)$. Two anchors (Theorem~\ref{thm:gauge-bundle}), $\dim\mathcal{S}_I=3N-4$.
\item $\Sol$: $\Stab(\eb)=D_4$. One anchor and the chamber of a second node (Theorem~\ref{thm:gauge-sol}), $\dim\mathcal{S}_I=3N-3$.
\end{enumerate}
\end{theorem}

The continuous dimension is the same in (b) and (c). What the twist removes, relative to $\Hb^2\times\R$, is one sign condition on a third node, because the base reflection and the fibre flip are independent symmetries of the product but a single symmetry of the bundle (Lemma~\ref{lem:stab}). Relative to $\R^3$ the bundle $\Nil$ needs two anchors instead of four. The proof is in Supplement~\ref{app:proofs}.

\subsection{Non-rigidity of finite configurations}\label{sec:rigidity}

In $\R^d$, $\mathbb{S}^d$ and $\Hb^d$ the distance matrix determines the configuration up to isometry, so \eqref{eq:invariance} is the \emph{only} source of non-identifiability of $Z$ given $\alpha$: these spaces are two-point homogeneous, and the classical congruence theorem states that two finite subsets with the same pairwise distances are related by a global isometry, the property on which the consistency theory of continuous-space network models rests \citep{shalizi2017consistency}. It fails in $\Nil$, $\Sol$ and $\SL$.

\begin{proposition}[Non-rigidity of small configurations]\label{prop:nonrigid}
In each $\M\in\{\Nil,\Sol,\SL\}$ there are pairs $\{p,q\}$ and $\{p,q'\}$ with $d_\M(p,q)=d_\M(p,q')$ that are not congruent, i.e.\ no $g\in\Isom(\M)$ maps $\{p,q\}$ onto $\{p,q'\}$. Explicitly, with $p=\eb$:
\begin{gather*}
\Nil:\ q=(1,0,0),\ q'=(0,0,1);\qquad \SL:\ q=(\tanh\tfrac12,0,0),\ q'=(0,0,1);\\
\Sol:\ q=(1,0,0),\ q'=(0,0,\cosh^{-1}\tfrac32).
\end{gather*}
Consequently, for $2\le N\le5$, $\dim\M^N-\dim\Isom(\M)>N(N-1)/2$, so the regular fibres of the distance map $D:\M^N\to\R^{N(N-1)/2}$ have positive dimension modulo the $\Isom(\M)$-orbits: generic configurations of at most five points are not determined by their distances even up to isometry (the statement is about regular points of $D$, nothing is claimed for singular configurations).
\end{proposition}

The proof (Supplement~\ref{app:proofs}) uses that every isometry maps the vertical field to $\pm V$, hence horizontal geodesic segments to horizontal ones, whereas the second pair is joined by a vertical segment (in $\Sol$, that every isometry preserves the $z$-axis direction up to sign). The dimension count is the following: the orbit space has dimension $3N-\dim\Isom(\M)$, which exceeds the number $N(N-1)/2$ of distances exactly for $N\le5$ in all three geometries. Thus the smallest network for which the gauge-fixed latent positions can possibly be identified from the distance matrix has $N_0=6$ nodes, for each of the three geometries. In constant curvature there is no such threshold, since distances determine configurations up to isometry for every $N$. Beyond $N_0$, identifiability is a \emph{generic} and \emph{local} property, and it depends on whether the intercept $\alpha$ is known.

\begin{theorem}[Generic local identifiability of the gauge-fixed positions]\label{thm:generic}
Let $\M\in\{\Nil,\Sol,\SL\}$ and let $U\subset\Omega_I$ be the open set of non-degenerate configurations in which no $z_j$ lies in the cut locus of another $z_i$ and all $z_i$ are distinct. On $U$ the distance map $D$ is real-analytic and $\rank dD_Z\le3N-\dim\Isom(\M)$.
\begin{enumerate}[(i)]
\item (\emph{fixed $\alpha$}) If $\rank dD_{Z_0}=3N-\dim\Isom(\M)$ at some $Z_0\in U$, then the set of $Z$ of maximal rank is open and dense in the connected component of $U$ containing $Z_0$, and at every such $Z$ the restriction of $D$ to the slice $\mathcal{S}_I$ is an immersion at $\Gamma_I(Z)$: the gauge-fixed positions are \emph{locally} identifiable from the distance matrix, and the Fisher information of $Z^\star$ at fixed $\alpha$ is non-singular whenever $0<p_{ij}<1$. This requires $\binom N2\ge3N-\dim\Isom(\M)$, i.e.\ $N\ge6$.
\item (\emph{estimated $\alpha$}) If the augmented matrix $[\mathbf 1,\ -dD|_{T\mathcal{S}_I}]$ has full column rank, the pair $(Z^\star,\alpha)$ is locally identifiable from the tie probabilities and its Fisher information is non-singular (regular identifiability in the sense of \citealp{rothenberg1971identification}). Full rank requires $\binom N2\ge3N-\dim\Isom(\M)+1$, i.e.\ $N\ge6$ for $\Nil$ and $\SL$ but $N\ge7$ for $\Sol$, and for $\Sol$ at $N=6$ the pair is not locally identifiable: the log-odds map sends the $16$-dimensional parameter $(Z^\star,\alpha)$ to $\R^{15}$ and is not injective on any open set (invariance of domain).
\item (\emph{induction}) If the (augmented) rank is maximal at a configuration with $N$ nodes, it is maximal at an open dense set of positions of an added node: adding a node whose gradient vectors $\nabla_{z_{N+1}}d(z_{N+1},z_j)$ towards three of the existing nodes span $T_{z_{N+1}}\M$ adds three to the rank, and such positions exist near any node whose incident gradients span the tangent space, which at a maximal-rank configuration some node must do.
\end{enumerate}
\end{theorem}

The theorem is conditional on a full-rank witness, its genericity statement is confined to the connected component of $U$ containing the witness (in $\Nil$ and $\SL$ the cut locus of a point lies in its fibre, a set of codimension two, and coincidences have codimension three, so $U$ has the two components of $\Omega_I$ distinguished by the sign of $\tilde\zeta$, which are exchanged by $\refl$ and have the same rank; in $\Sol$ the cut locus is not known in closed form and the statement is confined to the component of the witness), and its conclusion is local: an immersion excludes nearby non-congruent alternatives, not distant ones, and Supplement~\ref{supp:exp} finds distant alternatives at $N_0$ and $N_0+1$ in $\SL$. Witnesses are supplied numerically (Table~\ref{tab:rigidity}): in $\Nil$, where the differential is computed from the exact distance, both $dD|_{T\mathcal{S}_I}$ and $[\mathbf 1,-dD|_{T\mathcal{S}_I}]$ have full rank at every one of $100$ random configurations for $N=6,7,8,10$, although the smallest singular value at $N=6$ is tiny (median \nilSixMedian, minimum \nilSixMin), identification is weak where it first appears. In $\Sol$ and $\SL$ the differential is computed from the first-variation formula with minimising geodesics from a boundary-value solver, and the ranks are as predicted, including the failure of joint identification in $\Sol$ at $N=6$. For $\Nil$ the witness at $N=6$ is also certified rigorously: computing the slice Jacobian in interval arithmetic (the geodesic parameter of Proposition~\ref{prop:nildist} enclosed by interval bisection of the monotone equation, derivatives by implicit differentiation, all operations rounded outward, enclosure widths \certWidth) and factorising $J^{\tr}J-\lambda I$ by an interval Cholesky decomposition certifies $s_{\min}(J)\ge\certSixLower$ and $s_{\min}([\mathbf 1,-J])\ge\certSixAug$ at a configuration with $N=6$ (floating-point value \certSixFloat), and $\ge\certSevenLower$, $\certSevenAug$ at one with $N=7$ (Supplement~\ref{app:details}). With Theorem~\ref{thm:generic}(iii) this proves regular local identifiability of $(Z^\star,\alpha)$ in $\Nil$ at generic configurations for every $N\ge6$ (on the component of the witness, hence on both components of $\Omega_I$). For $\Sol$ and $\SL$, whose distances are known only numerically, the witnesses remain floating-point computations at chosen configurations. Proposition~\ref{prop:nonrigid} and Theorem~\ref{thm:generic} together say that anisotropy costs exactly one thing relative to constant curvature: the congruence theorem for finite metric subsets fails, so identification from distances becomes a generic, local statement with a threshold in $N$.

\subsection{The gauge-fixed posterior}\label{sec:target}

Gauge fixing changes the parameter from $Z$ to $Z^\star=\Gamma_I(Z)\in\mathcal{S}_I$ and the reference measure from $\dvol^{\otimes N}$ to a measure on the slice. The correct density on the slice carries a Jacobian for the constrained anchor. Let the prior be $z_i\overset{\text{iid}}{\sim}\LWN(\mu,\Sigma)$ with $\Sigma=\diag(\sigma_h^2,\sigma_h^2,\sigma_v^2)$, a flat prior on $\mu$ with respect to $\dvol$, and proper priors $p(\sigma_h,\sigma_v)$, $p(\alpha)$. The group $\Isom(\M)$ acts on $(Z,\mu)$ diagonally and, by Lemma~\ref{lem:equivariance}, the joint prior and the likelihood are invariant. On the slice we use the coordinates $(\rho,\zeta)$ of the second anchor, and the remaining nodes and $\mu$ keep their volume measure.

\begin{proposition}[Quotient posterior on the slice]\label{prop:target}
Let $N\ge2$ and let the priors on $(\sigma_h,\sigma_v)$ and $\alpha$ be proper. The joint prior on $(Z,\mu)$ is invariant under the diagonal action of $\Isom(\M)$ and the group is non-compact, so the joint posterior is an invariant, infinite measure on $\M^N\times\M$. The object of inference is its disintegration over the orbits, the \emph{quotient posterior}. It is a probability measure on the slice $\mathcal{S}_I\times\M\times(0,\infty)^2\times\R$ with density
\begin{equation}\label{eq:target}
\pi^\star(Z^\star,\mu,\sigma_h,\sigma_v,\alpha)\ \propto\ p(\mathcal{Y}\mid Z^\star,\alpha)\prod_{i\in[N]}f_{\LWN}(z^\star_i\mid\mu,\Sigma)\ p(\sigma_h,\sigma_v)\,p(\alpha)\ \Delta_\M(z^\star_{i_2}),
\end{equation}
with respect to $\prod_{i\ne i_1,i_2}\dvol(z_i^\star)\otimes d\rho\,d\zeta\otimes\dvol(\mu)\otimes d\sigma_h\,d\sigma_v\,d\alpha$ (in $\Sol$: $\dvol(z^\star_{i_2})$ restricted to $F$), where $\rho$ is the geodesic base radius of the second anchor (with respect to the disc coordinate $u=\tanh(\rho/2)$ of $\SL$ the factor would be $4u/(1-u^2)^2$), and
\begin{equation}\label{eq:slicejac}
\Delta_{\Nil}(\rho,0,\zeta)=\rho,\qquad \Delta_{\SL}(\rho,0,\zeta)=\sinh\rho,\qquad \Delta_{\Sol}\equiv1 ,
\end{equation}
\end{proposition}

The factor is the change-of-variables determinant of the polar decomposition $(\varphi,\rho,\zeta)\mapsto R_\varphi(\rho,0,\zeta)$ of the second anchor's coordinates. It is not the length of the stabiliser orbit (that orbit, $\{R_\varphi(\rho,0,\zeta)\}$, has tangent norm $\sqrt{S_K(\rho)^2+F_K(\rho)^2}$ because the orbit is not orthogonal to the slice). In $\Sol$ the stabiliser is finite and no factor arises. We write $\Delta_\M$ for this slice Jacobian, distinct from the chart Jacobian $J^\chi_\M$ of Section~\ref{sec:background}. The proof is in Supplement~\ref{app:proofs}. Both estimation schemes of Section~\ref{sec:estimation} target \eqref{eq:target}. The same argument shows that in the constant-curvature setting the second and third anchors of \citet{papamichalis2021latent} contribute factors $S_K(\rho_{i_2})^{d-1}$ and $S_K(\cdot)^{d-2}$ (Euclidean: $\rho_{i_2}^{\,d-1}$ times the distance of the third anchor from the axis). Supplement~\ref{supp:exp} shows that the factor moves the posterior of the anchors' base distance appreciably at $N=8$ and still at $N=40$ (that distance is itself isometry-invariant). The factor is the Jacobian of Bookstein-type anchor coordinates \citep{bookstein1986size,dryden2016statistical}, here for a curved anisotropic ambient space.

\subsection{Representation deficit and detectability}\label{sec:deficit}

The twist of $\Nil$ and $\Sol$ is a temperature (Proposition~\ref{prop:twistscale}, Supplement~\ref{supp:twist}): the parameters $(\beta,\tau,Z)$ and $(\beta/\tau,1,\phi_\tau Z)$ give the same law, so whether a latent space is twisted is a comparison between fitted models at their own scales, and the vertical offsets $v_{ij}$ satisfy the holonomy identity $v_{ij}+v_{jk}+v_{ki}=-\mathrm{Area}(w_i,w_j,w_k)$ around every triangle (Lemma~\ref{lem:holonomy}, Supplement~\ref{supp:twist}). Detection is therefore a model-comparison question, and it is decided by the divergence between the two models.

If the network is generated in $\M$ and fitted in a competitor geometry $T$ (a product, or a constant-curvature space), how much does the competitor lose, and can the loss be detected from one network? The first quantity is a functional of the true configuration and of the base rate.

\begin{proposition}[Fisher deficit]\label{prop:deficit}
Let $Z\in\M^N$, $\alpha\in\R$, $p_{ij}=\sigma(\alpha-d_\M(z_i,z_j))$ and $w_{ij}=p_{ij}(1-p_{ij})$. For a competitor geometry $T$ define
\begin{equation}\label{eq:deficit}
\Delta_T^2(Z,\alpha)=\inf_{Z'\in T^N,\ c\in\R}\ \tfrac12\sum_{i<j}w_{ij}\big(d_\M(z_i,z_j)-d_T(z'_i,z'_j)-c\big)^2 .
\end{equation}
Then for every $(Z',\alpha')$ the dyad Kullback--Leibler divergence $\sum_{i<j}\mathrm{KL}\big(\mathrm{Bern}(p_{ij})\,\|\,\mathrm{Bern}(\sigma(\alpha'-d_T(z'_i,z'_j)))\big)$ equals $\tfrac12\sum_{i<j}w_{ij}(\eta_{ij}-\eta'_{ij})^2+O(\sum|\eta_{ij}-\eta'_{ij}|^3)$ with $\eta_{ij}-\eta'_{ij}=-(d_\M(z_i,z_j)-d_T(z'_i,z'_j)-c)$, $c=\alpha-\alpha'$, and remainder bounded by $\sum|\eta_{ij}-\eta'_{ij}|^3/(36\sqrt3)$. Hence, for competitors whose log-odds stay close to the truth, $\Delta_T^2(Z,\alpha)$ approximates the smallest divergence between the true dyad law and the competitor's family, and the expected log-likelihood ratio of the true model against the best competitor. If the infimum is attained with $\Delta_T=0$, the distance matrix of $Z$ is realised in $T$ up to an additive constant. For $T=\R^3$ the set of such matrices has dimension at most $3N-5$, which is smaller than $\binom N2$ from $N=6$ on.
\end{proposition}

The proof is the Taylor expansion of the Bernoulli divergence in the log-odds (Supplement~\ref{app:proofs}). The quadratic is a local surrogate: whether it tracks the exact infimum is a numerical question, answered affirmatively in Section~\ref{sec:exp-twist} for our configurations by minimising the exact divergence directly. The dimension count for the last statement is $3N-6+1\ge\binom N2$ iff $N\le5$, the threshold of Proposition~\ref{prop:nonrigid}. That the distance matrices of generic anisotropic configurations lie outside the Euclidean set from $N=6$ follows from Theorem~\ref{thm:generic} at a witness (the image of $D$ has dimension $3N-4>3N-5$ there). The exact divergence gives a rigorous limit on detection.

\begin{theorem}[Undetectability from one network]\label{thm:lecam}
Let $P$ be the law of the network under $(Z,\alpha)$ in $\M$ and $Q$ its law under any $(Z',\alpha')$ in a competitor geometry $T$. Any test of ``the geometry is $T$'' against ``the geometry is $\M$'' that uses one network and has level at most $\nu$ under every law of the competitor family ($\sup_{Q}\mathbb{E}_Q\phi\le\nu$) has power at most $\nu+\mathrm{TV}(P,Q)\le\nu+\sqrt{\tfrac12\sum_{i<j}\mathrm{KL}\big(\mathrm{Bern}(p_{ij})\,\|\,\mathrm{Bern}(q_{ij})\big)}$ at $(Z,\alpha)$, for every $(Z',\alpha')$. In particular at most $\nu+\sqrt{\tfrac12\inf_{Z',\alpha'}\sum_{i<j}\mathrm{KL}_{ij}}$, whose second-order approximation is $\nu+\Delta_T(Z,\alpha)/\sqrt2$.
\end{theorem}

This is Le Cam's two-point argument with Pinsker's inequality, applied dyad by dyad (the dyads are independent given the positions). The bound is computable: any competitor configuration gives a valid upper bound on the power, and Section~\ref{sec:exp-twist} evaluates it at the minimiser of the exact divergence. The proposition turns ``does the twist matter'' into a computation: the smallest divergence divided by the number of dyads is, to first order, the expected held-out log-loss advantage of the true model over the best competitor with known parameters. Whether one network can reveal it depends on the sampling variability of the held-out difference (both fits pay estimation cost, which cancels in a paired comparison) and, from below, on Theorem~\ref{thm:lecam}. Section~\ref{sec:exp-twist} evaluates both. The functional is computed by weighted least squares over $T^N\times\R$, and the exact divergence by least squares on Bernoulli deviance residuals.

\subsection{Small-scale deficit}\label{sec:smallscale}

When the configuration is small relative to the curvature scale, the smallest divergence has an explicit leading term: the twist enters through the curvature of $\M$, and the part of the curvature that a competitor cannot absorb is its component along the equilibrium stresses of the configuration viewed as a bar framework in $\R^3$.

Let $o\in\M$, let $u_i=\exp_o^{-1}(z_i)\in T_o\M\cong\R^3$ be the normal coordinates of the configuration, let $X=(u_1,\dots,u_N)$ be the same points read as a configuration in $\R^3$, and let $\mathrm{Rm}_\M(u,v,v,u)$ be the curvature quadratic form of $\M$ at $o$, so that $\mathrm{Rm}_\M(u,v,v,u)=K(\sigma)|u\wedge v|^2$ with $K(\sigma)$ the sectional curvature of the plane $\sigma$ spanned by $u$ and $v$. With $U=u\times v$ and $U_3$ its component along the fibre direction at $o$, these forms are
\begin{equation}\label{eq:curvforms}
\begin{aligned}
&\mathrm{Rm}_{\Nil}=\tfrac14|U|^2-U_3^2,\qquad \mathrm{Rm}_{\SL}=\tfrac14|U|^2-2U_3^2,\qquad \mathrm{Rm}_{\Sol}=2U_3^2-|U|^2,\\
&\mathrm{Rm}_{\Hb^2\times\R}=-U_3^2,\qquad \mathrm{Rm}_{\Hb^3}=-|U|^2,\qquad \mathrm{Rm}_{\R^3}=0,
\end{aligned}
\end{equation}
where $U_3=u_1v_2-u_2v_1$ is twice the signed area of the triangle with vertices $o$ and the base projections of $u$ and $v$. The complete bar framework on $X$ has rigidity matrix $J_E(X)$, the differential of the Euclidean distance map, and its space of equilibrium stresses orthogonal to the constant vector is $\Omega(X)=\{\omega\in\R^{\binom N2}:J_E(X)^{\tr}\omega=0,\ \mathbf 1^{\tr}\omega=0\}$, of dimension $\binom N2-\operatorname{rank}[J_E(X),\mathbf{1}]$, equal to $\binom N2-3N+5$ when $X$ is infinitesimally rigid in $\R^3$ and $N\ge5$.

\begin{theorem}[Small-scale deficit, local]\label{thm:smallscale}
Let $\M\in\{\Nil,\Sol,\SL\}$ and let $T$ be $\R^3$, $\Hb^3$ or $\Hb^2\times\R$. Let $Z=(\exp_o(\varepsilon u_i))_{i\le N}$ for fixed $u_1,\dots,u_N$ whose framework $X$ is infinitesimally rigid in $\R^3$, and let $w_0=\sigma(\alpha)(1-\sigma(\alpha))$ be the common limit of the weights $w_{ij}$ as $\varepsilon\to0$. Define
\begin{equation}\label{eq:rdef}
r_{ij}=-\frac{[\mathrm{Rm}_\M-\mathrm{Rm}_T](u_i,u_j,u_j,u_i)}{6\,|u_i-u_j|},
\end{equation}
where for $T=\Hb^2\times\R$ the form $\mathrm{Rm}_T$ is evaluated with the fibre direction of $T$ rotated to a unit vector $n$, and let $\Pi$ be the orthogonal projection of $\R^{\binom N2}$ onto $\Omega(X)$. Then, as $\varepsilon\to0$,
\begin{equation}\label{eq:smallscale}\small
\inf_{Z'\in T^N,\ \alpha'}\ \sum_{i<j}\mathrm{KL}_{ij}\;=\;\Delta_T^2(Z,\alpha)\,(1+o(1)),\qquad \Delta_T^2(Z,\alpha)\;=\;\frac{w_0\,\varepsilon^6}{2}\,\min_n\big\|\Pi r\big\|^2+O(\varepsilon^7),
\end{equation}
with the minimum over $n$ omitted for $T=\R^3$ and $\Hb^3$. The infimum is over competitor configurations and intercepts whose distance matrices lie in a neighbourhood of $D_E(X)$ up to an additive constant (the proof is local in distance space; a global statement would need to exclude distant realisations, which we do not claim), and the coefficient of $\varepsilon^6$ may vanish, in which case the deficit is $O(\varepsilon^7)$. When $\Omega(X)=\{0\}$, which is the case for $N\le5$, the second statement reads $\Delta_T^2(Z,\alpha)=O(\varepsilon^7)$.
\end{theorem}

The theorem has three consequences. The smallest divergence to any product or constant-curvature competitor vanishes like the sixth power of the scale, so the anisotropy of a small configuration is invisible: the first-order effect of the twist on distances is a change of coordinates that a competitor absorbs, and what remains is curvature. The leading term is the curvature difference between $\M$ and $T$, read along each pair, and projected on the stresses of the configuration: pairs of points whose curvature term can be reproduced by moving the competitor's points contribute nothing, which is why no configuration with $N\le5$ points, whose stress space is empty, has a deficit at this order, in agreement with Proposition~\ref{prop:nonrigid}. Against $\Hb^2\times\R$ with its fibre aligned to that of $\Nil$, the curvature difference $\tfrac14|U|^2$ is isotropic: the squared areas cancel, so at small scale $\Nil$ is a product with one quarter of a unit of extra constant curvature, and the twist is not what distinguishes it from that product. Against $\R^3$ and $\Hb^3$ the anisotropic term $U_3^2$, the squared enclosed area, survives, and against every competitor the deficit of $\Sol$ is driven by $U_3^2$ with the opposite sign. Table~\ref{tab:smallscale} compares the leading term of \eqref{eq:smallscale}, computed with the exact weights $w_{ij}$ in place of $w_0$ (a difference of relative order $\varepsilon$), with the exact minimum divergence for one configuration at four scales: for $\R^3$ and $\Hb^3$ the ratio tends to one with an error linear in $\varepsilon$, as the theorem predicts.

The proof (Supplement~\ref{app:proofs}) combines the expansion $d_\M(\exp_o u,\exp_o v)^2=|u-v|^2-\tfrac13\mathrm{Rm}(u,v,v,u)+O(|u|^5+|v|^5)$ of the distance in normal coordinates (Lemma~\ref{lem:normal}, proved there), the analogous expansion for $T$, the fact that the Euclidean-plus-constant distance vectors form a submanifold whose normal space at $D_E(X)$ is $\Omega(X)$, and the argument of the proof of Proposition~\ref{prop:deficit} that competitors far from the truth cost a divergence bounded below by a constant independent of $\varepsilon$. In shape analysis the corresponding expansion has no curvature term because the ambient space is flat. Here the curvature difference is the perturbation, and its component in the tangent space of the competitor family is the part a competitor absorbs.

\subsection{Directed ties: the twist at first order}\label{sec:directed}

The small-scale theorem explains why undirected ties hide the twist: its first-order effect on distances is a coordinate change. Directed ties expose it. Let the network be directed, $Y_{ij}\sim\mathrm{Bernoulli}(p_{ij})$ independently for ordered pairs $i\ne j$, with
\begin{equation}\label{eq:directed}
\logit(p_{ij})=\alpha-d_\M(z_i,z_j)+\gamma\,v_{ij},
\end{equation}
where $v_{ij}$ is the vertical offset of Lemma~\ref{lem:holonomy}, antisymmetric in $(i,j)$, and $\gamma\in\R$. The term $\gamma v_{ij}$ makes ties from $i$ to $j$ more likely than from $j$ to $i$ when $z_j$ lies above the horizontal lift through $z_i$. In a product or a constant-curvature space the analogous model has $v_{ij}=\zeta_j-\zeta_i$ for a coordinate $\zeta$: the asymmetry is a potential difference, and the model is a latent distance model with a ranking term, the form taken by the additive latent distance models for directed ties (a ranking plus a symmetric distance); the bilinear and additive-and-multiplicative-effects models of \citet{hoff2005bilinear,hoff2021additive} and the model of \citet{sewell2015latent} are outside this class. Call this the gradient class: a symmetric distance term plus a coboundary asymmetry $u_{ij}=\phi_j-\phi_i$.

\begin{proposition}[Circulation and its deficit]\label{prop:directed}
Let $s_{ij}=\logit(p_{ij})-\logit(p_{ji})$. Under \eqref{eq:directed} in $\M_K$, $s_{ij}=2\gamma v_{ij}$ and for every triple
\[
s_{ij}+s_{jk}+s_{ki}=-2\gamma\,\Area_K(w_i,w_j,w_k),
\]
whereas every model of the gradient class has $s_{ij}+s_{jk}+s_{ki}=0$. Let $\bar A_{ij}=N^{-1}\sum_{k}\Area_K(w_i,w_j,w_k)$ and $\tilde w_{ij}=\min(w_{ij},w_{ji})$. The Fisher deficit of \eqref{eq:directed} against the whole gradient class, defined as in \eqref{eq:deficit} with the sum over ordered pairs and the infimum over the class, satisfies
\begin{equation}\label{eq:directed-deficit}
\Delta^2_{\mathrm{grad}}(Z,\alpha,\gamma)\ \ge\ \gamma^2\min_{i<j}\tilde w_{ij}\sum_{i<j}\bar A_{ij}^2 .
\end{equation}
\end{proposition}

The first statement is Lemma~\ref{lem:holonomy}: the asymmetries of the twisted model circulate, by an amount equal to the enclosed area, and those of a gradient model do not. The second turns the circulation into a divergence: the antisymmetric part of the log-odds residual is $\gamma v_{ij}-\gamma' u_{ij}$ with $u$ a coboundary, the coboundaries of the complete graph are the orthogonal complement of the cycle space (the combinatorial Hodge decomposition of \citealp{jiang2011statistical}), and the component of $v$ in the cycle space has entries $-\bar A_{ij}$, so the residual cannot be smaller than the projection of $\gamma v$ on the cycle space (Supplement~\ref{app:proofs}). The bound holds for every symmetric distance, at every scale and every $N$, and grows with the squared enclosed areas, so it is of order $N^2\gamma^2\sigma_h^4$ rather than $\varepsilon^6$: in directed networks the twist is a first-order, detectable quantity. Two consequences follow. The gauge group of \eqref{eq:directed} with $\gamma$ of known sign is the identity component of $\Isom(\M)$, because the reflection $\refl$ reverses the sign of every $v_{ij}$; with $\gamma$ free, $(Z,\gamma)\mapsto(\refl Z,-\gamma)$ remains a symmetry and only $|\gamma|$ is identified. And the regime in which these geometries should be fitted is now explicit: directed networks whose asymmetries circulate, that is, networks in which preferences are not transitive. Section~\ref{sec:exp-directed} shows the corresponding win condition.

\begin{remark}[Identifiability of $\alpha$]
Given $Z$, $\alpha$ is identified by \eqref{eq:model}. Jointly, $(\alpha,Z)$ and $(\alpha+c,Z')$ with $d(z'_i,z'_j)=d(z_i,z_j)+c$ would be indistinguishable. For a finite configuration such a $Z'$ may exist for small $|c|$ (the additive-constant problem of multidimensional scaling), but under a latent distribution with full support the population version cannot hold for $c\ne0$: arbitrarily close latent pairs force $c\ge0$, and for $c>0$ infinitely many points pairwise at distance $\ge c$ cannot lie in a bounded region (compactness). At finite $N$ the intercept is regularly identified with the positions when the augmented rank condition of Theorem~\ref{thm:generic}(ii) holds.
\end{remark}

\section{Estimation}\label{sec:estimation}

Both estimation schemes target the quotient posterior \eqref{eq:target}. The first is a Metropolis-within-Gibbs sampler on the gauge slice: random-walk proposals for each position in the chart coordinates (proposals that leave the slice are rejected), a Gaussian random walk for $\alpha$ and for the centre with the chart Jacobian in the acceptance ratio, and Gibbs steps for the scales under inverse-Gamma priors. The second-anchor update carries the slice factor $\Delta_\M$ of \eqref{eq:slicejac}. The second is black-box variational inference \citep{ranganath2014black} with a mean-field family of left-translated wrapped Normals for the positions, a Gaussian factor on the log-transformed slice coordinates of the second anchor, and Gaussian factors for $\alpha$, the log-scales and the centre, optimised by score-function gradients with Rao--Blackwellisation, control variates and optional dyad subsampling. Positions are initialised by multidimensional scaling of shortest-path distances. For the directed model of Section~\ref{sec:directed} both schemes must evaluate the directed likelihood, in which a node's move changes its row and its column; Section~\ref{sec:exp-counter} reports the sampler on the network where the model wins. The algorithms, the priors, the score formulas and the settings are in Supplement~\ref{supp:estimation}, and Section~\ref{sec:exp-sims} compares the two schemes on replicated simulations.

\section{Experiments}\label{sec:experiments}

The experiments follow the theory. Section~\ref{sec:exp-ident} reports the rank witnesses for the identifiability results (the distance-inversion, non-congruence and slice-Jacobian experiments are in Supplement~\ref{supp:exp}), Section~\ref{sec:exp-cal} is a replicated calibration study, Section~\ref{sec:exp-sims} compares the two estimation schemes on simulated networks, Section~\ref{sec:exp-twist} evaluates the representation deficit and the twist, Section~\ref{sec:exp-real} compares the three geometries with the other model geometries and with calibrated non-geometric baselines on four networks. All computations use the Python implementation accompanying the paper on one CPU core. Every number below is generated from the stored results. Unless stated otherwise the scales are estimated under the $\mathrm{InvGamma}(2,2)$ prior and the centre carries its own posterior (MCMC) or variational factor (BBVI). All fitted distances are tabulated surrogates (Supplement~\ref{app:geodesics}, the $\Nil$ table has maximum error $0.018$ and mean $0.001$ against Proposition~\ref{prop:nildist}, which is used for the certificates and validations), and their error should be kept in mind wherever differences of that size are reported.

\subsection{Identifiability}\label{sec:exp-ident}

\paragraph{Rank witnesses.} Table~\ref{tab:rigidity} gives the numerical rank witnesses for Theorem~\ref{thm:generic}, including the augmented matrix that governs joint identification with $\alpha$: full rank at all $100$ random $\Nil$ configurations for $N=6,7,8,10$, with a smallest singular value that is tiny at $N=6$ (median \nilSixMedian) and grows with $N$ (median \nilTenMedian\ at $N=10$). Full rank in $\SL$ at $N=6,7$. And in $\Sol$ full slice rank at $N=6$ (where the augmented matrix is $15\times16$ and structurally rank-deficient), full augmented rank at $N=7$.

\begin{table}[H]\centering\scriptsize
\caption{Rank certificates for Theorem~\ref{thm:generic} and for joint identification with $\alpha$. $J$ is the differential of the distance map restricted to the tangent space of the gauge slice ($k=\dim\mathcal{S}_I$ columns, $\binom N2$ rows) and $[\mathbf 1,-J]$ its augmentation by the intercept direction. For $\Nil$ the differential is computed by central differences of the exact distance at $100$ random configurations, and the columns give the median [min, max] of the smallest singular value of $J$, the median [min] for $[\mathbf 1,-J]$, and the percentage of configurations at which $J$, resp.\ $[\mathbf 1,-J]$, has full column rank (tolerance $10^{-8}$). For $\Sol$ and $\SL$ the differential comes from the first-variation formula with minimising geodesics from a boundary-value solver at one configuration, and the last column gives the numerical ranks. For $\Sol$ at $N=6$ the augmented matrix has $16$ columns and $15$ rows, so joint identification with $\alpha$ is impossible.}\label{tab:rigidity}
\resizebox{\textwidth}{!}{\begin{tabular}{lcccccccl}\toprule
$\M$ & $N$ & dyads & $k$ & $k{+}1$ & $s_{\min}(J)$ & $s_{\min}([\mathbf 1,-J])$ & full rank ($J$ / aug.) & method\\ \midrule
$\Nil$ & 6 & 15 & 14 & 15 & 0.007 [4e-06, 0.047] & 0.004 [3e-06] & 100\% / 100\% & exact, 100 configurations \\
$\Nil$ & 7 & 21 & 17 & 18 & 0.016 [2e-03, 0.052] & 0.015 [2e-03] & 100\% / 100\% & exact, 100 configurations \\
$\Nil$ & 8 & 28 & 20 & 21 & 0.024 [3e-03, 0.076] & 0.023 [2e-03] & 100\% / 100\% & exact, 100 configurations \\
$\Nil$ & 10 & 45 & 26 & 27 & 0.038 [1e-02, 0.146] & 0.037 [1e-02] & 100\% / 100\% & exact, 100 configurations \\
$\Sol$ & 6 & 15 & 15 & 16 & 0.0089 & 0 (structural, $15\times16$) & 15 / 15 & BVP geodesics, 1 configuration \\
$\Sol$ & 7 & 21 & 18 & 19 & 0.0095 & 0.0034 & 18 / 19 & BVP geodesics, 1 configuration \\
$\SL$ & 6 & 15 & 14 & 15 & 0.0063 & 0.0007 & 14 / 15 & BVP geodesics, 1 configuration \\
$\SL$ & 7 & 21 & 17 & 18 & 0.0128 & 0.0110 & 17 / 18 & BVP geodesics, 1 configuration \\
\bottomrule\end{tabular}}\end{table}

\subsection{Calibration over independent networks}\label{sec:exp-cal}

Supplement Tables~\ref{tab:covrep} and~\ref{tab:covN} report a replicated coverage study: independent latent configurations and networks ($N=30$), the sampler run on each, and the frequency with which the true pairwise distances, tie probabilities and $\alpha$ fall inside their $90\%$ (and $50\%$) posterior intervals, with Monte Carlo standard errors over networks. With the scales fixed at their true values the intervals are nearly calibrated in all five geometries, $87\%$ at nominal $90\%$ in each of $\Nil$, $\Sol$, $\SL$, $\R^3$ and $\Hb^3$, $46$--$49\%$ at nominal $50\%$, $\alpha$ covered $83$--$100\%$ of the time, no bias. With the scales estimated under the $\mathrm{InvGamma}(2,2)$ prior (mean $2$, against true $\sigma^2$ of $2$--$4$) the same intervals cover \covNilig\ ($\Nil$), \covSolig\ ($\Sol$), \covSLtig\ ($\SL$), \covEucig\ ($\R^3$) and \covHthreeig\ ($\Hb^3$), $\alpha$ is covered $30$--$67\%$ of the time, and the posterior mean distances are biased downward by $0.7$--$1.3$. With an $\mathrm{InvGamma}(3,8)$ prior, whose mean equals the true horizontal variance, coverage in $\Nil$ rises to \covNilMatched, $\alpha$ is covered \alphaCovMatched\ of the time and the bias falls to \biasMatched: about two thirds of the shortfall is prior mismatch and one third is the weak identification of the scale at $N=30$, which shrinks the configuration and moves $\alpha$ with it. Holding the scales at multidimensional-scaling plug-in values is worse still (a single-network check gave $44$--$53\%$). The phenomenon is the same in $\R^3$ and $\Hb^3$, so it is a property of latent space models at this size, not of the anisotropic geometries. At $N=60$ the shortfall disappears: with the $\mathrm{InvGamma}(3,8)$ prior the intervals cover \covSixtyMatched\ (Supplement Table~\ref{tab:covN}), $\alpha$ is covered \alphaCovSixty\ of the time and the bias is \biasSixty, indistinguishable from the true-scale runs. The under-coverage is therefore a thirty-node phenomenon of a weakly identified scale under a mismatched prior, not a general property of the models. The counts of networks per cell ($6$--$12$) are small and unequal (compute budget), so differences of ten percentage points between geometries are within Monte Carlo error. The contrasts between known and estimated scales at $N=30$, and between $N=30$ and $N=60$, are not. The chains behind these numbers mix slowly: on the $N=40$ $\Nil$ network of Section~\ref{sec:exp-sims}, four dispersed chains of $6000$ iterations give $\widehat R=\rhatAlpha$ for $\alpha$, $\rhatSig$ for $\sigma_h$ and at most $\rhatDmax$ for pairwise distances, with effective sample sizes of \essAlpha\ ($\alpha$) and at least \essDmin\ (distances) out of $16\,000$ draws. The most dispersed chain settles at a lower log-likelihood, consistent with the multimodality of Section~\ref{sec:exp-ident}. Longer runs and a mode-jumping move are needed before finer coverage statements can be made.

\subsection{Simulated networks: MCMC against variational inference}\label{sec:exp-sims}

For each geometry we drew ten independent sets of $N=40$ positions with $\alpha=2.5$, sampled one network for fitting and a second from the same positions (conditional edge replication), and ran Algorithms~\ref{alg:mcmc} and~\ref{alg:bbvi} on each, Table~\ref{tab:simrep} gives means with standard errors over the ten networks, and Supplement~\ref{supp:sims} and Figure~\ref{fig:sims} show one realisation in more detail. The schemes recover the distance structure comparably (correlation of the estimated with the true distances $0.72$--$0.76$ for MCMC and $0.68$--$0.76$ for BBVI, $0.91$--$0.93$ between the two), their replicate log-losses are within $0.007$ of each other and $0.06$--$0.07$ above the oracle, and both under-estimate $\alpha$, BBVI by about $0.6$ more than MCMC (in the single realisation the variational standard deviation of $\alpha$ is three to five times smaller than the posterior one). BBVI is ten times faster here and its cost per iteration is linear in the number of subsampled dyads.

\begin{table}[H]\centering\scriptsize
\caption{Replicated simulations: independent (positions, network) pairs per geometry ($N=40$, $\alpha=2.5$; scales $(2,2)$, $(1.8,1.4)$, $(1.5,2)$ for $\Nil$, $\Sol$, $\SL$), Algorithm~\ref{alg:mcmc} ($4000$ iterations) and Algorithm~\ref{alg:bbvi} ($600$ iterations, $n_s=8$) on each. Replicate log-loss on a second network from the same positions (oracle: true probabilities); correlation of posterior mean (MCMC) or variational mean (BBVI) distances with the truth and with each other; posterior/variational mean of $\alpha$; mean run time in seconds. Means with standard errors over networks.}\label{tab:simrep}
\resizebox{\textwidth}{!}{\begin{tabular}{lcccccccccc}\toprule
$\M$ & networks & oracle & MCMC log-loss & BBVI log-loss & $\rho$(MCMC, truth) & $\rho$(BBVI, truth) & $\rho$(MCMC, BBVI) & $\hat\alpha$ MCMC & $\hat\alpha$ BBVI & time MCMC / BBVI\\ \midrule
$\Nil$ & 10 & 0.375 (0.009) & 0.438 (0.008) & 0.444 (0.009) & 0.76 (0.02) & 0.76 (0.01) & 0.92 (0.01) & 1.96 (0.16) & 1.40 (0.12) & 13 / 1 \\
$\Sol$ & 10 & 0.494 (0.010) & 0.559 (0.008) & 0.563 (0.007) & 0.72 (0.01) & 0.68 (0.01) & 0.91 (0.01) & 2.15 (0.20) & 1.58 (0.10) & 38 / 4 \\
$\SL$ & 10 & 0.436 (0.011) & 0.504 (0.012) & 0.507 (0.012) & 0.73 (0.01) & 0.72 (0.01) & 0.93 (0.00) & 2.11 (0.10) & 1.44 (0.07) & 21 / 2 \\
\bottomrule\end{tabular}}\end{table}

\subsection{Representation deficit, win maps and the twist}\label{sec:exp-twist}

\paragraph{Deficit and exact divergence.} Table~\ref{tab:deficit} evaluates the Fisher deficit of Proposition~\ref{prop:deficit} for configurations generated in $\Nil$, $\Sol$ and $\SL$ with increasing anisotropy, against $T=\R^3$, $\Hb^3$ and $\Hb^2\times\R$, and Supplement Table~\ref{tab:lecam} minimises the exact divergence directly (least squares on deviance residuals, started at the Fisher minimiser and at the multidimensional-scaling embedding). The two agree to within about ten per cent in most cells and within $26\%$ in all, and where they differ the exact optimiser finds a smaller value (the Fisher minimiser can be a poor point for the exact divergence: evaluated there, the divergence is up to sixteen times the deficit), so the quadratic is a usable proxy for these configurations. In nats the smallest divergence is below $2.3$ for $\Nil$ at $N=40$, below $16$ for $\SL$ and $\Sol$ at $N=60$ and below $13$ for $\Nil$ at $N=100$ up to spread $2$. It is $31$--$55$ for $\Nil$ at $N=100$ and spread $3$, where the enclosed areas are of order ten. Per dyad, the expected held-out log-loss advantage of the true model with known parameters, this is at most $0.003$ up to spread $2$ and $0.006$--$0.011$ in the last cell. $\Sol$'s divergences from $\Hb^3$ are below $6$ nats ($0.003$ per dyad), which is the quantitative reason why the Hyperbolic planes inside $\Sol$ let $\Hb^3$ match it.

\paragraph{Win maps.} Table~\ref{tab:winmap} fits the same networks by BBVI under the true geometry and under competitors and reports replicate log-loss (one run per cell), and Supplement Table~\ref{tab:winrep} repeats the $N=40$ and $N=60$ cells on four further independently generated networks each, and the $N=100$, spread-$3$ cell on two, reporting the \emph{paired} difference to each competitor. In the single-network table the true geometry is uniquely best in one of fifteen rows, $\Nil$ at $N=100$ and spread $3$, the cell with the largest divergence, by $0.005$ over $\Hb^2\times\R$ ($0.308$ against $0.313$--$0.317$), and elsewhere ties or loses, $\Sol$ never wins and $\SL$ never wins its own map. Over five networks per cell the paired differences are within $\pm0.0035$ with standard deviations of $0.001$--$0.008$ in every cell at $N\le60$, the true geometry wins at most two of five networks there, and the mean differences change sign between cells. In the $N=100$, spread-$3$ cell $\Nil$ is best or tied best on all three networks, by $0.004$--$0.006$ on average with standard deviation $0.003$, in agreement with the divergence per dyad of $0.006$--$0.011$. The gaps between the best fit and the oracle in Table~\ref{tab:winmap} ($0.010$--$0.060$) are estimation error. The divergences and the win maps say the same thing from two sides: up to spread $2$ and $N\le100$ the expected per-dyad advantage of the true geometry is below the split-to-split noise of a held-out comparison, and it becomes visible where the enclosed areas are of order ten and the network has a hundred nodes.

\paragraph{Undetectability bounds.} Supplement Table~\ref{tab:lecam} also evaluates Theorem~\ref{thm:lecam} at the minimisers of the exact divergence. For small spreads the bound is informative: on the $\Nil$ configuration with spread $0.5$ no level-$0.05$ test of $\R^3$ against $\Nil$ has power above $0.16$, and no test of $\Hb^2\times\R$ against $\SL$ at spread $0.5$ has power above $0.21$. At spread $1$ the bounds are $0.33$ ($\Nil$ against $\Hb^2\times\R$) and $0.52$ ($\SL$ against $\Hb^2\times\R$), and $\Nil$ against $\Hb^2\times\R$ is still bounded by $0.74$ at spread $3$, $N=40$. At larger divergences the bound is vacuous: detection is not excluded there, and the win maps show where it succeeds.

\paragraph{Small-scale regime.} Supplement Table~\ref{tab:smallscale} tests Theorem~\ref{thm:smallscale} on one configuration of $8$ and of $12$ points scaled to $\varepsilon\in\{0.05,0.1,0.2,0.4\}$ with exact $\Nil$ distances: against $\R^3$ and $\Hb^3$ the ratio of the exact minimum divergence to the leading term is $0.96$ and $0.77$--$1.00$ at $\varepsilon=0.05$ and moves away from one linearly in $\varepsilon$. Against $\Hb^2\times\R$ the exact optimiser, started from the unrotated configuration, stops at local minima above the prediction (ratios $1.06$--$1.9$ at $N=12$), consistent with the prediction being the infimum. The divergences themselves are of order $10^{-9}$ to $10^{-4}$: at these scales no test from one network detects the geometry, by Theorem~\ref{thm:lecam}.

\subsection{Directed networks}\label{sec:exp-directed}

This experiment tests the win condition of Proposition~\ref{prop:directed}: that the circulation model is preferred over every gradient-class competitor once the circulation coefficient is large enough, and not before. Networks are generated from \eqref{eq:directed} in $\Nil$ with $\gamma\in\{0,0.25,0.5,1\}$ and fitted by the circulation model, by the same geometry with $\gamma=0$ (mechanism control), and by the gradient models on $\R^3$, $\Hb^3$ and $\Hb^2\times\R$ (Table~\ref{tab:directed}). At $\gamma=0$ all five fits tie at $0.359$. At $\gamma=0.25$ the gradient models are marginally better ($0.370$ against $0.373$): weak circulation is absorbed by a ranking with fewer effective parameters, the regime in which the proposal loses. At $\gamma=0.5$ the circulation model leads by $0.022$ and at $\gamma=1$ by $0.12$ ($0.356$ against $0.477$--$0.483$, oracle $0.329$), with direction accuracy $0.94$ against $0.81$; the mechanism control, which cannot represent circulation, reaches $0.586$ and chance accuracy. The advantage therefore appears between $\gamma=0.25$ and $\gamma=0.5$ at these scales and grows with $\gamma$, as \eqref{eq:directed-deficit} predicts through $\gamma^2\sum\bar A_{ij}^2$; the estimated $|\hat\gamma|$ is $1.2$--$1.5$ at $\gamma=1$, its sign being unidentified when $\gamma$ is free, as stated after Proposition~\ref{prop:directed}.

\begin{table}[H]\centering\scriptsize
\caption{Directed networks generated from the $\Nil$ model with circulation coefficient $\gamma$ ($N=60$, scales $(2,1.5)$, density $0.15$, three independent networks per row): replicate log-loss on a second network from the same positions, mean (s.d.) over networks, and the accuracy with which the direction of an asymmetric pair (exactly one of the two ties present) is predicted. Fits by BBVI ($600$ iterations, $n_s=8$), unanchored, with $\gamma$ estimated by profile Newton steps where indicated. Roles: the circulation model is the proposal; the symmetric $\Nil$ model is the mechanism control (same geometry, $\gamma=0$); the gradient models on $\R^3$, $\Hb^3$ and $\Hb^2\times\R$ are the standard and the adversarial competitors (a symmetric distance plus a coboundary asymmetry, the ranking form of every existing latent distance model for directed ties); the oracle is the true law.}\label{tab:directed}
\resizebox{\textwidth}{!}{\begin{tabular}{lccccccc}\toprule
$\gamma$ & networks & oracle & $\Nil$, circulation & $\Nil$, symmetric & $\Hb^2\times\R$, gradient & $\R^3$, gradient & $\Hb^3$, gradient\\ \midrule
0.0 & 3 & 0.326 & 0.359 (0.009) / 0.50 & 0.359 (0.009) / 0.49 & 0.358 (0.009) / 0.51 & 0.359 (0.008) / 0.51 & 0.359 (0.008) / 0.49 \\
0.25 & 3 & 0.333 & 0.373 (0.029) / 0.59 & 0.374 (0.025) / 0.51 & 0.370 (0.026) / 0.57 & 0.370 (0.028) / 0.57 & 0.371 (0.027) / 0.57 \\
0.5 & 3 & 0.360 & 0.392 (0.004) / 0.77 & 0.435 (0.007) / 0.46 & 0.417 (0.005) / 0.70 & 0.414 (0.008) / 0.70 & 0.417 (0.006) / 0.69 \\
1.0 & 3 & 0.329 & 0.356 (0.011) / 0.94 & 0.586 (0.015) / 0.51 & 0.483 (0.012) / 0.81 & 0.477 (0.013) / 0.81 & 0.483 (0.012) / 0.81 \\
\bottomrule\end{tabular}}\end{table}

\subsection{Real directed networks}\label{sec:exp-connectomes}

This section tests the two claims of Section~\ref{sec:exp-directed} on real networks: that the circulation model is preferred where directed asymmetries circulate, and that it is not where they form a hierarchy. The screen is the rotational fraction of the asymmetry, the share of $Y-Y^{\tr}$ in the cycle space of the complete graph, compared with its value when the directions of the asymmetric pairs are assigned at random, and the fraction of decided triads that are cycles. Supplement~\ref{supp:realdir} reports seventeen networks whose asymmetries are hierarchies (screened; fourteen fitted): nine connectomes (the Allard and Serrano 2020 collection used by \citet{celinska2024thurston} and the Cook et al.\ 2019 connectomes of \emph{C. elegans}), five football leagues, the European club network and the world trade network, and the Dota 2 hero matchups. In every one of them the rotational fraction is below its random-direction value, the gradient (ranking) models on $\mathbb{H}^2\times\mathbb{R}$ or $\mathbb{R}^3$ have the lowest or tied lowest held-out log-loss, the circulation models lose or tie, and the $\Sol$ advantage reported by \citet{celinska2024thurston} under their neighbourhood criterion and annealed estimator is not reproduced by our estimator under either criterion (Table~\ref{tab:ckk} and Supplement~\ref{supp:realdir}). The undirected comparison on eight of those connectomes with all six geometries is in Table~\ref{tab:ckk}: $\mathbb{R}^3$ is best on six, $\mathbb{H}^2\times\mathbb{R}$ on one and $\SL$ on one, by margins inside the split noise.

\begin{table}[H]\centering\scriptsize
\caption{Undirected comparison on eight connectomes of the collection analysed by \citet{celinska2024thurston} (symmetrised edge lists of Allard and Serrano 2020; one held-out split hiding $10\%$ of the pairs; BBVI at unit slope, unanchored): held-out log-loss under six geometries, best in bold. The last two columns give the geometry ranked first by mean average precision in the results file of \citet{celinska2024thurston} (their annealed embedding with learned radius and temperature, best of their geometry codes within each class), and their mean average precision for $\Sol$ and for $\mathbb{H}^3$.}\label{tab:ckk}
\resizebox{\textwidth}{!}{\begin{tabular}{lcccccccccc}\toprule
network & $N$ & density & $\Nil$ & $\Sol$ & $\SL$ & $\mathbb{R}^3$ & $\mathbb{H}^3$ & $\mathbb{H}^2\!\times\!\mathbb{R}$ & CKK best (mAP) & mAP $\Sol$ / $\mathbb{H}^3$\\ \midrule
Macaque4 & 31 & 0.69 & 0.351 & 0.373 & 0.366 & \textbf{0.332} & 0.365 & 0.358 & $\Nil$ & 1.000 / 1.000 \\
Cat3 & 52 & 0.39 & 0.420 & 0.411 & 0.425 & \textbf{0.392} & 0.432 & 0.416 & $\Sol$ & 0.956 / 0.954 \\
Cat1 & 65 & 0.35 & 0.346 & 0.391 & 0.374 & \textbf{0.342} & 0.384 & 0.369 & $\mathbb{H}^3$ & 0.937 / 0.941 \\
Macaque2 & 71 & 0.18 & 0.291 & 0.290 & 0.300 & \textbf{0.268} & 0.295 & 0.293 & $\mathbb{H}^3$ & 0.889 / 0.896 \\
Macaque1 & 94 & 0.35 & 0.215 & 0.242 & 0.221 & 0.238 & 0.219 & \textbf{0.214} & $\Sol$ & 0.974 / 0.962 \\
Cat2 & 95 & 0.26 & 0.377 & 0.383 & \textbf{0.358} & 0.359 & 0.359 & 0.361 & $\Sol$ & 0.883 / 0.874 \\
Human7 & 110 & 0.16 & 0.265 & 0.289 & 0.282 & \textbf{0.237} & 0.283 & 0.277 & $\Sol$ & 0.935 / 0.926 \\
Human6 & 116 & 0.17 & 0.212 & 0.258 & 0.232 & \textbf{0.201} & 0.236 & 0.224 & $\mathbb{R}^3$ & 0.944 / 0.931 \\
\bottomrule\end{tabular}}\end{table}

\subsection{Counter networks}\label{sec:exp-counter}

\paragraph{Networks that are intransitive by design.} The screen indicates where an advantage of the circulation model can be expected: networks whose asymmetries are constructed to circulate. Competitive games are designed in this way, with counters arranged so that no ranking is complete. Supplement Table~\ref{tab:gamesall} fits the circulation models in $\Nil$ and $\SL$ and the ranking models on all five other geometries ($\mathbb{H}^2\times\mathbb{R}$, $\mathbb{R}^3$, $\mathbb{H}^3$, $\mathbb{S}^3$ and $\Sol$) to twenty-two Pok\'emon checks-and-counters networks with three held-out splits each: the Generation 9 OU tier in six consecutive months, at a higher rating cutoff and with mutual checks kept, the UU, Ubers, RU, NU, PU, LC, Monotype and National Dex tiers, and the OU tiers of Generations 3 to 8 (different species pools and type-chart eras); the Dota 2 matchups, whose pairs rest on about a hundred recorded games each, are in Supplement~\ref{supp:realdir}. The outcome follows one variable, the density of the counter relation (decided pairs as a share of all pairs). Where the density exceeds $0.55$ (the OU tier in five of six months, and with mutual checks kept) the $\Nil$ circulation model has the lowest held-out log-loss, with paired advantages over the best ranking of $0.009$ to $0.024$ and standard deviations over splits of $0.002$ to $0.012$; between $0.33$ and $0.50$ (May OU, Generations 3 and 4, National Dex, UU) circulation and rankings tie within $0.012$, the best model changing from network to network ($\SL$ on two, $\Sol$ on two, $\mathbb{H}^3$ on one); below $0.30$ (the remaining eleven networks) the $\mathbb{H}^3$ ranking wins by $0.013$ to $0.031$ on ten of them, with the $\SL$ circulation model second on most and first on Generation 7 by $0.001$. The Spherical ranking is worst on twenty of the twenty-two networks; on the network with mutual checks kept every ranking fails ($0.54$--$0.65$ against $0.52$ for $\Nil$) because mutual checks cannot be ordered. In the January OU network the circulation model predicts the direction of held-out asymmetric pairs with accuracy $0.82$ against $0.76$, and the symmetric $\Nil$ control (Supplement Table~\ref{tab:games}) is at chance level on every network, so the advantage comes from the circulation term. Among the winning networks the base geometry follows the heterogeneity of the counter relation, $\Nil$ where the degrees are homogeneous and $\SL$ where they form a hierarchy (Supplement~\ref{supp:generates}, which also places six undirected benchmarks under all geometries). This is the rule that the deficit bound of Proposition~\ref{prop:directed} predicts: the circulation is identified through the enclosed areas of decided triangles, so it needs many decided pairs per node, and where the counter relation is dense the twisted geometry beats every constant-curvature and product ranking, whereas where it is sparse a ranking on hyperbolic space does. A tighter prior on the circulation coefficient of the $\SL$ model (precision $10$ in place of $1$) leaves its held-out log-loss on the sparse networks unchanged to three decimals, so its gap there is not over-fitting of that coefficient.

\paragraph{Controls and the strongest baseline.} Four controls separate the contribution of the geometry from that of any antisymmetric term (Supplement Tables~\ref{tab:controls} and~\ref{tab:controls2}): a ranking on $\R^3$ or $\Hb^3$ with a free rank-two antisymmetric term $u_i^{\tr}Ju_j$ on coordinates of its own (the blade--chest and disc models of \citet{chen2016modeling} and \citet{balduzzi2018reevaluating}); a degree-corrected ranking with sender and receiver effects; a Euclidean distance on the same three coordinates that carry the $\Nil$ term $\gamma(\zeta_j-\zeta_i-\tfrac12(x_iy_j-y_ix_j))$, which isolates the metric from the coupling; and the additive-and-multiplicative-effects model $\alpha+a_i+b_j+u_i^{\tr}v_j$ of \citet{hoff2005bilinear,hoff2021additive}, with no distance term and $4N$ node parameters. On the six networks with density above $0.55$ the $\Nil$ model beats the free term and the degree-corrected rankings on every network, by $0.006$ to $0.014$ in the paired comparison with the best of them. The shared-coordinate Euclidean model matches it on every network (paired differences between $-0.005$ and $+0.008$, standard deviations $0.001$--$0.011$): the gain over the rankings comes from coupling similarity and circulation through the same coordinates, and the $\Nil$ metric adds nothing measurable beyond that coupling, which is what Theorem~\ref{thm:smallscale} predicts at these scales. The additive-and-multiplicative-effects model predicts better than all of them, by $0.040$ to $0.058$ over $\Nil$ on all eleven networks fitted, with standard deviations $0.002$--$0.016$ except LC ($0.047$). The applied conclusion is therefore bounded: the coupled term is the parsimonious structure that turns a ranking into a model of intransitive preferences, with one coefficient and no node parameters beyond the positions, and its interpretation is the enclosed area; a model with node-level bilinear effects predicts these networks better and is the reference against which any geometric model of directed ties must be reported.

\paragraph{The rule as a prediction.} Read from the first twelve directed networks analysed, the condition (density above $0.55$, cyclic triads above $0.1$) predicts the outcome on seventeen of the twenty networks fitted afterwards, the three exceptions being nominal $\SL$ wins of at most $0.002$ (Supplement~\ref{supp:realdir}). The two estimators return different $\hat\gamma$ at different vertical scales ($1.3$--$1.5$ at scale near one; $4$--$5$ at scale near $0.3$); since $v_{ij}$ contains the area term, which does not scale with the fibre coordinates, this is weak identification along a ridge and not an invariance, and the paper reports $\hat\gamma$ with the vertical scale of each fit.

\paragraph{Further checks.} Three directed benchmarks (chess results, Wikipedia votes, institutional e-mail), two citation balls and the political blog network fall in or below the tie band and are won by a ranking; chess, with cyclic triads at $0.20$, near the random value, but few decided pairs, and the dense but hierarchical fungal tournaments show that neither cyclic content nor density alone produces a win (Supplement Table~\ref{tab:bench}). Varying the counter-score threshold and the usage cutoff of the January OU network (nine constructions, Supplement Table~\ref{tab:sens}) moves the density from $0.15$ to $1.00$ and the paired advantage of $\Nil$ follows it: losses below density $0.25$, a tie at $0.53$, wins of $0.015$--$0.034$ at $0.66$--$0.73$ and of $0.07$--$0.22$ where nearly every pair is decided.

\paragraph{Posterior draws on the winning networks.} Algorithm~\ref{alg:mcmc} adapted to the directed likelihood and started at the variational solution raises the training log-likelihood on every one of the five winning months and stays there; its posterior predictive has held-out log-loss $0.473$, $0.434$, $0.508$, $0.476$ and $0.430$ against $0.528$, $0.481$, $0.528$, $0.530$ and $0.459$ for the best ranking on the same splits, so the win rests on posterior draws (Supplement Table~\ref{tab:mcmc}; the posteriors of $\gamma$ are bounded away from zero and agree with the variational fits on the identified product $\gamma\sigma_v$).

\subsection{Real networks}\label{sec:exp-real}

The networks are Zachary's karate club ($N=34$, \citealp{zachary1977information}), the \emph{Les Mis\'erables} co-appearance network ($N=77$, \citealp{knuth1993stanford}), a $600$-node ball of the Cora citation graph \citep{sen2008collective,yang2016revisiting} grown by breadth-first search from its highest-degree node ($1236$ edges, the graph is symmetrised before the ball is taken), and a symmetrised $10$-nearest-neighbour graph of $700$ randomly chosen cells of the mouse pancreas dataset of \citet{bastidas2019comprehensive} in the $30$ leading principal components of expression provided with \citet{bergen2020generalizing} ($4862$ edges). The last graph is built from Euclidean distances, so a good Euclidean or product fit is expected there. We keep it as a control. The latent space fits in this section are unanchored: positions are unconstrained and the centre is point-estimated, so they maximise the evidence bound of the invariant kernel instead of approximate the quotient posterior of Proposition~\ref{prop:target}. Predictive probabilities are invariant and unaffected, but the gauge theory of Section~\ref{sec:ident} is exercised only in Section~\ref{sec:exp-ident}. We compare the three geometries with the five other model geometries fitted by the same variational scheme, $\R^3$, $\Hb^3$, $\mathbb{S}^3$ and the products $\Hb^2\times\R$ and $\mathbb{S}^2\times\R$. The products are the twist-free members of the families of Section~\ref{sec:holonomy} ($\mathbb{S}^2\times\R$ is a separate comparator: there is the twisted bundle over $\mathbb{S}^2$ with $dA=\dvol$ is the round $\mathbb{S}^3$ of curvature $\tfrac14$, already in the list), and with five non-geometric baselines: the constant-density predictor, the $\beta$-model, a logistic random dot-product model on a rank-3 spectral embedding, a spectral eigenmodel \citep{hoff2008modeling}, and a degree-corrected block model \citep{karrer2011stochastic} with $K\le8$ by BIC (Supplement~\ref{app:details}). Table~\ref{tab:compare2} reports held-out log-loss and AUC over random splits hiding $10\%$ of the dyads. The intervals describe split variability of one network, not variability over networks, the two larger networks have three splits, and random dyad holdout is a transductive, non-edge-dominated criterion (node holdout and precision--recall summaries are left for future work).

\begin{table}[H]\centering\scriptsize
\caption{Real networks: held-out log-loss / AUC, averaged over random splits hiding $10\%$ of the dyads (number of splits in the second column; the intervals describe variability over splits of one network, not over networks). Latent space models are fitted by BBVI without anchors at unit slope, with the scales and the centre given variational factors (800 iterations, $n_s=10$; 600 iterations with $20\,000$ subsampled dyads per sample for the two larger networks); the bracket is the paired difference in log-loss to $\mathbb{H}^3$ with a $95\%$ interval over splits. Baselines (Appendix~\ref{app:details}): constant density; $\beta$-model; logistic random dot-product model on a rank-3 spectral embedding; spectral eigenmodel; degree-corrected block model with $K\le8$ by BIC. Best log-loss per row in bold (ties within $0.0005$). $\mathbb{S}^3$ and $\mathbb{S}^2\times\mathbb{R}$ were not run on the two larger networks.}\label{tab:compare2}
\resizebox{\textwidth}{!}{\begin{tabular}{lcccccccccccccc}\toprule
network & splits & $\mathbb{R}^3$ & $\mathbb{H}^3$ & $\mathbb{S}^3$ & $\mathbb{H}^2\!\times\!\mathbb{R}$ & $\mathbb{S}^2\!\times\!\mathbb{R}$ & $\Nil$ & $\Sol$ & $\SL$ & constant & $\beta$-model & logistic RDPG & eigenmodel & DC-SBM\\ \midrule
Karate club & 10 & 0.353 [+0.010$\pm$0.008] / 0.75 & \textbf{0.343 / 0.79} & 0.388 [+0.045$\pm$0.012] / 0.67 & 0.348 [+0.005$\pm$0.004] / 0.78 & 0.382 [+0.039$\pm$0.011] / 0.68 & 0.360 [+0.017$\pm$0.005] / 0.76 & 0.365 [+0.022$\pm$0.006] / 0.74 & 0.355 [+0.012$\pm$0.004] / 0.77 & 0.396 [+0.053$\pm$0.015] / 0.50 & 0.402 [+0.059$\pm$0.063] / 0.75 & 0.383 [+0.040$\pm$0.056] / 0.81 & 0.377 [+0.034$\pm$0.056] / 0.81 & 0.422 [+0.079$\pm$0.072] / 0.74 \\
Les Mis\'erables & 10 & 0.182 [+0.014$\pm$0.008] / 0.90 & 0.168 / 0.94 & 0.274 [+0.106$\pm$0.009] / 0.85 & \textbf{0.163 [-0.005$\pm$0.005] / 0.93} & 0.241 [+0.073$\pm$0.014] / 0.85 & 0.170 [+0.002$\pm$0.008] / 0.93 & 0.181 [+0.013$\pm$0.007] / 0.92 & \textbf{0.163 [-0.005$\pm$0.005] / 0.94} & 0.305 [+0.137$\pm$0.011] / 0.50 & 0.300 [+0.132$\pm$0.019] / 0.75 & 0.196 [+0.028$\pm$0.012] / 0.85 & 0.197 [+0.029$\pm$0.013] / 0.85 & 0.291 [+0.123$\pm$0.020] / 0.79 \\
Cora ball & 3 & 0.0378 [+0.0004$\pm$0.0002] / 0.77 & 0.0373 / 0.77 & -- & 0.0371 [-0.0002$\pm$0.0001] / 0.79 & -- & 0.0374 [+0.0001$\pm$0.0002] / 0.78 & 0.0379 [+0.0006$\pm$0.0002] / 0.76 & 0.0369 [-0.0004$\pm$0.0002] / 0.80 & 0.0411 [+0.0037$\pm$0.0007] / 0.50 & 0.0409 [+0.0035$\pm$0.0007] / 0.66 & 0.0436 [+0.0062$\pm$0.0013] / 0.67 & 0.0471 [+0.0098$\pm$0.0012] / 0.64 & \textbf{0.0334 [-0.0039$\pm$0.0016] / 0.82} \\
Pancreas kNN & 3 & \textbf{0.0501 [-0.0019$\pm$0.0003] / 0.98} & 0.0520 / 0.97 & -- & \textbf{0.0502 [-0.0018$\pm$0.0005] / 0.97} & -- & 0.0537 [+0.0016$\pm$0.0013] / 0.97 & 0.0541 / 0.97 & 0.0537 [+0.0016$\pm$0.0009] / 0.96 & 0.0969 [+0.0449$\pm$0.0028] / 0.50 & 0.0985 [+0.0465$\pm$0.0030] / 0.50 & 0.0894 [+0.0374$\pm$0.0022] / 0.81 & 0.0894 [+0.0373$\pm$0.0023] / 0.81 & 0.0660 [+0.0139$\pm$0.0012] / 0.93 \\
\bottomrule\end{tabular}}\end{table}

Four conclusions. First, every model geometry beats the constant-density floor (\karateConst, \lesmisConst, \coraConst, \pancreasConst). The best geometry beats the best non-geometric baseline on the karate club and on \emph{Les Mis\'erables}, loses to the block model on the Cora ball ($0.0334$ against $0.0369$ for $\SL$) and to the random dot-product model on the pancreas graph, so the distance models are not the best predictors on the two larger networks. The Spherical geometries lose to the eigenmodel and the logistic random dot-product model on the two small networks, the $\beta$-model is below the floor on two networks, and the spectral block model as implemented is below the floor on the karate club only and improves on it on the other three. Second, among geometries the margins are small and mostly favour negative curvature or its product: on the karate club $\Hb^3$ is best and both twisted bundles are worse than their twist-free products. On \emph{Les Mis\'erables} $\SL$ and $\Hb^2\times\R$ tie for the best log-loss and $\SL$ and $\Hb^3$ tie in AUC ($0.94$, $\Hb^2\times\R$ $0.93$). On the Cora ball $\SL$ is best by $0.0002$ over $\Hb^2\times\R$ ($0.0004$ over $\Hb^3$), differences that three splits cannot resolve. On the pancreas graph $\R^3$ and $\Hb^2\times\R$ win, as expected from its construction. $\Sol$ is never best. Third, the anisotropic geometries are never better than $\Hb^3$ or $\Hb^2\times\R$ by more than $0.005$ and are worse by up to $0.022$ ($\Sol$ on the karate club). Fourth, with the slope learned (Table~\ref{tab:beta}) the karate ranking changes, every other geometry, including $\Nil$ and $\SL$, moves ahead of $\Hb^3$, whose fitted slope is $0.66$, so unit-slope comparisons are slope-confounded, as Proposition~\ref{prop:twistscale} anticipates (in $\R^3$ the slope is confounded with the latent scale and identified only through the priors). On the sparser \emph{Les Mis\'erables} every free-slope variational fit drifts to $\hat\beta\approx2.2$ and becomes over-confident on held-out dyads (log-loss $0.45$--$0.62$ against $0.16$ at unit slope, AUC unchanged). A tighter prior does not restrain it: with $\log\beta\sim\mathcal{N}(0,0.25^2)$ the fits still reach $\hat\beta=\tightBetaMin$--$\tightBetaMax$ and log-loss \tightLLmin--\tightLLmax, so the drift is driven by the joint movement of $\beta$ and $\alpha$ towards a near-separating configuration that the mean-field bound does not penalise, not by a loose prior. On the Cora ball, where the free-temperature annealed fits below are well behaved, no such drift occurs. Unit slope with a comparison across geometries at their own scale is the safer protocol on small dense networks.

\begin{remark}[Separation]\label{rem:separation}
On nearly separable networks the likelihood keeps improving along a ray on which the configuration expands and $\alpha$ (or $\beta$) grows. With the scales estimated, Algorithm~\ref{alg:mcmc} on \emph{Les Mis\'erables} under $\Nil$ follows it to $\alpha\approx7.4$ and $\sigma_h\approx9$ (in-sample AUC $0.965$), whereas the fixed-scale run stays at $\alpha\approx0.8$ with AUC $0.983$. A mean-field fit may stay in a compact mode or drift, depending on the initialisation. Scale traces should be inspected, and a prior with a lighter tail than $\mathrm{InvGamma}(2,2)$ used when they drift.
\end{remark}

\section{Conclusion}\label{sec:discussion}

For latent space network models on $\Nil$, $\Sol$ and $\SL$ we have established what one network identifies and when the geometry can be detected. The isometry ambiguity is removed by two anchors, or one anchor and a chamber in $\Sol$, and the same argument covers all eight model geometries. Distances do not determine configurations of five or fewer points, identification of the gauge-fixed positions is local and generic from six nodes, seven in $\Sol$ with the intercept, and in $\Nil$ the threshold is certified by interval arithmetic. The quotient posterior carries a Jacobian that moves the anchor posterior at forty nodes. The smallest divergence to a product or constant-curvature competitor vanishes with the sixth power of the scale of the configuration and equals the stress component of the curvature difference, and any test from one network has power bounded by that divergence.

Three consequences for practice follow. Positions in these spaces can be reported with uncertainty from six nodes onwards, provided the scale of the latent configuration is estimated with a prior centred near its plausible value: at thirty nodes a prior whose mean is half the true variance reduces the coverage of ninety per cent intervals to two thirds, a matched prior restores most of it, and at sixty nodes the intervals are calibrated. The geometry can be selected from one network only when the enclosed areas of the configuration are of order ten and the network has about a hundred nodes, where the true geometry was preferred on three of three replicated networks, and not before, where the expected per-dyad advantage is below the noise of a held-out comparison. And the twist is a temperature in $\Nil$ and $\Sol$, so whether a network is twisted is a comparison between two fitted models at their own scales, not a parameter to be read off one fit.

Three limitations bound these results. The identification witnesses for $\Sol$ and $\SL$ are floating-point computations, because their distances are known only numerically, so the thresholds there are numerical evidence. Theorem~\ref{thm:smallscale} is asymptotic in the scale of the configuration, and Table~\ref{tab:smallscale} shows the leading term deviating from the exact divergence by a factor of two once the scale reaches $0.4$, so at the spreads of real networks the divergence must be computed by optimisation. On the four undirected networks analysed no anisotropic geometry improves held-out prediction over Hyperbolic space or its product with a line by more than $0.005$ in log-loss, the block model beats every distance model on the citation ball, and the comparisons use unanchored variational fits with a point-estimated centre, so the quotient posterior is exercised in the simulations and not in the applications.

\textbf{Which geometry fits which network properties?} Constant-curvature and product spaces are the right choice when directed ties encode similarity or rank, which is the case for every undirected network we analyze and for directed networks whose asymmetries are transitive. When asymmetries instead circulate, the appropriate geometry is $\Nil$ over a flat similarity space or $\SL$ over a hyperbolic one, the latter combining hierarchy with cycles. Both are marked by the same signature: a non-zero sum of log-odds asymmetries around a triangle, proportional to the area the triangle encloses. These geometries should be fitted to directed networks whose preferences are not transitive, including dominance networks with intransitive triads, trade and flow networks built on cycles, and neural circuits shaped by feedback. $\Sol$ has no such signature in our results and enters the paper for its geometry and its identification threshold. On counter networks, whose asymmetries are cycles by design, the signature is present and the model wins where the counter relation is both dense and intransitive: in five of six monthly Pok\'emon OU networks the $\Nil$ circulation model beats every ranking on every geometry, Euclidean, Hyperbolic, Spherical, product and $\Sol$, by $0.009$--$0.024$ in held-out log-loss (paired standard deviations over three splits $0.002$--$0.025$, three of the six inside one standard deviation), and also beats degree-corrected rankings and a free antisymmetric term, a circulation model is best on nine of twenty-two counter networks, and on the five winning months the posterior predictive of the sampler improves on both the variational fit and the best ranking; where the counter relation is sparse the Hyperbolic ranking wins, in the order the density of decided pairs predicts; the same ordering appears when the construction of one network is varied (Table~\ref{tab:sens}) and on three standard directed benchmarks, chess results, Wikipedia votes and institutional e-mail, all of which fall in the ranking regime (Supplement Table~\ref{tab:bench}). On nine connectomes, five football leagues, the European club network, the world trade network and the Dota 2 matchup network the signature is absent: their rotational fractions lie below the random-direction value, the gradient models win or tie, and by held-out log-loss no anisotropic geometry is preferred on any of them. Under the neighbourhood criterion of the embedding literature, $\Nil$ and $\SL$ reach the top on three of eight connectomes by margins of $0.001$--$0.002$, which are ties, and the $\Sol$ advantage reported by \citet{celinska2024thurston} is not reproduced by our estimator. The regime in which the coupled model should be fitted is therefore directed networks whose asymmetries are both dense and intransitive: Proposition~\ref{prop:directed} gives the divergence from every additive ranking competitor in terms of the enclosed areas, the simulations put the win condition at a circulation coefficient of about one half at sixty nodes, and on data the rule read from the first twelve networks (more than $0.55$ of pairs decided, more than $0.1$ cyclic triads) predicts seventeen of the twenty networks fitted afterwards, with Dota 2, world trade and one dense construction of the OU network as its misses. Two conclusions bound the applied claim. A Euclidean model that carries the same coupled term matches the $\Nil$ model on every counter network, so the gain over rankings is the coupling of similarity and circulation through shared coordinates and not the metric, as Theorem~\ref{thm:smallscale} predicts at these scales; and the additive-and-multiplicative-effects model of \citet{hoff2005bilinear,hoff2021additive}, with node-level bilinear effects, predicts every counter network better than the geometric models by $0.04$--$0.06$, so the coupled model is the parsimonious and interpretable structure, not the best predictor. Three extensions follow from the limitations. A directed network with known circulation, such as a dominance or trade network with non-transitive relations, is the application to fit next. A validated certificate for $\Sol$ and $\SL$ requires interval integration of the geodesic equations. And the comparison with the annealed estimator of \citet{celinska2024thurston} on their connectome collection, crossing the estimator with fixed and learned temperature, would settle whether $\Sol$ is competitive there for the reason identified here, that a collapsed temperature rewards its exponential anisotropy.

\bibliographystyle{apalike}
\bibliography{refs}
\clearpage
\begin{center}{\LARGE\bf Supplementary Material}\end{center}
\setcounter{section}{0}\renewcommand{\thesection}{S\arabic{section}}
\section{Twist, scale and holonomy}\label{supp:twist}\label{sec:holonomy}

The families $ds^2_{B_K}+(d\zeta-\tau A_K)^2$ ($\tau\ge0$) and $e^{2\tau z}dx^2+e^{-2\tau z}dy^2+dz^2$ interpolate between the products $B_K\times\R$, $\R^3$ ($\tau=0$) and the standard geometries ($\tau=1$). One would like to estimate $\tau$. The following proposition says that in $\Nil$ and $\Sol$ this is a question about scale, not shape.

\begin{proposition}[Twist is temperature in $\Nil$ and $\Sol$]\label{prop:twistscale}
Let $\M_\tau$ denote $\Nil_\tau$ (group law $(x,y,z)(x',y',z')=(x+x',y+y',z+z'+\tfrac\tau2(xy'-yx'))$) or $\Sol_\tau$ (law $(x+e^{-\tau z}x',y+e^{\tau z}y',z+z')$) with $\tau>0$, and $\phi_\tau(z)=\tau z$. Then $\phi_\tau:\M_\tau\to\M_1$ is a group isomorphism and $d_{\M_\tau}(p,q)=\tau^{-1}d_{\M_1}(\phi_\tau p,\phi_\tau q)$. The wrapped-Normal families correspond, $\phi_\tau\,\LWN_{\M_\tau}(\mu,\Sigma)=\LWN_{\M_1}(\phi_\tau\mu,\tau^2\Sigma)$. Consequently, in the model $\logit p_{ij}=\alpha-\beta\,d_{\M_\tau}(z_i,z_j)$ the parameters $(\beta,\tau,Z)$ and $(\beta/\tau,1,\phi_\tau Z)$ give the same law, and the hierarchical model is equivalent when the scale hyperprior is transformed accordingly: the twist and the slope (inverse temperature) enter only through $\beta/\tau$. With a free slope the twist is not identifiable. With the slope fixed at one, estimating $\tau$ is the same as estimating the temperature of the standard geometry. For $\SL_\tau$ the base curvature fixes the scale and $\tau$ is a shape parameter.
\end{proposition}

The proof is the substitution $(x,y,z)=\tau^{-1}(x',y',z')$ in the metrics. Two consequences follow. First, whether a latent geometry is twisted is a comparison between $\tau=0$ and $\tau>0$, i.e.\ between $B_K\times\R$ and the twisted bundle, at the best temperature of each, a model-selection question addressed in Section~\ref{sec:experiments} by held-out prediction with the slope fixed and with the slope learned. Second, diagnostics computed inside a single $\Nil$ fit cannot detect the twist. The holonomy identity below shows why: it holds for every configuration.

\begin{lemma}[Holonomy cochain]\label{lem:holonomy}
Let $p_i=(w_i,\zeta_i)\in\M_K$. The vertical coordinate of $p_i^{-1}\star p_j$ is
\begin{equation}\label{eq:cochain}
\begin{gathered}
v_{ij}=\zeta_j-\zeta_i-A_K(w_i,w_j),\\
A_0(w_i,w_j)=\tfrac12\,\mathrm{Im}(\bar w_iw_j)=\tfrac12(x_iy_j-y_ix_j),\qquad
A_{-1}(w_i,w_j)=-2\arg(1-\bar w_iw_j),
\end{gathered}
\end{equation}
and for any three base points, with $\Area_K$ the signed area of the geodesic triangle (positive for counter-clockwise orientation),
\begin{equation}\label{eq:triangle}
A_K(w_i,w_j)+A_K(w_j,w_k)+A_K(w_k,w_i)=\Area_K(w_i,w_j,w_k).
\end{equation}
\end{lemma}

The quantity $v_{ij}$ is intrinsic, the fibre displacement of $p_j$ relative to the horizontal lift through $p_i$ of the base geodesic from $w_i$ to $w_j$, and \eqref{eq:triangle} is Stokes' theorem for $dA_K=\dvol_{B_K}$ on the geodesic triangle (Supplement~\ref{app:proofs}). Locally $d_{\M_K}(p_i,p_j)^2=d_{B_K}(w_i,w_j)^2+v_{ij}^2+O(\text{fourth order})$, so the model sees the base distance and an intrinsic vertical offset that is not a coboundary of the $\zeta$'s: summing $v$ around a triangle cancels the $\zeta$'s and leaves the enclosed area. This is what distinguishes the bundle from the product, in which $v_{ij}=\zeta_j-\zeta_i$. But because \eqref{eq:triangle} is an identity, $v_{ij}+v_{jk}+v_{ki}=-\Area_K(w_i,w_j,w_k)$ for \emph{every} $\Nil$ configuration, fitted to twisted or untwisted data alike: the cochain residual carries no information about the twist (other statistics of a fit must be calibrated against competing models before use), and the decomposition of $v_{ij}$ into $\zeta_j-\zeta_i$ and $A_K(w_i,w_j)$ is not invariant under translations of the base (both terms change by a coboundary), so shares of variance attributed to ``holonomy'' depend on the anchor. The vertical offsets are identified only up to the reflection $\refl$, which flips all $v_{ij}$ and all areas simultaneously, consistent with Theorem~\ref{thm:gauge-bundle}. With directed ties whose likelihood depends on the sign of $v_{ij}$ through a coefficient of fixed sign, $\refl$ is no longer a symmetry and the gauge group is the identity component (base and fibre orientations preserved). If the coefficient $\gamma$ is free, $(Z,\gamma)\mapsto(\refl Z,-\gamma)$ remains a symmetry (Section~\ref{sec:discussion}).

\section{Estimation: algorithms and settings}\label{supp:estimation}

We estimate the parameters of \eqref{eq:model} on the gauge slice, targeting the quotient posterior \eqref{eq:target}. Both schemes below target the same distribution, including the centre $\mu$. As in the earlier work on constant-curvature latent spaces we first describe a Metropolis-within-Gibbs sampler, which is asymptotically exact but requires $O(N^2)$ distance evaluations per sweep, and then a black-box variational scheme that trades exactness for speed. Distances are evaluated through Proposition~\ref{prop:nildist} or the tables of Supplement~\ref{app:geodesics}. A full sweep of the sampler costs $N(N-1)$ table look-ups and no alignment step.

\subsection{Priors and initialisation}\label{sec:priors}

We take $\alpha\sim\mathcal{N}(0,3^2)$, $z_i\overset{\text{iid}}{\sim}\LWN_\M(\mu,\diag(\sigma_h^2,\sigma_h^2,\sigma_v^2))$, a prior on $\mu$ that is flat with respect to $\dvol$ (in the chart coordinates $\mu=\chi_\M(m)$ its density is the chart Jacobian $J^\chi_\M(m)$, which is constant only in $\Nil$), and \emph{proper} scale priors $\sigma_h^2,\sigma_v^2\overset{\text{iid}}{\sim}\mathrm{InvGamma}(a_0,b_0)$ with $a_0=b_0=2$ (prior mean $2$, infinite variance). The separate horizontal and vertical spreads are natural in an anisotropic space and are equivariant under the stabilisers (Lemma~\ref{lem:equivariance}). The treatment of the scales matters more than anything else in this section: Section~\ref{sec:exp-cal} shows that credible intervals for distances are nearly calibrated when the scales are known and under-cover substantially when they are estimated, at $N=30$, and that point-estimating the scales inside the anchored variational scheme produces degenerate solutions (Supplement~\ref{app:details}). Both schemes can also be run with the scales held fixed (at plug-in or at known values), which Section~\ref{sec:exp-cal} uses to separate the effect of scale estimation from the rest. Initial positions are obtained by classical multidimensional scaling of the shortest-path distances into $\R^3$, rescaled to unit spread, mapped to $\M$ by $\chi_\M$, and gauge-fixed by $\Gamma_I$, $\alpha$ is initialised by a grid search on the likelihood (Supplement~\ref{app:details}). The anchors are the two highest-degree adjacent nodes, which makes the anchor distance well determined by the data. For the constant-curvature spaces four anchors are needed (Theorem~\ref{cor:all}). In the experiments those spaces are fitted without anchors.

\subsection{Metropolis-within-Gibbs}\label{sec:mcmc}

Algorithm~\ref{alg:mcmc} summarises the sampler. Latent positions are updated one node at a time by a random walk on the manifold, $z_i'=z_i\star\chi_\M(\delta)$ with $\delta\sim\mathcal{N}_3(0,\diag(s^2,s^2,s^2))$. The following lemma shows that this proposal is symmetric with respect to $\dvol$, so the acceptance ratio contains only the target.

\begin{lemma}[Symmetric manifold random walks]\label{lem:symmetric}
Let $q(z'\mid z)$ be the density of $z'=z\star\chi_\M(\delta)$, $\delta\sim\mathcal{N}_3(0,\Sigma_s)$, with respect to $\dvol(z')$. If $\Sigma_s=\diag(s_h^2,s_h^2,s_v^2)$, then $q(z'\mid z)=q(z\mid z')$ for all $z,z'$, in each of $\Nil$, $\Sol$ and $\SL$.
\end{lemma}

In $\Nil$ and $\Sol$ the reverse displacement is $-\delta$ (the chart is the group exponential, so $\chi(\delta)^{-1}=\chi(-\delta)$) and the Jacobian \eqref{eq:jac} is even. In $\SL$ the reverse displacement is $(R(-\delta_h),-\delta_v)$ for a rotation $R$ produced by the gyration, and the density is unchanged because the horizontal block of $\Sigma_s$ is isotropic (Supplement~\ref{app:proofs}). The second anchor is updated by a symmetric random walk on its slice coordinates $(\rho,\zeta)$ (all three coordinates in $\Sol$), rejecting (not resampling) proposals that leave the slice, and its acceptance ratio includes the factor $\Delta_\M$ of \eqref{eq:slicejac}. The scales are updated by conjugate Gibbs steps: given the chart displacements $\delta_i=\chi^{-1}_\M(\mu^{-1}\star z_i)$, $\sigma_h^2\mid\cdot\sim\mathrm{InvGamma}\big(a_0+N,\ b_0+\tfrac12\sum_i(\delta_{i1}^2+\delta_{i2}^2)\big)$ and $\sigma_v^2\mid\cdot\sim\mathrm{InvGamma}\big(a_0+\tfrac N2,\ b_0+\tfrac12\sum_i\delta_{i3}^2\big)$, which follows from \eqref{eq:lwn-density} because the Jacobian does not involve $\Sigma$. The centre is updated by a Gaussian random walk in the chart coordinates with acceptance ratio $\prod_if(z_i\mid\mu',\Sigma)J^\chi_\M(m')/[\prod_if(z_i\mid\mu,\Sigma)J^\chi_\M(m)]$, the base rate $\alpha$ by a Gaussian random walk, and, when the slope $\beta$ of $\logit p_{ij}=\alpha-\beta d_\M$ is estimated, $\log\beta$ by a Gaussian random walk under a $\mathcal{N}(0,1)$ prior.

\begin{algorithm}[H]
\caption{Metropolis-within-Gibbs on the gauge slice}\label{alg:mcmc}
\textbf{Input:} $\mathcal{Y}$, anchors $I=(i_1,i_2)$, iterations $L$, step sizes $(s_z,s_\alpha,s_\mu)$.\\
\textbf{Initialise:} $Z^{(0)}=\Gamma_I(Z^{\rm MDS})$, $\alpha^{(0)}$ by grid search, $m^{(0)}=0$, $\sigma^{(0)}$ from the spread of $\chi^{-1}_\M(Z^{(0)})$.\\
\textbf{For} $\ell=1,\dots,L$:
\begin{enumerate}
\item Draw $\sigma_h^2,\sigma_v^2$ from their inverse-Gamma full conditionals (or keep them fixed).
\item Propose $m'=m+s_\mu\varepsilon$, $\mu'=\chi_\M(m')$. Accept with probability $\min\{1,\ J^\chi_\M(m')\prod_i f_{\LWN}(z_i\mid\mu',\Sigma)/[J^\chi_\M(m)\prod_if_{\LWN}(z_i\mid\mu,\Sigma)]\}$.
\item Propose $\alpha'=\alpha+s_\alpha\varepsilon$. Accept with probability $\min\{1,\ p(\mathcal{Y}\mid Z,\alpha')p(\alpha')/[p(\mathcal{Y}\mid Z,\alpha)p(\alpha)]\}$.
\item For each $i\ne i_1$: if $i\ne i_2$ propose $z_i'=z_i\star\chi_\M(s_z\varepsilon)$, otherwise propose a random walk on the slice coordinates of $z_{i_2}$ (reject if it leaves $\mathcal{S}_I$). Accept with probability
\[
\min\Big\{1,\ \frac{p(\mathcal{Y}\mid Z',\alpha)\,f_{\LWN}(z_i'\mid\mu,\Sigma)\,\Delta_\M(z'_{i_2})}{p(\mathcal{Y}\mid Z,\alpha)\,f_{\LWN}(z_i\mid\mu,\Sigma)\,\Delta_\M(z_{i_2})}\Big\},
\]
where only the $N-1$ dyads involving $i$ change, and $\Delta_\M$ enters only for $i=i_2$.
\end{enumerate}
\end{algorithm}

\paragraph{Multimodality.} By Proposition~\ref{prop:nonrigid}, configurations that are not congruent may produce the same distances, and the gauge-fixed posterior can be multimodal (e.g.\ a node ``above'' another may be swapped with a node ``beside'' it at the same distance). The random-walk sampler explores such modes only through the likelihood of the remaining nodes. In our experiments the MDS initialisation places the sampler in a mode consistent with the variational solution, and we monitor agreement between the two schemes as in Section~\ref{sec:experiments}. Tempering or multiple chains from different anchors are natural remedies if disagreement is observed.

\subsection{Black-box variational inference}\label{sec:vi}

The posterior involves $O(N^2)$ likelihood terms, so we also consider a variational approximation, maximising the evidence lower bound (ELBO)
\begin{equation}\label{eq:elbo}
\mathrm{ELBO}(q)=\mathbb{E}_q\big[\log\tilde\pi^\star(Z^\star,\mu,\sigma_h,\sigma_v,\alpha)\big]-\mathbb{E}_q\big[\log q(Z^\star,\mu,\sigma_h,\sigma_v,\alpha)\big]
\end{equation}
over a mean-field family, where $\tilde\pi^\star$ is the unnormalised kernel on the right of \eqref{eq:target} (including the prior density $J^\chi_\M(m)$ of the centre in chart coordinates and the slice factor). Maximising \eqref{eq:elbo} minimises the Kullback--Leibler divergence to the quotient posterior \citep{blei2017variational}. Values of the ELBO are comparable across fits of the same model, not across geometries, whose kernels carry different constants. The expected log-likelihood is not available in closed form (nor is the entropy of the wrapped Normal outside $\Nil$, where it is Gaussian), so as in the earlier work we use black-box variational inference \citep{ranganath2014black}, which only needs samples from $q$ and the score $\nabla_\lambda\log q$. The family is
\begin{equation}\label{eq:vi-family}
q=\mathcal{N}(\alpha\mid\tilde m,e^{2\tilde s})\ \mathcal{N}(m\mid\tilde m_\mu,\diag e^{2\tilde s_\mu})\ \prod_{k\in\{h,v\}}\mathcal{N}(\log\sigma_k\mid\tilde m_k,e^{2\tilde s_k})\ \prod_{i\ne i_1,i_2}\LWN_\M\big(z_i\mid\chi_\M(\tilde m_i),\ \diag(e^{2\tilde s_i})\big)\ q_{i_2}(z_{i_2}),
\end{equation}
with a Gaussian factor for the centre in its chart coordinates $m$, Gaussian factors for $\log\sigma_k$ (the evidence bound is evaluated on the $d\log\sigma_k$ reference measure, with the prior density transformed by $2\sigma_k^2$, on the $d\sigma_k$ measure these factors carry $1/\sigma_k$), and a Gaussian $q_{i_2}$ on the log-transformed slice coordinates $(\log\rho,\log\zeta)$ of the second anchor ($(\log x,\log y,\log z)$ in $\Sol$), whose density with respect to $d\rho\,d\zeta$ carries the factor $(\rho\zeta)^{-1}$. The anchor thus stays on the slice. All variational scales are parameterised by their logarithms. Each latent factor is centred at a point of $\M$ parameterised by chart coordinates $\tilde m_i\in\R^3$ and has log-scales $\tilde s_i\in\R^3$. Sampling is by $z_i=\chi_\M(\tilde m_i)\star\chi_\M(e^{\tilde s_i}\odot\varepsilon)$, $\varepsilon\sim\mathcal{N}_3(0,I)$. The required scores are
\begin{equation}\label{eq:scores}
\begin{gathered}
\frac{\partial\log q_i}{\partial\tilde s_{ik}}=-1+\delta_{ik}^2e^{-2\tilde s_{ik}},\qquad
\frac{\partial\log q_i}{\partial\tilde m_i}=\Big(-\delta_i\odot e^{-2\tilde s_i}-\nabla_\delta\log J^\chi_\M(\delta_i)\Big)^{\!\tr}\frac{\partial\delta_i}{\partial\tilde m_i},\\
\delta_i=\chi_\M^{-1}\big(\chi_\M(\tilde m_i)^{-1}\star z_i\big),
\end{gathered}
\end{equation}
where the chart Jacobian $\partial\delta_i/\partial\tilde m_i$ is $\big(\begin{smallmatrix}-I_2&0\\ \ast&-1\end{smallmatrix}\big)$ in $\Nil$ and is evaluated by central differences in the other geometries (Supplement~\ref{app:details}) and $\nabla_\delta\log J_\M$ follows from \eqref{eq:jac}. Algorithm~\ref{alg:bbvi} gives the scheme: $n_s$ samples per iteration, Rao--Blackwellised per-node gradients (only the terms of $\log\tilde\pi^\star$ involving $z_i$ multiply the score of $q_i$), the control variate of \citet{ranganath2014black} (its coefficient is estimated from the same draws, which introduces a finite-sample bias of order $1/n_s$), and RMSProp-scaled steps with the accumulator updated before it is used. The scale and centre factors are updated with the same score-function estimator, against the inverse-Gamma prior and the Jacobian $J^\chi_\M(m)$ respectively; with a learned slope, a Gaussian factor on $\log\beta$ is added. Dyad subsampling with importance weights---each iteration evaluates a uniform random subset of $M$ dyads, scaled by $\binom N2/M$, which is unbiased for the ELBO gradient (Theorem 2 of \citealp{papamichalis2021latent})---brings the per-iteration cost to $O(M+N)$ and is used for the networks with several hundred nodes in Section~\ref{sec:exp-real}; the score bound $\|\nabla^{\M}_{z_i}\ell_{ij}\|\le\beta$ of that paper holds in any Riemannian manifold, hence here, away from cut loci.

\begin{algorithm}[H]
\caption{Black-box variational inference on the gauge slice}\label{alg:bbvi}
\textbf{Input:} $\mathcal{Y}$, anchors, sample size $S$, learning rate $\eta$.\\
\textbf{Initialise:} $\tilde m_i=\chi_\M^{-1}(z_i^{(0)})$, $\tilde s_i=\log 0.3$, $(\tilde m,\tilde s)=(\alpha^{(0)},\log0.2)$, $\tilde m_\mu=0$, $\tilde s_\mu=\log0.2$, $\tilde m_k=\log\sigma_k^{(0)}$, $\tilde s_k=\log0.1$.\\
\textbf{Repeat} until the ELBO stabilises:
\begin{enumerate}
\item Draw $n_s$ samples $(\alpha^{(s)},m^{(s)},\sigma^{(s)},Z^{(s)})\sim q$ by reparameterisation, and evaluate all dyad log-likelihoods $\ell_{ij}^{(s)}$ (or a uniform subsample of $M$ dyads, rescaled).
\item For each factor $\lambda_i\in\{\tilde m_i,\tilde s_i\}$: $f_i^{(s)}=\nabla_{\lambda_i}\log q_i(z_i^{(s)})\,\big[\sum_{j\ne i}\ell^{(s)}_{ij}+\log f_{\LWN}(z_i^{(s)}\mid\mu^{(s)},\Sigma^{(s)})+\log\Delta_\M(z_i^{(s)})\mathbf 1_{i=i_2}-\log q_i(z_i^{(s)})\big]$, $h_i^{(s)}=\nabla_{\lambda_i}\log q_i(z_i^{(s)})$, $\hat a_i=\widehat{\mathrm{Cov}}(f_i,h_i)/\widehat{\mathrm{Var}}(h_i)$, $\hat g_i=S^{-1}\sum_s(f_i^{(s)}-\hat a_ih_i^{(s)})$.
\item Same for $(\tilde m,\tilde s)$ with the full log-likelihood and $\log p(\alpha)$.
\item Same for the scale factors with the terms $\sum_i\log f_{\LWN}(z_i\mid\mu,\Sigma)+\log p(\sigma)-\log q(\sigma)$, and for the centre factor with $\sum_i\log f_{\LWN}(z_i\mid\mu,\Sigma)+\log J^\chi_\M(m)-\log q(m)$. \item $G\leftarrow0.9G+0.1\hat g^2$, $\lambda\leftarrow\lambda+\eta\,\hat g/(\sqrt{G}+10^{-8})$ (elementwise).
\end{enumerate}
\end{algorithm}

\section{Additional experiments}\label{supp:exp}

\subsection{Identifiability experiments}
\paragraph{The threshold.} Proposition~\ref{prop:nonrigid} and Theorem~\ref{thm:generic} predict that the distance matrix determines the gauge-fixed positions locally only from $N=6$. Table~\ref{tab:inverse} tests this without noise: for $12$ random configurations per $N$ we compute the distance matrix and re-estimate the gauge-fixed positions by least squares, from a start perturbed by Gaussian noise of standard deviation $0.3$ and from a random start. For $N\le5$ the perturbed start converges to an exact solution that differs from the truth by $0.08$--$0.25$ (the solution set has positive dimension and the solver stops at its nearest point). For $N\ge6$ the perturbed start recovers the truth to the solver's tolerance (below $10^{-6}$) in $\Nil$ and $\SL$, and in $\Sol$ from $N=7$ (the tabulated distances of $\Sol$ and $\SL$ are piecewise linear, which stalls the solver in some replicates and, in $\Sol$, occasionally stops it at a near-solution away from the truth). At $N=6$ in $\Sol$ the system is square ($15$ equations, $15$ unknowns) and ill-conditioned, giving a median error of $0.04$, a conditioning effect, not the intercept threshold of Theorem~\ref{thm:generic}(ii), which plays no role here because the distances themselves are inverted. Random starts find non-congruent exact solutions for $N\le5$ in $\Nil$ and $\SL$, and still at $N=6$ (two of twelve) and $N=7$ (one of twelve) in $\SL$. None was found from $N=8$ on. Local recovery from $N_0$ and the persistence of distant alternatives past it are both what Theorem~\ref{thm:generic} says. The generation and the recovery use the same distance function, so the experiment tests the theorem's prediction, not the accuracy of the tables.

\begin{table}[H]\centering\small
\caption{Noise-free inverse problem: for $12$ random configurations per $N$ the distance matrix is computed and the gauge-fixed positions are re-estimated by least squares. Entries: median recovery error (root mean square over coordinates, after gauge fixing) from a start perturbed by $\mathcal{N}(0,0.3^2)$ noise / number of random starts (out of 12) that reached an exact solution ($\text{residual}<10^{-4}$, $2\times10^{-3}$ for the tabulated distances) that is not congruent to the truth (error $>0.05$). Exact distances for $\Nil$, tables for $\Sol$ and $\SL$.}\label{tab:inverse}
\resizebox{\textwidth}{!}{\begin{tabular}{lcccccccc}\toprule
$\M\ \backslash\ N$ & 3 & 4 & 5 & 6 & 7 & 8 & 10 & 14\\ \midrule
$\Nil$ & 0.14 / 2 & 0.15 / 3 & 0.13 / 2 & $<10^{-6}$ / 0 & $<10^{-6}$ / 0 & $<10^{-6}$ / 0 & $<10^{-6}$ / 0 & $<10^{-6}$ / 0 \\
$\Sol$ & 0.08 / 1 & 0.12 / 2 & 0.14 / 0 & 0.04 / 0 & $<10^{-6}$ / 0 & $<10^{-6}$ / 0 & $<10^{-6}$ / 0 & $<10^{-6}$ / 0 \\
$\SL$ & 0.25 / 2 & 0.13 / 3 & 0.14 / 2 & $<10^{-6}$ / 2 & $<10^{-6}$ / 1 & $<10^{-6}$ / 0 & $<10^{-6}$ / 0 & $<10^{-6}$ / 0 \\
\bottomrule\end{tabular}}\end{table}

\paragraph{Non-congruent alternatives.} Table~\ref{tab:multi} counts the positions of a further node compatible with its distances to $k$ fixed nodes in general position, found by $40$ restarts of a least-squares solver: with $k=3$, two solutions were found in every configuration in $\Nil$ and $\SL$ and $2.3$ on average in $\Sol$ (restarts establish that at least this many exist). With $k\ge4$, one. In $\R^3$ there are also two, but they are mirror images across the plane of the three fixed nodes and hence congruent. In an anisotropic space the isometry group fixing three generic points is trivial, so the two four-point configurations are different. Figure~\ref{fig:multimodal} shows a two-dimensional section of the population log-likelihood of one node's position given the others at their true positions: two maxima for $N=4$, one for $N=12$. A posterior can be bimodal for a node with few informative dyads. Whether it is depends on the prior and the data.

\begin{table}[H]\centering\small
\caption{Number of positions of a further node compatible with its distances to $k$ fixed nodes (mean over 30 random configurations; in parentheses the percentage of configurations with more than one solution), found by 40 random restarts of a least-squares solver.}\label{tab:multi}
\begin{tabular}{lccc}\toprule
$\M\ \backslash\ k$ & 3 & 4 & 5\\ \midrule
$\Nil$ & 2.0 (100\%) & 1.0 (0\%) & 1.0 (0\%) \\
$\Sol$ & 2.3 (100\%) & 1.0 (0\%) & 1.0 (0\%) \\
$\SL$ & 2.0 (100\%) & 1.0 (0\%) & 1.0 (0\%) \\
\bottomrule\end{tabular}\end{table}

\begin{figure}[H]\centering
\includegraphics[width=0.85\textwidth]{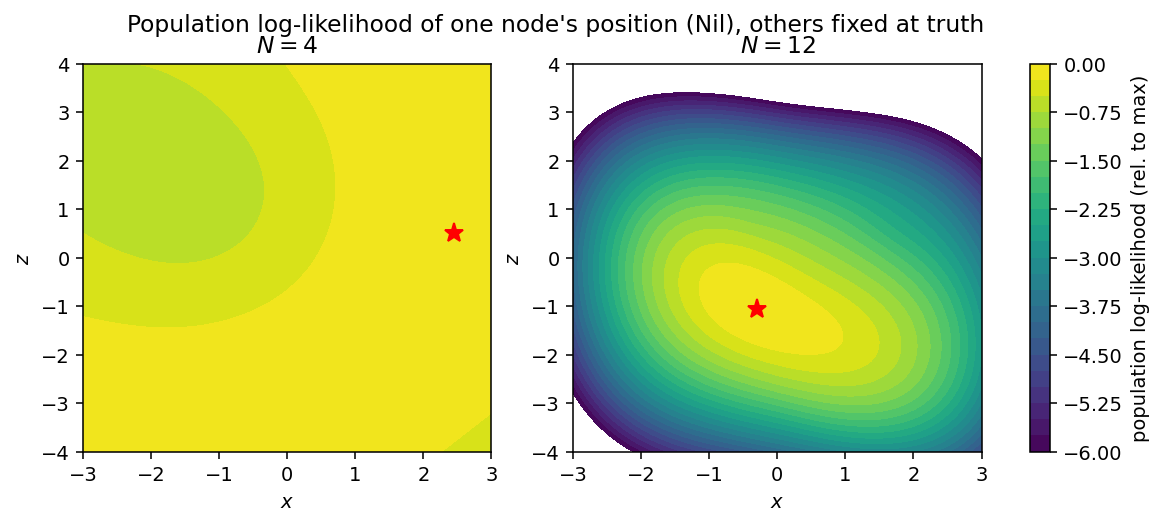}
\caption{Population log-likelihood of the position of one node in $\Nil$ ($y$ coordinate fixed at the truth, other nodes at their true positions, true position starred), for $N=4$ and $N=12$.}\label{fig:multimodal}
\end{figure}

\paragraph{The slice Jacobian.} Figure~\ref{fig:jacobian} compares the posterior of the base radius $\rho_{i_2}$ (written $r_{i_2}$ in the figure) of the second anchor from Algorithm~\ref{alg:mcmc} with and without the factor $\Delta_{\Nil}=r$ on simulated $\Nil$ networks with the scales held fixed: the posterior mean (s.d.) is \jacviiiwith\ with the factor against \jacviiiwithout\ without it for $N=8$, and \jacxlwith\ against \jacxlwithout\ for $N=40$. The base radius $\rho_{i_2}$ (written $r_{i_2}$ in the figure) is the base distance between the projections of the two anchors, an isometry-invariant quantity of the configuration. It is weakly informed by the data and the volume of the rotation orbit tilts its posterior away from the origin. Its effect on other invariant quantities (pairwise distances of non-anchor nodes) has not been measured.

\begin{figure}[H]\centering
\includegraphics[width=0.75\textwidth]{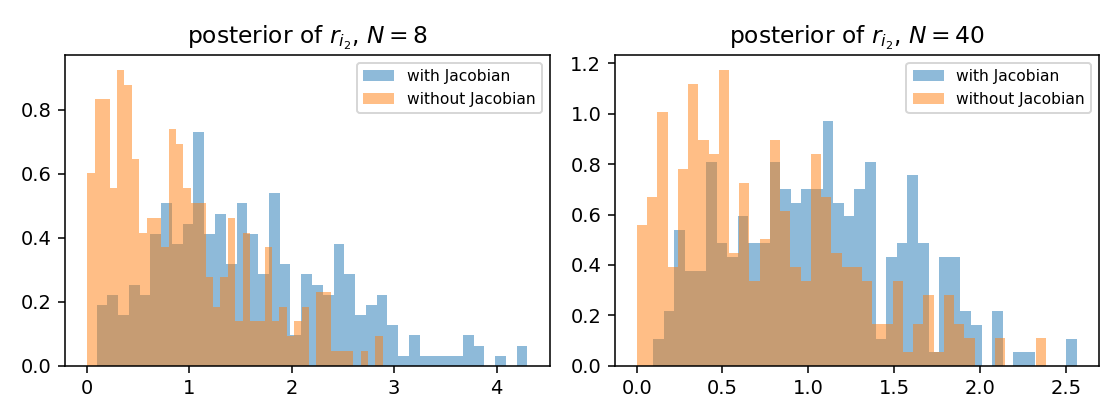}
\caption{Posterior of the base radius $\rho_{i_2}$ (written $r_{i_2}$ in the figure) of the second anchor with and without the slice Jacobian, $\Nil$, $N=8$ and $N=40$ (8\,000 iterations, scales fixed).}\label{fig:jacobian}
\end{figure}

\subsection{Single-realisation simulation}\label{supp:sims}
\begin{table}[H]\centering\scriptsize
\caption{Simulation study ($N=40$, one network per geometry; true positions $\LWN(\eb,\diag(\sigma_h^2,\sigma_h^2,\sigma_v^2))$ with $(\sigma_h,\sigma_v)=$ \simsig, $\alpha=2.5$; scales estimated under the inverse-Gamma prior in both schemes). Columns: edge density; AUC of the true probabilities; $\alpha$ (true; posterior mean (s.d.); variational mean (s.d.)); posterior mean / variational estimate of $(\sigma_h,\sigma_v)$; in-sample AUC; replicate log-loss (oracle: true probabilities); correlation of the posterior mean distance matrix with the truth (MCMC, BBVI) and between the two schemes; wall-clock seconds (MCMC \simiters\ iterations, BBVI 2000 iterations with $n_s=20$).}\label{tab:sims}
\resizebox{\textwidth}{!}{\begin{tabular}{lcccccccccccccccccc}\toprule
 & & & \multicolumn{3}{c}{$\alpha$} & \multicolumn{2}{c}{$(\hat\sigma_h,\hat\sigma_v)$} & \multicolumn{2}{c}{AUC} & \multicolumn{3}{c}{replicate log-loss} & \multicolumn{3}{c}{distance corr.} & \multicolumn{2}{c}{time (s)}\\
\cmidrule(lr){4-6}\cmidrule(lr){7-8}\cmidrule(lr){9-10}\cmidrule(lr){11-13}\cmidrule(lr){14-16}\cmidrule(lr){17-18}
$\M$ & dens. & AUC$_0$ & true & MCMC & BBVI & MCMC & BBVI & MCMC & BBVI & MCMC & BBVI & oracle & MCMC & BBVI & MCMC--BBVI & MCMC & BBVI\\ \midrule
$\Nil$ & 0.23 & 0.85 & 2.5 & 1.87 (0.35) & 1.40 (0.11) & (1.7,1.3) & (1.4,1.1) & 0.91 & 0.88 & 0.444 & 0.451 & 0.389 & 0.82 & 0.80 & 0.94 & 51 & 13 \\
$\Sol$ & 0.34 & 0.79 & 2.5 & 2.06 (0.41) & 1.89 (0.08) & (1.5,1.1) & (1.8,0.6) & 0.89 & 0.84 & 0.585 & 0.587 & 0.521 & 0.72 & 0.68 & 0.92 & 144 & 36 \\
$\SL$ & 0.30 & 0.80 & 2.5 & 2.42 (0.46) & 1.45 (0.09) & (1.6,1.1) & (1.2,0.8) & 0.89 & 0.86 & 0.567 & 0.566 & 0.488 & 0.69 & 0.67 & 0.96 & 80 & 17 \\
\bottomrule\end{tabular}}\end{table}

\begin{figure}[H]\centering
\includegraphics[width=0.9\textwidth]{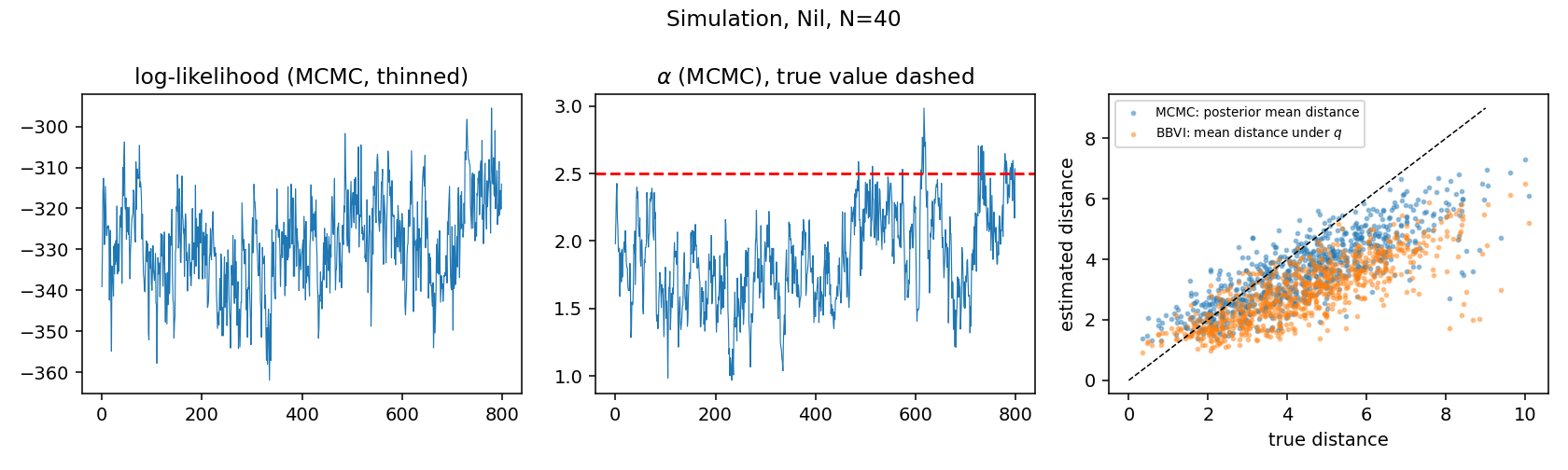}\\[2pt]
\includegraphics[width=0.9\textwidth]{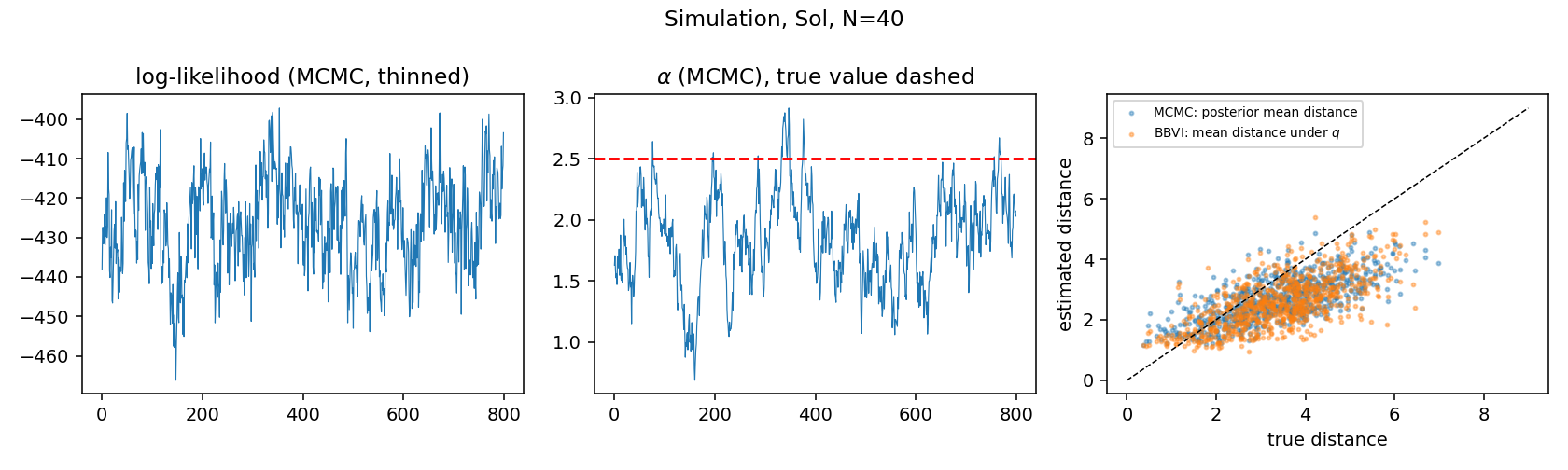}\\[2pt]
\includegraphics[width=0.9\textwidth]{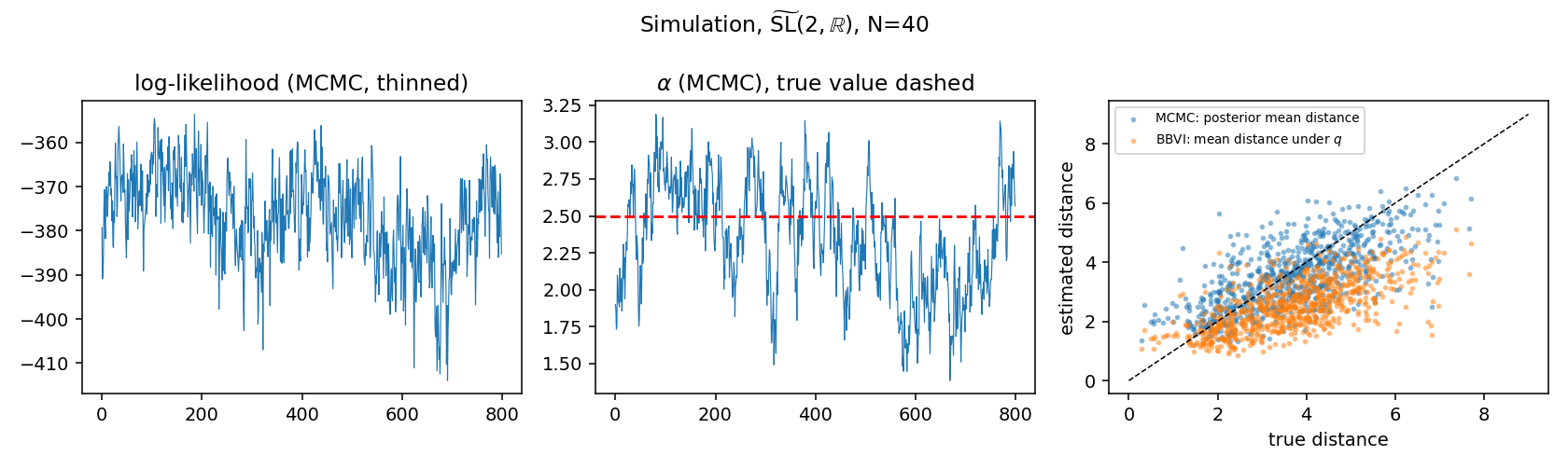}
\caption{Simulations ($N=40$, one realisation per geometry). Left and centre: MCMC traceplots (thinned) of the log-likelihood and of $\alpha$ (true value dashed). Right: true pairwise distances against the posterior mean distances (MCMC) and the mean distances under $q$ (BBVI).}\label{fig:sims}
\end{figure}

\subsection{Coverage over five geometries and as a function of size}
\begin{table}[H]\centering\small
\caption{Replicated coverage study: independent latent configurations and networks ($N=30$; $(\sigma_h,\sigma_v)=(2,2)$ for $\Nil$ and $\mathbb{R}^3$, $(1.8,1.4)$ for $\Sol$, $(1.5,2)$ for $\SL$, $(1.5,1.5)$ for $\mathbb{H}^3$; $\alpha=2.5$; $4000$ MCMC iterations each; the constant-curvature fits are unanchored). Per replicate we record the fraction of true pairwise distances inside the $90\%$ and $50\%$ posterior intervals, the fraction of true tie probabilities inside their $90\%$ intervals, whether $\alpha$ is inside its $90\%$ interval, and the bias and root-mean-square error of the posterior mean distances; entries are means over replicates with Monte Carlo standard errors in parentheses. ``True scales'': $(\sigma_h,\sigma_v)$ fixed at their true values; ``InvGamma'': scales estimated.}\label{tab:covrep}
\resizebox{\textwidth}{!}{\begin{tabular}{llcccccccc}\toprule
$\M$ & scales & networks & 90\% dist. & 50\% dist. & 90\% prob. & $\alpha$ covered & bias & RMSE & $\hat\alpha$\\ \midrule
$\Nil$ & InvGamma$(2,2)$ & 12 & 0.67 (0.05) & 0.30 (0.03) & 0.78 (0.02) & 33\% & -1.21 (0.16) & 1.80 (0.12) & 1.50 (0.12) \\
$\Nil$ & InvGamma$(3,8)$ & 12 & 0.82 (0.04) & 0.41 (0.03) & 0.84 (0.02) & 75\% & -0.58 (0.17) & 1.45 (0.09) & 2.05 (0.13) \\
 & true scales & 12 & 0.87 (0.01) & 0.47 (0.01) & 0.88 (0.01) & 92\% & -0.05 (0.10) & 1.28 (0.04) & 2.49 (0.08) \\
$\Sol$ & InvGamma$(2,2)$ & 6 & 0.75 (0.05) & 0.34 (0.04) & 0.82 (0.02) & 67\% & -0.80 (0.11) & 1.29 (0.07) & 1.67 (0.07) \\
 & true scales & 6 & 0.87 (0.02) & 0.47 (0.02) & 0.88 (0.01) & 83\% & -0.10 (0.13) & 1.00 (0.05) & 2.34 (0.15) \\
$\SL$ & InvGamma$(2,2)$ & 8 & 0.74 (0.06) & 0.35 (0.04) & 0.80 (0.02) & 50\% & -0.89 (0.18) & 1.42 (0.14) & 1.70 (0.16) \\
 & true scales & 6 & 0.87 (0.02) & 0.46 (0.02) & 0.88 (0.01) & 100\% & -0.07 (0.15) & 1.10 (0.06) & 2.44 (0.12) \\
$\mathbb{R}^3$ & InvGamma$(2,2)$ & 10 & 0.65 (0.07) & 0.30 (0.05) & 0.75 (0.03) & 30\% & -1.34 (0.25) & 1.97 (0.19) & 1.43 (0.20) \\
 & true scales & 10 & 0.87 (0.01) & 0.47 (0.01) & 0.87 (0.01) & 90\% & -0.04 (0.13) & 1.35 (0.06) & 2.48 (0.11) \\
$\mathbb{H}^3$ & InvGamma$(2,2)$ & 8 & 0.80 (0.05) & 0.40 (0.06) & 0.85 (0.01) & 62\% & -0.67 (0.21) & 1.22 (0.14) & 1.88 (0.19) \\
 & true scales & 8 & 0.87 (0.02) & 0.49 (0.03) & 0.88 (0.01) & 88\% & -0.04 (0.13) & 0.98 (0.05) & 2.45 (0.13) \\
\bottomrule\end{tabular}}\end{table}

\begin{table}[H]\centering\small
\caption{Coverage in $\Nil$ as a function of network size and scale prior (true scales $(2,2)$, $\alpha=2.5$; columns as in Table~\ref{tab:covrep}). $\mathrm{InvGamma}(2,2)$ has mean $2$, $\mathrm{InvGamma}(3,8)$ has mean $4$, the true horizontal variance.}\label{tab:covN}
\begin{tabular}{llccccccc}\toprule
$N$ & scales & networks & 90\% dist. & 50\% dist. & 90\% prob. & $\alpha$ covered & bias & $\hat\alpha$\\ \midrule
$30$ & InvGamma$(2,2)$ & 12 & 0.67 (0.05) & 0.30 (0.03) & 0.78 (0.02) & 33\% & -1.21 (0.16) & 1.50 (0.12) \\
$30$ & InvGamma$(3,8)$ & 12 & 0.82 (0.04) & 0.41 (0.03) & 0.84 (0.02) & 75\% & -0.58 (0.17) & 2.05 (0.13) \\
$30$ & true scales & 12 & 0.87 (0.01) & 0.47 (0.01) & 0.88 (0.01) & 92\% & -0.05 (0.10) & 2.49 (0.08) \\
$60$ & InvGamma$(3,8)$ & 6 & 0.88 (0.01) & 0.47 (0.02) & 0.86 (0.01) & 100\% & -0.19 (0.05) & 2.33 (0.05) \\
$60$ & true scales & 6 & 0.87 (0.02) & 0.46 (0.02) & 0.87 (0.01) & 100\% & -0.08 (0.04) & 2.43 (0.03) \\
\bottomrule\end{tabular}\end{table}

\subsection{Fisher deficit and single-network win maps}

\begin{table}[H]\centering\small
\caption{Representation deficit of the win-map configurations against each competitor geometry $T$. Each cell gives: the Fisher deficit $\Delta_T^2$ (nats); the minimum exact dyad divergence $\mathrm{KL}_T$ found by direct optimisation (an upper bound on the infimum; nats); $\mathrm{KL}_T$ per dyad, which is the expected held-out log-loss advantage of the true model over the best competitor to first order; and the bound $0.05+\sqrt{\mathrm{KL}_T/2}$ of Theorem~\ref{thm:lecam} on the power of any level-$0.05$ test from one network ($1$ is vacuous).}\label{tab:lecam}
\resizebox{\textwidth}{!}{\begin{tabular}{lllcccc}\toprule
truth & $N$ & knob & dyads & $T=\mathbb{R}^3$ & $T=\mathbb{H}^3$ & $T=\mathbb{H}^2\!\times\!\mathbb{R}$\\ \midrule
$\Nil$ & 40 & 0.5 & 780 & 0.0 / 0.0 / 0.0000 / 0.16 & 0.3 / 0.2 / 0.0003 / 0.40 & 0.0 / 0.0 / 0.0000 / 0.15 \\
$\Nil$ & 40 & 1.0 & 780 & 0.3 / 0.3 / 0.0004 / 0.45 & 1.2 / 1.2 / 0.0016 / 0.83 & 0.2 / 0.2 / 0.0002 / 0.33 \\
$\Nil$ & 40 & 2.0 & 780 & 2.1 / 2.3 / 0.0029 / 1.00 & 2.0 / 2.0 / 0.0026 / 1.00 & 1.2 / 1.2 / 0.0015 / 0.83 \\
$\Nil$ & 40 & 3.0 & 780 & 3.1 / 3.4 / 0.0044 / 1.00 & 2.7 / 3.2 / 0.0042 / 1.00 & 0.9 / 0.9 / 0.0012 / 0.74 \\
$\SL$ & 60 & 0.5 & 1770 & 0.1 / 0.1 / 0.0000 / 0.25 & 0.6 / 0.6 / 0.0003 / 0.60 & 0.1 / 0.1 / 0.0000 / 0.21 \\
$\SL$ & 60 & 1.0 & 1770 & 1.6 / 1.7 / 0.0010 / 0.98 & 4.2 / 4.1 / 0.0023 / 1.00 & 0.5 / 0.4 / 0.0002 / 0.52 \\
$\SL$ & 60 & 2.0 & 1770 & 13.1 / 15.4 / 0.0087 / 1.00 & 10.0 / 10.6 / 0.0060 / 1.00 & 3.3 / 2.9 / 0.0016 / 1.00 \\
$\Sol$ & 60 & 0.3 & 1770 & 1.6 / 1.6 / 0.0009 / 0.94 & 3.9 / 4.1 / 0.0023 / 1.00 & 1.9 / 1.8 / 0.0010 / 1.00 \\
$\Sol$ & 60 & 0.8 & 1770 & 3.0 / 3.0 / 0.0017 / 1.00 & 5.7 / 5.5 / 0.0031 / 1.00 & 2.6 / 2.6 / 0.0015 / 1.00 \\
$\Sol$ & 60 & 1.4 & 1770 & 5.2 / 5.3 / 0.0030 / 1.00 & 4.3 / 4.2 / 0.0024 / 1.00 & 3.3 / 3.4 / 0.0019 / 1.00 \\
$\Sol$ & 60 & 2.0 & 1770 & 9.6 / 10.7 / 0.0060 / 1.00 & 4.4 / 4.4 / 0.0025 / 1.00 & 4.9 / 4.9 / 0.0028 / 1.00 \\
$\Nil$ & 100 & 0.5 & 4950 & 0.2 / 0.2 / 0.0000 / 0.39 & 2.6 / 2.5 / 0.0005 / 1.00 & 0.2 / 0.2 / 0.0000 / 0.40 \\
$\Nil$ & 100 & 1.0 & 4950 & 2.1 / 2.2 / 0.0004 / 1.00 & 9.9 / 9.9 / 0.0020 / 1.00 & 1.6 / 1.6 / 0.0003 / 0.95 \\
$\Nil$ & 100 & 2.0 & 4950 & 11.5 / 12.7 / 0.0026 / 1.00 & 11.6 / 11.6 / 0.0023 / 1.00 & 6.1 / 6.0 / 0.0012 / 1.00 \\
$\Nil$ & 100 & 3.0 & 4950 & 38.2 / 42.9 / 0.0087 / 1.00 & 29.3 / 31.3 / 0.0063 / 1.00 & 75.3 / 55.4 / 0.0112 / 1.00 \\
\bottomrule\end{tabular}}\end{table}

\begin{table}[H]\centering\scriptsize
\caption{Replicated win maps: five independently generated networks per cell (three for the $N=100$ cell; knob as in Table~\ref{tab:winmap}). ``own'': mean replicate log-loss of the fit under the true geometry; the other columns give the \emph{paired} difference (true geometry minus competitor, negative favours the truth) with its standard deviation over networks; ``wins'': networks on which the true geometry had the lowest (or tied lowest) replicate log-loss.}\label{tab:winrep}
\resizebox{\textwidth}{!}{\begin{tabular}{llcccllllc}\toprule
truth & $N$ & knob & networks & oracle & \multicolumn{4}{c}{own log-loss; paired differences to competitors, mean (s.d.)} & wins\\ \midrule
$\Nil$ & 40 & 0.5 & 5 & 0.387 & own: 0.413 & $\mathbb{R}^3$: -0.0013 (0.0031) & $\mathbb{H}^3$: -0.0015 (0.0018) & $\mathbb{H}^2\!\times\!\mathbb{R}$: -0.0007 (0.0016) & 2/5 \\
$\Nil$ & 40 & 1.0 & 5 & 0.370 & own: 0.411 & $\mathbb{R}^3$: -0.0031 (0.0061) & $\mathbb{H}^3$: +0.0022 (0.0030) & $\mathbb{H}^2\!\times\!\mathbb{R}$: +0.0000 (0.0030) & 0/5 \\
$\Nil$ & 40 & 2.0 & 5 & 0.327 & own: 0.389 & $\mathbb{R}^3$: -0.0017 (0.0025) & $\mathbb{H}^3$: -0.0031 (0.0080) & $\mathbb{H}^2\!\times\!\mathbb{R}$: +0.0001 (0.0017) & 1/5 \\
$\Nil$ & 40 & 3.0 & 5 & 0.264 & own: 0.337 & $\mathbb{R}^3$: +0.0027 (0.0046) & $\mathbb{H}^3$: -0.0029 (0.0047) & $\mathbb{H}^2\!\times\!\mathbb{R}$: +0.0003 (0.0012) & 0/5 \\
$\Nil$ & 100 & 3.0 & 3 & 0.268 & own: 0.297 & $\mathbb{R}^3$: -0.0060 (0.0032) & $\mathbb{H}^3$: -0.0052 (0.0035) & $\mathbb{H}^2\!\times\!\mathbb{R}$: -0.0039 (0.0033) & 3/3 \\
$\Sol$ & 60 & 0.3 & 5 & 0.368 & own: 0.406 & $\mathbb{R}^3$: -0.0004 (0.0037) & $\mathbb{H}^3$: +0.0015 (0.0015) & $\mathbb{H}^2\!\times\!\mathbb{R}$: +0.0012 (0.0038) & 0/5 \\
$\Sol$ & 60 & 0.8 & 5 & 0.365 & own: 0.408 & $\mathbb{R}^3$: +0.0005 (0.0036) & $\mathbb{H}^3$: -0.0018 (0.0044) & $\mathbb{H}^2\!\times\!\mathbb{R}$: +0.0004 (0.0035) & 1/5 \\
$\Sol$ & 60 & 1.4 & 5 & 0.360 & own: 0.402 & $\mathbb{R}^3$: +0.0001 (0.0025) & $\mathbb{H}^3$: +0.0009 (0.0035) & $\mathbb{H}^2\!\times\!\mathbb{R}$: +0.0010 (0.0039) & 2/5 \\
$\Sol$ & 60 & 2.0 & 5 & 0.340 & own: 0.395 & $\mathbb{R}^3$: -0.0011 (0.0066) & $\mathbb{H}^3$: +0.0035 (0.0024) & $\mathbb{H}^2\!\times\!\mathbb{R}$: +0.0032 (0.0037) & 0/5 \\
$\SL$ & 60 & 0.5 & 5 & 0.375 & own: 0.410 & $\mathbb{H}^3$: -0.0006 (0.0031) & $\mathbb{H}^2\!\times\!\mathbb{R}$: -0.0014 (0.0023) & $\Nil$: -0.0013 (0.0017) & 0/5 \\
$\SL$ & 60 & 1.0 & 5 & 0.354 & own: 0.401 & $\mathbb{H}^3$: +0.0008 (0.0013) & $\mathbb{H}^2\!\times\!\mathbb{R}$: -0.0002 (0.0016) & $\Nil$: -0.0003 (0.0022) & 1/5 \\
$\SL$ & 60 & 2.0 & 5 & 0.316 & own: 0.368 & $\mathbb{H}^3$: +0.0014 (0.0039) & $\mathbb{H}^2\!\times\!\mathbb{R}$: +0.0023 (0.0008) & $\Nil$: -0.0019 (0.0011) & 0/5 \\
\bottomrule\end{tabular}}\end{table}

\begin{table}[H]\centering\small
\caption{Fisher deficit $\Delta_T^2(Z,\alpha)$ of Proposition~\ref{prop:deficit} for the win-map configurations of Table~\ref{tab:winmap}: the weighted additive-constant distortion of the true distance matrix into each competitor geometry $T$, in nats (a local minimum of a non-convex weighted least-squares problem, hence an upper bound). Knob as in Table~\ref{tab:winmap}; the exact divergences are in Table~\ref{tab:lecam}.}\label{tab:deficit}
\begin{tabular}{llcccccc}\toprule
truth & $N$ & knob & $\alpha$ & dyads & $T=\mathbb{R}^3$ & $T=\mathbb{H}^3$ & $T=\mathbb{H}^2\!\times\!\mathbb{R}$\\ \midrule
$\Nil$ & 40 & 0.5 & -0.2 & 780 & 0.0 & 0.3 & 0.0 \\
$\Nil$ & 40 & 1.0 & 0.4 & 780 & 0.3 & 1.2 & 0.2 \\
$\Nil$ & 40 & 2.0 & 1.7 & 780 & 2.1 & 2.0 & 1.2 \\
$\Nil$ & 40 & 3.0 & 2.5 & 780 & 3.1 & 2.7 & 0.9 \\
$\Nil$ & 100 & 0.5 & 0.0 & 4950 & 0.2 & 2.6 & 0.2 \\
$\Nil$ & 100 & 1.0 & 0.5 & 4950 & 2.1 & 9.9 & 1.6 \\
$\Nil$ & 100 & 2.0 & 1.3 & 4950 & 11.5 & 11.6 & 6.1 \\
$\Nil$ & 100 & 3.0 & 2.6 & 4950 & 38.2 & 29.3 & 75.3 \\
$\SL$ & 60 & 0.5 & 0.0 & 1770 & 0.1 & 0.6 & 0.1 \\
$\SL$ & 60 & 1.0 & 0.8 & 1770 & 1.6 & 4.2 & 0.5 \\
$\SL$ & 60 & 2.0 & 2.0 & 1770 & 13.1 & 10.0 & 3.3 \\
$\Sol$ & 60 & 0.3 & 0.3 & 1770 & 1.6 & 3.9 & 1.9 \\
$\Sol$ & 60 & 0.8 & 0.6 & 1770 & 3.0 & 5.7 & 2.6 \\
$\Sol$ & 60 & 1.4 & 0.9 & 1770 & 5.2 & 4.3 & 3.3 \\
$\Sol$ & 60 & 2.0 & 1.6 & 1770 & 9.6 & 4.4 & 4.9 \\
\bottomrule\end{tabular}\end{table}

\begin{table}[H]\centering\scriptsize
\caption{Win maps: replicate log-loss of BBVI fits (800 iterations, $n_s=10$) under the true geometry and under competing geometries, as the anisotropy of the truth grows. Knob: horizontal spread $\sigma_h$ for $\Nil$ (vertical spread $1.5$) and $\SL$ (vertical spread $2$), vertical spread $\sigma_v$ for $\Sol$ (horizontal spread $1.5$); $\alpha$ chosen for an expected density of $0.15$. The best fit in each row is in bold; ``oracle'' is the log-loss of the true probabilities.}\label{tab:winmap}
\begin{tabular}{llcccllll}\toprule
truth & $N$ & knob & dens. & oracle & \multicolumn{4}{c}{replicate log-loss by fitted geometry}\\ \midrule
$\Nil$ & 40 & 0.5 & 0.16 & 0.399 & $\Nil$: 0.411 & $\mathbb{R}^3$: \textbf{0.409} & $\mathbb{H}^3$: 0.411 & $\mathbb{H}^2\!\times\!\mathbb{R}$: 0.411 \\
$\Nil$ & 40 & 1.0 & 0.17 & 0.382 & $\Nil$: 0.410 & $\mathbb{R}^3$: 0.408 & $\mathbb{H}^3$: \textbf{0.407} & $\mathbb{H}^2\!\times\!\mathbb{R}$: 0.407 \\
$\Nil$ & 40 & 2.0 & 0.16 & 0.287 & $\Nil$: 0.344 & $\mathbb{R}^3$: 0.350 & $\mathbb{H}^3$: \textbf{0.344} & $\mathbb{H}^2\!\times\!\mathbb{R}$: 0.346 \\
$\Nil$ & 40 & 3.0 & 0.15 & 0.273 & $\Nil$: 0.336 & $\mathbb{R}^3$: \textbf{0.333} & $\mathbb{H}^3$: 0.347 & $\mathbb{H}^2\!\times\!\mathbb{R}$: 0.337 \\
$\Nil$ & 100 & 0.5 & 0.15 & 0.367 & $\Nil$: 0.392 & $\mathbb{R}^3$: 0.390 & $\mathbb{H}^3$: \textbf{0.389} & $\mathbb{H}^2\!\times\!\mathbb{R}$: 0.389 \\
$\Nil$ & 100 & 1.0 & 0.15 & 0.373 & $\Nil$: 0.403 & $\mathbb{R}^3$: 0.403 & $\mathbb{H}^3$: \textbf{0.400} & $\mathbb{H}^2\!\times\!\mathbb{R}$: 0.401 \\
$\Nil$ & 100 & 2.0 & 0.15 & 0.323 & $\Nil$: 0.352 & $\mathbb{R}^3$: 0.352 & $\mathbb{H}^3$: \textbf{0.350} & $\mathbb{H}^2\!\times\!\mathbb{R}$: 0.351 \\
$\Nil$ & 100 & 3.0 & 0.14 & 0.280 & $\Nil$: \textbf{0.308} & $\mathbb{R}^3$: 0.317 & $\mathbb{H}^3$: 0.314 & $\mathbb{H}^2\!\times\!\mathbb{R}$: 0.313 \\
$\Sol$ & 60 & 0.3 & 0.15 & 0.354 & $\Sol$: 0.400 & $\mathbb{R}^3$: 0.399 & $\mathbb{H}^3$: 0.399 & $\mathbb{H}^2\!\times\!\mathbb{R}$: \textbf{0.398} \\
$\Sol$ & 60 & 0.8 & 0.14 & 0.360 & $\Sol$: 0.404 & $\mathbb{R}^3$: 0.404 & $\mathbb{H}^3$: \textbf{0.403} & $\mathbb{H}^2\!\times\!\mathbb{R}$: 0.403 \\
$\Sol$ & 60 & 1.4 & 0.14 & 0.368 & $\Sol$: 0.405 & $\mathbb{R}^3$: 0.406 & $\mathbb{H}^3$: \textbf{0.402} & $\mathbb{H}^2\!\times\!\mathbb{R}$: 0.403 \\
$\Sol$ & 60 & 2.0 & 0.14 & 0.322 & $\Sol$: 0.386 & $\mathbb{R}^3$: 0.393 & $\mathbb{H}^3$: \textbf{0.381} & $\mathbb{H}^2\!\times\!\mathbb{R}$: 0.382 \\
$\SL$ & 60 & 0.5 & 0.14 & 0.379 & $\SL$: 0.411 & $\mathbb{H}^3$: \textbf{0.411} & $\mathbb{H}^2\!\times\!\mathbb{R}$: 0.412 & $\Nil$: 0.411 \\
$\SL$ & 60 & 1.0 & 0.16 & 0.360 & $\SL$: 0.409 & $\mathbb{H}^3$: 0.409 & $\mathbb{H}^2\!\times\!\mathbb{R}$: 0.409 & $\Nil$: \textbf{0.406} \\
$\SL$ & 60 & 2.0 & 0.14 & 0.332 & $\SL$: 0.385 & $\mathbb{H}^3$: 0.387 & $\mathbb{H}^2\!\times\!\mathbb{R}$: \textbf{0.382} & $\Nil$: 0.387 \\
\bottomrule\end{tabular}\end{table}

\paragraph{Profile over the twist in $\SL$.} A profile of held-out loss over $\tau\in\{0,0.5,1,2\}$ on two single dense networks of $60$ nodes (code accompanying the paper) did not separate $\tau\in\{0,0.5,1\}$. We do not report it as evidence, since it has no replication and no Monte Carlo error.

\subsection{Learned slope and the annealed estimator}
\begin{table}[H]\centering\small
\caption{Learned slope: variational fits of $\logit p_{ij}=\alpha-\beta\,d_\M(z_i,z_j)$ with variational factors for $\alpha$, $\log\beta$, the scales and the centre (unanchored, 800 iterations, $n_s=10$), on the first five held-out splits; held-out log-loss, AUC, posterior means of $\beta$ and $\alpha$, and the unit-slope log-loss on the same splits. $\mathrm{PSL}(2,\mathbb{R})$ is the compact-fibre quotient of $\SL$ (fibre period $2\pi$).}\label{tab:beta}
\begin{tabular}{llccccc}\toprule
network & $\M$ & log-loss & AUC & $\hat\beta$ & $\hat\alpha$ & log-loss at $\beta=1$\\ \midrule
Karate club & $\mathbb{R}^3$ & 0.329 & 0.781 & 0.92 & 0.86 & 0.330 \\
 & $\mathbb{H}^3$ & 0.343 & 0.839 & 0.66 & 0.55 & 0.321 \\
 & $\mathbb{H}^2\!\times\!\mathbb{R}$ & 0.324 & 0.839 & 0.85 & 0.80 & 0.325 \\
 & $\Nil$ & 0.337 & 0.784 & 1.01 & 0.79 & 0.337 \\
 & $\SL$ & 0.332 & 0.820 & 0.82 & 0.60 & 0.333 \\
 & $\mathrm{PSL}(2,\mathbb{R})$ & 0.331 & 0.811 & 0.90 & 0.62 & -- \\
 & $\Sol$ & 0.340 & 0.782 & 0.97 & 0.53 & 0.345 \\
Les Mis\'erables & $\mathbb{R}^3$ & 0.467 & 0.908 & 2.17 & 3.02 & 0.170 \\
 & $\mathbb{H}^3$ & 0.582 & 0.921 & 2.27 & 3.61 & 0.164 \\
 & $\mathbb{H}^2\!\times\!\mathbb{R}$ & 0.447 & 0.928 & 2.02 & 3.57 & 0.161 \\
 & $\Nil$ & 0.536 & 0.931 & 2.30 & 3.30 & 0.163 \\
 & $\SL$ & 0.512 & 0.937 & 2.14 & 3.58 & 0.157 \\
 & $\mathrm{PSL}(2,\mathbb{R})$ & 0.624 & 0.934 & 2.39 & 3.51 & -- \\
\bottomrule\end{tabular}\end{table}

\begin{table}[H]\centering\small
\caption{Annealed maximum likelihood with learned $(R,T)$, in the style of \citet{celinska2024thurston} (our implementation: 60 annealing sweeps over positions with $(\alpha,\beta)=(R/T,1/T)$ re-fitted by logistic regression after each sweep), on the first three held-out splits: held-out log-loss, AUC, fitted $\beta$ and $\alpha$, and the held-out log-loss of the Bayesian unit-slope fit on the same splits (not available for the Cora ball, whose Bayesian fits used other splits).}\label{tab:anneal}
\begin{tabular}{llccccc}\toprule
network & $\M$ & log-loss & AUC & $\hat\beta$ & $\hat\alpha$ & Bayesian, $\beta=1$\\ \midrule
Karate club & $\mathbb{R}^3$ & 1.5346 & 0.730 & 6.7 & 15.5 & 0.3124 \\
 & $\mathbb{H}^3$ & 0.7064 & 0.863 & 5.4 & 22.6 & 0.3011 \\
 & $\mathbb{H}^2\!\times\!\mathbb{R}$ & 0.9628 & 0.861 & 6.0 & 21.9 & 0.3061 \\
 & $\Nil$ & 0.9907 & 0.689 & 4.2 & 11.2 & 0.3186 \\
 & $\SL$ & 1.0805 & 0.800 & 5.6 & 20.8 & 0.3149 \\
 & $\Sol$ & 0.9968 & 0.761 & 6.0 & 16.8 & 0.3242 \\
Les Mis\'erables & $\mathbb{R}^3$ & 0.3523 & 0.896 & 5.2 & 9.5 & 0.1580 \\
 & $\mathbb{H}^3$ & 0.2791 & 0.914 & 5.3 & 17.4 & 0.1579 \\
 & $\mathbb{H}^2\!\times\!\mathbb{R}$ & 0.3488 & 0.920 & 6.0 & 14.8 & 0.1538 \\
 & $\Nil$ & 0.3386 & 0.907 & 5.2 & 10.8 & 0.1478 \\
 & $\SL$ & 0.3457 & 0.910 & 5.3 & 14.6 & 0.1470 \\
 & $\Sol$ & 0.2968 & 0.921 & 5.2 & 11.3 & 0.1579 \\
Cora ball & $\mathbb{R}^3$ & 0.0491 & 0.789 & 2.8 & 0.7 & 0.0378 \\
 & $\mathbb{H}^3$ & 0.0388 & 0.804 & 1.1 & 3.1 & 0.0373 \\
 & $\mathbb{H}^2\!\times\!\mathbb{R}$ & 0.0438 & 0.791 & 1.6 & 0.6 & 0.0371 \\
 & $\Nil$ & 0.0493 & 0.766 & 2.5 & 0.6 & 0.0374 \\
 & $\SL$ & 0.0449 & 0.778 & 1.7 & 0.7 & 0.0369 \\
 & $\Sol$ & 0.0480 & 0.754 & 2.3 & 0.7 & 0.0379 \\
\bottomrule\end{tabular}\end{table}

\subsection{Real directed networks with hierarchical asymmetries}\label{supp:realdir}

This section reports the real directed networks on which the circulation models lose or tie. Provenance: the edge lists of the $21$ connectomes of Allard and Serrano (2020), the collection used by \citet{celinska2024thurston}, of which nine are directed; the Cook et al.\ (2019) chemical connectomes of both sexes of \emph{C. elegans} from Netzschleuder; fifteen seasons of the five largest European football leagues and the Champions and Europa League from openfootball (CC0); world trade net flows from CEPII BACI 2023 (Etalab licence); Dota 2 hero matchups from OpenDota. All are public.

Table~\ref{tab:rotfrac} gives the rotational fraction of the asymmetry, the share of $Y-Y^{\tr}$ that lies in the cycle space of the complete graph and that no ranking reproduces, with its value when the directions of the asymmetric pairs are assigned at random. In every directed connectome the fraction is below the random-direction value ($0.75$--$0.98$ against $0.96$--$0.99$): the observed asymmetries are more transitive than random, which is the signature of a hierarchy, not of circulation. Table~\ref{tab:directedreal} shows the consequence. On the five cortical networks the gradient model on $\mathbb{H}^2\times\mathbb{R}$ has the lowest or tied lowest held-out log-loss on four (Cat1, Cat2, Macaque1, Macaque2) and $\mathbb{R}^3$ on Cat3, with direction accuracies $0.77$--$0.91$ where the asymmetry is predictable at all; the $\Nil$ circulation model never wins, its estimated $|\gamma|$ is $0.00$--$0.18$, and the $\SL$ circulation model, whose vertical offset contains a ranking term, ties the product on Macaque1 ($0.251$, accuracy $0.89$ against $0.90$) and is within $0.006$ elsewhere. On the \emph{C. elegans} network no model predicts the direction of asymmetric pairs above $0.58$, and the log-losses are within $0.004$ of each other. Cortical asymmetries are laminar hierarchies, the asymmetry of the worm connectome is not predictable from positions at this size, and in neither case is there circulation to model: the win condition of Section~\ref{sec:exp-directed} is not met on these data, and the recommended use of the circulation model is restricted to networks whose rotational fraction exceeds the random-direction value.

\begin{table}[H]\centering\small
\caption{Directed connectomes: number of nodes, density of directed ties, fraction of node pairs with exactly one direction present, and the rotational fraction of the asymmetry (share of the squared norm of $Y-Y^{\tr}$ in the cycle space of the complete graph, i.e.\ the part that no ranking reproduces), with its value when the direction of every asymmetric pair is assigned at random (mean of $20$ draws). A rotational fraction below the random-direction value means the asymmetry is more transitive than random. Sources: Allard and Serrano (2020) edge lists for CElegans (Varshney et al.\ 2011), Cat1--3 (Scannell et al.\ 1999) and Macaque1--3 (CoCoMac); Cook et al.\ (2019) hermaphrodite and male chemical connectomes from Netzschleuder (neurons only for the hermaphrodite; all cells for the male).}\label{tab:rotfrac}
\begin{tabular}{lccccc}\toprule
network & $N$ & density & asymmetric pairs & rotational fraction & random directions\\ \midrule
CElegans & 279 & 0.039 & 0.041 & 0.97 & 0.99 \\
Cat1 & 65 & 0.274 & 0.154 & 0.85 & 0.97 \\
Cat2 & 95 & 0.238 & 0.048 & 0.93 & 0.98 \\
Cat3 & 52 & 0.308 & 0.160 & 0.87 & 0.96 \\
Macaque1 & 94 & 0.273 & 0.146 & 0.75 & 0.98 \\
Macaque2 & 71 & 0.150 & 0.052 & 0.92 & 0.97 \\
Macaque3 & 242 & 0.070 & 0.069 & 0.93 & 0.99 \\
CElegans2019\_herm & 338 & 0.033 & 0.043 & 0.97 & 0.99 \\
CElegans2019\_male & 575 & 0.016 & 0.023 & 0.98 & 1.00 \\
\bottomrule\end{tabular}\end{table}

\begin{table}[H]\centering\scriptsize
\caption{Directed connectomes, held-out dyads ($10\%$ of node pairs hidden in both directions; number of splits in the third column): held-out directed log-loss / accuracy of the predicted direction on asymmetric held-out pairs / $|\hat\gamma|$, for the circulation model in $\Nil$ and in $\SL$, the gradient models in $\mathbb{H}^2\times\mathbb{R}$ and $\mathbb{R}^3$, and the symmetric $\Nil$ model (mechanism control). The circulation fits are restarted at $\gamma_0\in\{0,0.4\}$ and the restart with the lower training log-loss is reported.}\label{tab:directedreal}
\resizebox{\textwidth}{!}{\begin{tabular}{lccccccc}\toprule
network & $N$ & splits & $\Nil$ circulation & $\SL$ circulation & $\mathbb{H}^2\times\mathbb{R}$ gradient & $\mathbb{R}^3$ gradient & $\Nil$ symmetric\\ \midrule
Cat3 & 52 & 2 & 0.385 / 0.72 / 0.03 & 0.387 / 0.71 / 0.18 & 0.381 / 0.80 / 0.23 & 0.375 / 0.77 / 0.14 & 0.393 / 0.52 / 0.00 \\
Cat1 & 65 & 2 & 0.355 / 0.75 / 0.10 & 0.348 / 0.81 / 0.25 & 0.344 / 0.83 / 0.33 & 0.353 / 0.79 / 0.23 & 0.357 / 0.53 / 0.00 \\
Macaque2 & 71 & 2 & 0.275 / 0.43 / 0.01 & 0.279 / 0.50 / 0.02 & 0.275 / 0.54 / 0.03 & 0.275 / 0.79 / 0.03 & 0.275 / 0.54 / 0.00 \\
Macaque1 & 94 & 2 & 0.280 / 0.80 / 0.18 & 0.251 / 0.89 / 0.53 & 0.251 / 0.90 / 0.91 & 0.283 / 0.91 / 0.81 & 0.283 / 0.47 / 0.00 \\
Cat2 & 95 & 2 & 0.351 / 0.46 / 0.00 & 0.347 / 0.61 / 0.01 & 0.342 / 0.57 / 0.02 & 0.345 / 0.51 / 0.02 & 0.350 / 0.36 / 0.00 \\
CElegans & 279 & 1 & 0.127 / 0.54 / 0.05 & 0.127 / 0.57 / 0.09 & 0.124 / 0.57 / 0.13 & 0.127 / 0.53 / 0.06 & 0.128 / 0.44 / 0.00 \\
\bottomrule\end{tabular}}\end{table}

Table~\ref{tab:annealmap} repeats the comparison under the neighbourhood criterion of the embedding literature, mean average precision of the ranking of nodes by fitted distance, for our annealed maximum-likelihood embedding with learned temperature and for the Bayesian fit. With the annealed embedding the anisotropic geometries reach the top on three networks, by margins at the level of the third decimal: $\Nil$ on Cat1 ($0.917$ against $0.916$ for $\mathbb{R}^3$) and, by training likelihood, on Macaque4, and $\SL$ on Macaque2 ($0.862$ against $0.860$). Elsewhere $\mathbb{H}^3$, $\mathbb{H}^2\times\mathbb{R}$ or $\mathbb{R}^3$ is best. Our $\mathbb{H}^3$ values are within $0.01$--$0.03$ of those of \citet{celinska2024thurston}, whereas our $\Sol$ values are $0.02$--$0.09$ below theirs, so the $\Sol$ advantage they report is not reproduced by our estimator on their data; the most probable cause is the surrogate $\Sol$ distance and the annealer used here, and the paper therefore does not claim to have refuted their finding, only that it does not transfer to our estimator or to held-out prediction.

\begin{table}[H]\centering\scriptsize
\caption{Neighbourhood criterion on eight connectomes: mean average precision of the ranking of the other nodes by fitted distance (higher is better), for the annealed maximum-likelihood embedding with learned radius and temperature (first line of each network, three restarts, the one with the lowest training loss) and for the Bayesian unit-slope fit (second line, posterior mean distances). Best geometry per line in bold. Last column: the mean average precision reported by \citet{celinska2024thurston} for $\Sol$ and for $\mathbb{H}^3$ with their annealer.}\label{tab:annealmap}
\resizebox{\textwidth}{!}{\begin{tabular}{lcccccccc}\toprule
network & $N$ / fit & $\Nil$ & $\Sol$ & $\SL$ & $\mathbb{R}^3$ & $\mathbb{H}^3$ & $\mathbb{H}^2\!\times\!\mathbb{R}$ & CKK $\Sol$ / $\mathbb{H}^3$\\ \midrule
Macaque4 & 31 & 0.998 & 0.998 & 0.996 & \textbf{0.999} & 0.996 & 0.997 & 1.000 / 1.000 \\
 & Bayes & \textbf{0.985} & 0.979 & 0.983 & 0.981 & 0.983 & 0.983 &  \\
Cat3 & 52 & 0.944 & 0.939 & 0.944 & 0.940 & 0.936 & \textbf{0.945} & 0.956 / 0.954 \\
 & Bayes & 0.904 & 0.916 & 0.906 & \textbf{0.927} & 0.902 & 0.920 &  \\
Cat1 & 65 & \textbf{0.917} & 0.906 & 0.914 & 0.916 & 0.908 & 0.910 & 0.937 / 0.941 \\
 & Bayes & 0.888 & 0.883 & 0.883 & \textbf{0.907} & 0.886 & 0.882 &  \\
Macaque2 & 71 & 0.860 & 0.844 & \textbf{0.862} & 0.844 & 0.859 & 0.860 & 0.889 / 0.896 \\
 & Bayes & 0.807 & 0.786 & 0.801 & \textbf{0.810} & 0.796 & 0.801 &  \\
Macaque1 & 94 & 0.884 & 0.889 & 0.928 & 0.871 & \textbf{0.937} & 0.919 & 0.974 / 0.962 \\
 & Bayes & 0.865 & 0.878 & 0.910 & 0.846 & \textbf{0.930} & 0.904 &  \\
Cat2 & 95 & 0.796 & 0.790 & 0.825 & 0.786 & \textbf{0.849} & 0.830 & 0.883 / 0.874 \\
 & Bayes & 0.783 & 0.795 & 0.841 & 0.783 & \textbf{0.861} & 0.848 &  \\
Human7 & 110 & 0.934 & 0.904 & 0.922 & 0.937 & 0.931 & \textbf{0.938} & 0.935 / 0.926 \\
 & Bayes & 0.862 & 0.809 & 0.849 & \textbf{0.923} & 0.835 & 0.883 &  \\
Human6 & 116 & 0.946 & 0.909 & 0.936 & \textbf{0.947} & 0.926 & 0.943 & 0.944 / 0.931 \\
 & Bayes & 0.858 & 0.852 & 0.855 & \textbf{0.936} & 0.846 & 0.893 &  \\
\bottomrule\end{tabular}}\end{table}

\paragraph{Contest networks.} Sports results are the setting in which intransitive triads are expected. Table~\ref{tab:football} fits the circulation and gradient models to dominance networks built from fifteen seasons of the five largest European football leagues (openfootball, CC0): $39$--$54$ clubs per league, $48$--$67\%$ of club pairs with a decided dominance. The rotational fraction of the asymmetry is again below its random-direction value ($0.63$--$0.70$ against $0.95$--$0.96$), so results are more transitive than random, as a strength ranking predicts. Over three held-out splits the circulation model in $\Nil$ ties the gradient models on Spain, France and Germany (paired differences within $\pm0.004$, standard deviations $0.01$--$0.06$) and loses on Italy and England ($+0.028$ and $+0.022$); the fitted circulation coefficients are large ($|\hat\gamma|\approx1.2$--$1.7$) because the model uses the vertical coordinate for the ranking, and the symmetric control fails ($0.46$--$0.58$), but no league shows circulation beyond what a ranking explains. Two of the three ties looked like wins of $0.01$--$0.03$ with two splits, which is why the table reports three and the paired standard deviation.

\begin{table}[H]\centering\scriptsize
\caption{Football dominance networks (openfootball \texttt{football.json}, CC0; seasons 2010--11 to 2024--25 where available, $66$ season files): nodes are clubs, $Y_{ij}=1$ if $i$ won more of its matches against $j$ than it lost. Columns: number of clubs; fraction of club pairs with a decided dominance; rotational fraction of the asymmetry / its random-direction value; held-out splits; held-out directed log-loss / direction accuracy for the circulation models in $\Nil$ and $\SL$, the gradient models, and the symmetric $\Nil$ control; paired difference of the $\Nil$ circulation model to the better gradient model (negative favours circulation), mean (s.d.) over splits.}\label{tab:football}
\resizebox{\textwidth}{!}{\begin{tabular}{lcccccccccc}\toprule
league & $N$ & decided pairs & rot.\ fraction / random & splits & $\Nil$ circ. & $\SL$ circ. & $\mathbb{H}^2\times\mathbb{R}$ grad. & $\mathbb{R}^3$ grad. & $\Nil$ sym. & $\Nil$ $-$ best gradient\\ \midrule
Spain (La Liga) & 41 & 0.67 & 0.64 / 0.95 & 3 & 0.450 / 0.78 & 0.452 / 0.76 & 0.454 / 0.74 & 0.452 / 0.75 & 0.572 / 0.40 & -0.001 (0.054) \\
France (Ligue 1) & 39 & 0.63 & 0.70 / 0.95 & 3 & 0.480 / 0.76 & 0.490 / 0.74 & 0.479 / 0.76 & 0.482 / 0.76 & 0.584 / 0.48 & +0.001 (0.021) \\
Germany (Bundesliga) & 39 & 0.60 & 0.66 / 0.95 & 3 & 0.409 / 0.80 & 0.413 / 0.83 & 0.407 / 0.78 & 0.408 / 0.77 & 0.543 / 0.49 & +0.005 (0.011) \\
Italy (Serie A) & 46 & 0.56 & 0.63 / 0.96 & 3 & 0.413 / 0.81 & 0.410 / 0.81 & 0.386 / 0.81 & 0.388 / 0.81 & 0.529 / 0.48 & +0.028 (0.013) \\
England (Premier League) & 54 & 0.48 & 0.69 / 0.96 & 3 & 0.367 / 0.78 & 0.367 / 0.79 & 0.344 / 0.81 & 0.348 / 0.81 & 0.455 / 0.44 & +0.023 (0.028) \\
\bottomrule\end{tabular}}\end{table}

Table~\ref{tab:flows} adds the two directed networks with the largest cyclic content we could obtain: the European club network, in which the Champions League and Europa League link the five national leagues into one hierarchy, and the network of net trade flows between $216$ countries in 2023, in which $17\%$ of decided triads are cyclic, the highest fraction of any real network outside the counter networks (Generation 4 OU has $0.21$ and chess $0.20$). In both the gradient model on $\mathbb{H}^2\times\mathbb{R}$ has the lowest held-out log-loss ($0.115$ and $0.538$); the $\SL$ circulation model is second on trade by $0.006$ and on the clubs by $0.001$, the $\Nil$ circulation model has the highest direction accuracy on trade ($0.710$ against $0.704$) at a higher log-loss, and the symmetric control fails on both. Trade imbalances do circulate, but the circulating part is not predictable from latent positions beyond what a size ranking with a Hyperbolic base explains. Across the seventeen directed networks of this section, the rotational fraction of the asymmetry is below its random-direction value in every case, and the gradient models win or tie in every case.

\begin{table}[H]\centering\scriptsize
\caption{Two further directed dominance networks, one held-out split ($10\%$ of pairs hidden in both directions): European clubs (the five leagues of Table~\ref{tab:football} linked by Champions League and Europa League matches, openfootball, CC0; clubs with at least eight decided pairs), and world trade net flows (CEPII BACI 2023, Etalab licence; $Y_{ij}=1$ if the exports of $i$ to $j$ exceed those of $j$ to $i$; countries with at least $100$ million USD of trade). Columns: nodes; decided pairs; rotational fraction of the asymmetry / random-direction value; fraction of cyclic triads among decided triads (random directions give $0.25$); held-out directed log-loss / direction accuracy for the circulation models in $\Nil$ and $\SL$, the gradient models, and the symmetric $\Nil$ control.}\label{tab:flows}
\resizebox{\textwidth}{!}{\begin{tabular}{lccccccccc}\toprule
network & $N$ & decided & rot.\ fraction / random & cyclic triads & $\Nil$ circ. & $\SL$ circ. & $\mathbb{H}^2\times\mathbb{R}$ grad. & $\mathbb{R}^3$ grad. & $\Nil$ sym.\\ \midrule
European clubs, five leagues and cups (2010--25) & 269 & 3490 & 0.93 / 0.99 & -- & 0.120 / 0.49 & 0.116 / 0.59 & 0.115 / 0.69 & 0.121 / 0.51 & 0.120 / 0.48 \\
World trade net flows, BACI 2023 & 216 & 16268 & 0.81 / 0.99 & 0.17 & 0.554 / 0.71 & 0.544 / 0.71 & 0.538 / 0.70 & 0.545 / 0.71 & 0.600 / 0.50 \\
\bottomrule\end{tabular}}\end{table}

\begin{table}[H]\centering\scriptsize
\caption{Matchup networks that are intransitive by design. Pok\'emon: $Y_{ji}=1$ if $j$ is listed as a check or counter of $i$ with score $p-4\sigma>0.5$ in the Smogon usage statistics (January 2024, rating cutoff $1500$) of three tiers of Generation 9, species with usage above $1\%$ ($1.9$ million battles in OU). Dota 2: $Y_{ij}=1$ if hero $i$ won more than half of at least $30$ recorded games against hero $j$ (OpenDota, September 2026). Columns as in Table~\ref{tab:flows}, then held-out splits, held-out log-loss / direction accuracy per model, and the paired differences of the $\Nil$ and the $\SL$ circulation models to the better ranking, mean (s.d.) over splits; negative favours circulation.}\label{tab:games}
\resizebox{\textwidth}{!}{\begin{tabular}{lcccccccccccc}\toprule
network & $N$ & decided & rot.\ fr.\ / random & cyclic triads & splits & $\Nil$ circ. & $\SL$ circ. & $\mathbb{H}^2\times\mathbb{R}$ grad. & $\mathbb{R}^3$ grad. & $\Nil$ sym. & $\Nil-$ranking & $\SL-$ranking\\ \midrule
Pok\'emon OU (Smogon, 2024-01) & 93 & 2809 & 0.72 / 0.98 & 0.12 & 3 & 0.501 / 0.82 & 0.511 / 0.80 & 0.524 / 0.76 & 0.528 / 0.76 & 0.630 / 0.47 & -0.023 (0.007) & -0.013 (0.007) \\
Pok\'emon UU & 104 & 1776 & 0.75 / 0.98 & 0.11 & 3 & 0.352 / 0.82 & 0.340 / 0.82 & 0.343 / 0.81 & 0.348 / 0.81 & 0.409 / 0.54 & +0.009 (0.009) & -0.003 (0.005) \\
Pok\'emon Ubers & 100 & 1203 & 0.74 / 0.98 & 0.06 & 3 & 0.302 / 0.85 & 0.285 / 0.85 & 0.278 / 0.86 & 0.287 / 0.85 & 0.338 / 0.46 & +0.024 (0.004) & +0.007 (0.003) \\
Dota 2 hero matchups (OpenDota) & 120 & 4975 & 0.85 / 0.98 & 0.18 & 2 & 0.563 / 0.67 & 0.559 / 0.67 & 0.552 / 0.68 & 0.558 / 0.68 & 0.608 / 0.49 & +0.011 (0.004) & +0.008 (0.006) \\
\bottomrule\end{tabular}}\end{table}

\begin{table}[H]\centering\small
\caption{Small-scale deficit (Theorem~\ref{thm:smallscale}) against the exact minimum divergence for one $\Nil$ configuration of $N$ points scaled to diameter of order $\varepsilon$ ($\alpha=0.5$; exact $\Nil$ distances). Each cell gives the leading-order prediction $\tfrac12\|\Pi^w r\|_w^2$, the minimum exact divergence found by least squares on deviance residuals from the normal-coordinate configuration, and their ratio (exact over predicted). For $\mathbb{H}^2\times\mathbb{R}$ the prediction is minimised over the orientation of the product structure, and the exact optimiser, which starts from the unrotated configuration, reports local minima above the prediction. ``stresses'': dimension of the space of equilibrium stresses of the complete framework on $N$ points in $\mathbb{R}^3$, orthogonal to the constant vector.}\label{tab:smallscale}
\resizebox{\textwidth}{!}{\begin{tabular}{lllccc}\toprule
$N$ & $\varepsilon$ & stresses & $T=\mathbb{R}^3$: predicted / exact / ratio & $T=\mathbb{H}^3$ & $T=\mathbb{H}^2\!\times\!\mathbb{R}$\\ \midrule
8 & 0.05 & 9 & 1.57e-09 / 1.50e-09 / 0.96 & 2.09e-10 / 1.62e-10 / 0.77 & 2.02e-11 / 1.42e-10 / 7.03 \\
8 & 0.1 & 9 & 1.08e-07 / 9.78e-08 / 0.90 & 1.50e-08 / 1.22e-08 / 0.81 & 1.09e-09 / 2.03e-09 / 1.87 \\
8 & 0.2 & 9 & 7.75e-06 / 6.03e-06 / 0.78 & 8.26e-07 / 9.66e-07 / 1.17 & 8.62e-08 / 1.20e-07 / 1.39 \\
8 & 0.4 & 9 & 5.28e-04 / 2.72e-04 / 0.52 & 6.59e-05 / 6.18e-05 / 0.94 & 5.69e-06 / 9.00e-06 / 1.58 \\
12 & 0.05 & 35 & 2.55e-09 / 2.44e-09 / 0.96 & 7.15e-09 / 7.18e-09 / 1.00 & 3.23e-10 / 6.19e-10 / 1.92 \\
12 & 0.1 & 35 & 1.74e-07 / 1.58e-07 / 0.91 & 4.41e-07 / 4.60e-07 / 1.04 & 1.57e-08 / 2.27e-08 / 1.44 \\
12 & 0.2 & 35 & 1.21e-05 / 9.58e-06 / 0.79 & 2.87e-05 / 2.75e-05 / 0.96 & 1.38e-06 / 1.46e-06 / 1.06 \\
12 & 0.4 & 35 & 7.71e-04 / 4.27e-04 / 0.55 & 1.78e-03 / 1.22e-03 / 0.69 & 8.20e-05 / 7.80e-05 / 0.95 \\
\bottomrule\end{tabular}}\end{table}

\begin{table}[H]\centering\scriptsize
\caption{Standard directed benchmark networks, restricted to the most active nodes, three held-out splits: chess game results among the $150$ players with the most games ($i\to j$ if $i$ won more of their games than it lost, draws ignored); Wikipedia adminship votes and institutional e-mail among the $150$ nodes of highest degree ($i\to j$ if $i$ voted for, or e-mailed, $j$). The last two rows are the two largest pairwise competition tournaments of the Dryad archive of Soliveres et al.\ (2018), fifteen percent of pairs held out, three splits. The remaining rows are two further months of the National Dex counter network, two directed citation balls (breadth-first, $200$ papers, $i\to j$ if $i$ cites $j$) and the political blog network ($i\to j$ if $i$ links to $j$). Columns as in Table~\ref{tab:gamesall}.}\label{tab:bench}
\resizebox{\textwidth}{!}{\begin{tabular}{lcccccccccccccc}\toprule
network & $N$ & decided & density & cyclic & splits & $\Nil$ circ. & $\SL$ circ. & $\mathbb{H}^2\!\times\!\mathbb{R}$ & $\mathbb{R}^3$ & $\mathbb{H}^3$ & $\mathbb{S}^3$ & $\Sol$ & $\Nil-$ranking & $\SL-$ranking\\ \midrule
Chess, top 150 players by games (KONECT/Kaggle 2010) & 150 & 1560 & 0.14 & 0.20 & 3 & 0.242 & 0.240 & 0.237 & 0.238 & \textbf{0.236} & 0.241 & 0.238 & +0.006 (0.001) & +0.004 (0.002) \\
Wikipedia adminship votes, 150 most active (SNAP wiki-Vote) & 150 & 3863 & 0.35 & 0.03 & 3 & 0.344 & 0.335 & \textbf{0.324} & 0.330 & 0.327 & 0.373 & 0.344 & +0.020 (0.005) & +0.011 (0.005) \\
European institution e-mail, 150 most active (SNAP email-Eu-core) & 150 & 3669 & 0.33 & 0.10 & 3 & 0.460 & 0.460 & 0.456 & 0.462 & \textbf{0.448} & 0.514 & 0.453 & +0.012 (0.004) & +0.013 (0.002) \\
Wood-decay fungi, Europe (Soliveres et al.\ 2018), 31 species & 31 & 462 & 0.99 & 0.06 & 3 & 0.512 & 0.425 & 0.379 & 0.379 & 0.382 & 0.392 & \textbf{0.377} & +0.135 (0.061) & +0.047 (0.047) \\
Wood-decay fungi, US, 37 species & 37 & 206 & 0.31 & 0.02 & 3 & 0.282 & 0.273 & 0.269 & 0.276 & \textbf{0.259} & 0.273 & 0.273 & +0.022 (0.002) & +0.014 (0.004) \\
Pok\'emon National Dex, Apr 2024 & 122 & 3129 & 0.42 & 0.09 & 3 & 0.412 & \textbf{0.406} & 0.410 & 0.415 & 0.409 & 0.445 & 0.414 & +0.004 (0.008) & -0.002 (0.003) \\
Pok\'emon National Dex, Jul 2024 & 111 & 2825 & 0.46 & 0.10 & 3 & 0.426 & \textbf{0.422} & 0.425 & 0.431 & 0.423 & 0.455 & 0.426 & +0.004 (0.010) & -0.000 (0.003) \\
arXiv hep-th citations, 200-paper ball (SNAP) & 200 & 2243 & 0.11 & 0.00 & 2 & 0.174 & 0.155 & 0.143 & 0.150 & \textbf{0.142} & 0.160 & 0.148 & +0.032 (0.001) & +0.013 (0.000) \\
arXiv hep-ph citations, 200-paper ball (SNAP) & 200 & 980 & 0.05 & 0.00 & 1 & 0.100 & 0.089 & \textbf{0.080} & 0.084 & 0.082 & 0.097 & 0.087 & +0.019 (0.000) & +0.009 (0.000) \\
Political blogs, 150 most linked (Adamic and Glance 2005) & 150 & 3448 & 0.31 & 0.01 & 2 & 0.312 & 0.305 & \textbf{0.293} & 0.294 & 0.295 & 0.341 & 0.297 & +0.021 (0.013) & +0.014 (0.001) \\
European institution e-mail, 250 most active & 250 & 6566 & 0.21 & 0.06 & 2 & 0.353 & 0.344 & 0.341 & 0.356 & \textbf{0.338} & 0.411 & 0.346 & +0.015 (0.003) & +0.006 (0.001) \\
\bottomrule\end{tabular}}\end{table}

\begin{table}[H]\centering\scriptsize
\caption{Counter networks with every geometry, three held-out splits each: held-out directed log-loss (mean over splits) of the circulation models in $\Nil$ and $\SL$ and of the ranking (gradient) models on $\mathbb{H}^2\times\mathbb{R}$, $\mathbb{R}^3$, $\mathbb{H}^3$, $\mathbb{S}^3$ and $\Sol$; best per row in bold. Pok\'emon networks from the Smogon usage statistics of 2024 (rating cutoff $1500$ unless stated; species with usage above $1\%$; $j\to i$ if $j$ is a check or counter of $i$ with score $p-4\sigma>0.5$; in the row so labelled the threshold is $0.3$ and pairs decided in both directions are kept). ``decided'': node pairs with a counter relation; ``density'': their share of all pairs; ``cyclic'': fraction of decided triads that are cycles (random directions give $0.25$). Last columns: paired differences of the $\Nil$ and the $\SL$ circulation models to the best ranking, mean (s.d.) over splits; negative favours circulation.}\label{tab:gamesall}
\resizebox{\textwidth}{!}{\begin{tabular}{lcccccccccccccc}\toprule
network & $N$ & decided & density & cyclic & splits & $\Nil$ circ. & $\SL$ circ. & $\mathbb{H}^2\!\times\!\mathbb{R}$ & $\mathbb{R}^3$ & $\mathbb{H}^3$ & $\mathbb{S}^3$ & $\Sol$ & $\Nil-$ranking & $\SL-$ranking\\ \midrule
Gen 9 OU, Jan & 93 & 2809 & 0.66 & 0.12 & 3 & \textbf{0.501} & 0.511 & 0.524 & 0.528 & 0.526 & 0.529 & 0.527 & -0.022 (0.005) & -0.012 (0.008) \\
Gen 9 OU, Jan, mutual checks kept & 93 & 4229 & 0.99 & 0.13 & 3 & \textbf{0.520} & 0.564 & 0.598 & 0.548 & 0.654 & 0.542 & 0.575 & -0.022 (0.023) & +0.021 (0.015) \\
Gen 9 OU, Feb & 84 & 2262 & 0.65 & 0.12 & 3 & \textbf{0.500} & 0.514 & 0.517 & 0.518 & 0.515 & 0.537 & 0.517 & -0.014 (0.002) & -0.001 (0.003) \\
Gen 9 OU, Mar & 83 & 2160 & 0.63 & 0.13 & 3 & \textbf{0.497} & 0.503 & 0.514 & 0.516 & 0.514 & 0.588 & 0.513 & -0.014 (0.017) & -0.008 (0.011) \\
Gen 9 OU, Apr & 87 & 2275 & 0.61 & 0.12 & 3 & \textbf{0.499} & 0.512 & 0.525 & 0.526 & 0.523 & 0.537 & 0.530 & -0.024 (0.025) & -0.011 (0.015) \\
Gen 9 OU, May & 88 & 1903 & 0.50 & 0.11 & 3 & 0.456 & 0.452 & 0.455 & 0.454 & 0.455 & 0.467 & \textbf{0.452} & +0.006 (0.002) & +0.003 (0.005) \\
Gen 9 OU, Jun & 91 & 2318 & 0.57 & 0.11 & 3 & \textbf{0.468} & 0.478 & 0.479 & 0.482 & 0.476 & 0.494 & 0.478 & -0.009 (0.007) & +0.001 (0.004) \\
Gen 9 OU, Jan, cutoff 1695 & 70 & 504 & 0.21 & 0.07 & 3 & 0.312 & 0.303 & 0.296 & 0.302 & \textbf{0.290} & 0.328 & 0.303 & +0.022 (0.003) & +0.013 (0.003) \\
Gen 9 UU & 104 & 1776 & 0.33 & 0.11 & 3 & 0.352 & \textbf{0.340} & 0.343 & 0.348 & 0.340 & 0.381 & 0.346 & +0.012 (0.009) & +0.000 (0.006) \\
Gen 9 Ubers & 100 & 1203 & 0.24 & 0.06 & 3 & 0.302 & 0.285 & 0.278 & 0.287 & \textbf{0.270} & 0.299 & 0.285 & +0.031 (0.007) & +0.014 (0.001) \\
Gen 3 OU & 55 & 663 & 0.45 & 0.19 & 3 & 0.474 & \textbf{0.469} & 0.477 & 0.484 & 0.475 & 0.504 & 0.487 & +0.000 (0.012) & -0.005 (0.016) \\
Gen 4 OU & 57 & 521 & 0.33 & 0.21 & 3 & 0.424 & 0.428 & 0.429 & 0.433 & 0.422 & 0.451 & \textbf{0.422} & +0.003 (0.012) & +0.007 (0.005) \\
Gen 5 OU & 66 & 584 & 0.27 & 0.15 & 3 & 0.335 & 0.324 & 0.322 & 0.329 & \textbf{0.314} & 0.356 & 0.332 & +0.021 (0.003) & +0.010 (0.003) \\
Gen 6 OU & 99 & 1007 & 0.21 & 0.16 & 3 & 0.289 & 0.278 & 0.274 & 0.283 & \textbf{0.267} & 0.309 & 0.288 & +0.022 (0.003) & +0.011 (0.003) \\
Gen 7 OU & 104 & 1511 & 0.28 & 0.11 & 3 & 0.350 & \textbf{0.339} & 0.340 & 0.344 & 0.339 & 1.180 & 0.345 & +0.013 (0.005) & +0.002 (0.006) \\
Gen 8 OU & 87 & 729 & 0.19 & 0.14 & 3 & 0.275 & 0.265 & 0.262 & 0.268 & \textbf{0.255} & 0.285 & 0.266 & +0.020 (0.007) & +0.010 (0.004) \\
Gen 9 LC & 55 & 278 & 0.19 & 0.17 & 3 & 0.306 & 0.296 & 0.298 & 0.307 & \textbf{0.291} & 0.323 & 0.300 & +0.015 (0.000) & +0.005 (0.007) \\
Gen 9 PU & 107 & 877 & 0.15 & 0.12 & 3 & 0.238 & 0.230 & 0.226 & 0.233 & \textbf{0.219} & 0.258 & 0.232 & +0.019 (0.003) & +0.011 (0.002) \\
Gen 9 NU & 86 & 720 & 0.20 & 0.17 & 3 & 0.300 & 0.292 & 0.288 & 0.293 & \textbf{0.285} & 0.315 & 0.297 & +0.015 (0.006) & +0.007 (0.004) \\
Gen 9 RU & 89 & 858 & 0.22 & 0.14 & 3 & 0.300 & 0.284 & 0.275 & 0.285 & \textbf{0.271} & 0.321 & 0.283 & +0.029 (0.005) & +0.013 (0.005) \\
Gen 9 Monotype & 146 & 2301 & 0.22 & 0.13 & 3 & 0.293 & 0.285 & 0.280 & 0.284 & \textbf{0.277} & 0.317 & 0.282 & +0.016 (0.004) & +0.008 (0.004) \\
Gen 9 National Dex & 115 & 2954 & 0.45 & 0.09 & 3 & 0.413 & 0.408 & 0.410 & 0.412 & \textbf{0.407} & 0.457 & 0.411 & +0.006 (0.005) & +0.001 (0.002) \\
\bottomrule\end{tabular}}\end{table}

\begin{table}[H]\centering\scriptsize
\caption{Two controls on the counter networks (three held-out splits): a ranking on $\R^3$ or $\Hb^3$ with a free rank-two antisymmetric term $s_{ij}=u_i^{\tr}Ju_j$ (``skew'', $2N$ node parameters, the blade--chest and disc models of \citet{chen2016predicting} and \citet{balduzzi2018reevaluating}), and a degree-corrected ranking with sender and receiver effects $a_i+b_j$ (``dc''). Held-out directed log-loss, best in bold; last column: paired difference of the $\Nil$ circulation model to the best control, mean (s.d.) over splits.}\label{tab:controls}
\resizebox{\textwidth}{!}{\begin{tabular}{lccccccccccc}\toprule
network & density & cyclic & splits & $\Nil$ circ. & $\SL$ circ. & $\Hb^3$ rank. & $\R^3$ skew & $\Hb^3$ skew & $\R^3$ dc & $\Hb^3$ dc & $\Nil-$best control\\ \midrule
OU Jan & 0.66 & 0.12 & 3 & \textbf{0.501} & 0.511 & 0.526 & 0.523 & 0.518 & 0.528 & 0.530 & -0.014 (0.004) \\
OU Feb & 0.65 & 0.12 & 3 & \textbf{0.500} & 0.514 & 0.515 & 0.508 & 0.506 & 0.524 & 0.529 & -0.006 (0.003) \\
OU Mar & 0.63 & 0.13 & 3 & \textbf{0.497} & 0.503 & 0.514 & 0.524 & 0.519 & 0.521 & 0.521 & -0.012 (0.011) \\
OU Apr & 0.61 & 0.12 & 3 & \textbf{0.499} & 0.512 & 0.523 & 0.519 & 0.513 & 0.522 & 0.521 & -0.010 (0.015) \\
OU Jun & 0.57 & 0.11 & 3 & \textbf{0.468} & 0.478 & 0.476 & 0.488 & 0.479 & 0.475 & 0.478 & -0.006 (0.010) \\
OU Jan, mutual checks & 0.99 & 0.13 & 3 & \textbf{0.520} & 0.564 & 0.654 & 0.539 & 0.609 & 0.542 & 0.553 & -0.014 (0.018) \\
OU May & 0.50 & 0.11 & 3 & 0.456 & 0.452 & 0.455 & 0.457 & 0.454 & \textbf{0.450} & 0.454 & +0.006 (0.001) \\
National Dex & 0.45 & 0.09 & 3 & 0.413 & 0.408 & 0.407 & 0.405 & \textbf{0.398} & 0.406 & 0.447 & +0.014 (0.002) \\
Gen 3 OU & 0.45 & 0.19 & 3 & 0.474 & 0.469 & 0.475 & 0.452 & \textbf{0.442} & 0.490 & 0.493 & +0.032 (0.035) \\
Gen 7 OU & 0.28 & 0.11 & 3 & 0.350 & 0.339 & 0.339 & 0.340 & 0.330 & \textbf{0.329} & 0.331 & +0.022 (0.007) \\
UU & 0.33 & 0.11 & 3 & 0.352 & 0.340 & 0.340 & 0.345 & \textbf{0.336} & 0.484 & 0.340 & +0.017 (0.008) \\
Ubers & 0.24 & 0.06 & 3 & 0.302 & 0.285 & 0.270 & 0.282 & 0.273 & 0.265 & \textbf{0.263} & +0.039 (0.009) \\
LC & 0.19 & 0.17 & 3 & 0.306 & 0.296 & 0.291 & 0.329 & 0.312 & \textbf{0.275} & 0.277 & +0.032 (0.015) \\
RU & 0.22 & 0.14 & 3 & 0.300 & 0.284 & \textbf{0.271} & 0.281 & 0.274 & 0.308 & 0.340 & +0.032 (0.013) \\
\bottomrule\end{tabular}}\end{table}

\begin{table}[H]\centering\scriptsize
\caption{The decisive control and the strongest baseline on eleven counter networks (three held-out splits, $10\%$ of pairs hidden in both directions). ``$\R^3$ shared'': Euclidean symmetric distance on the same three coordinates that carry the $\Nil$ directed term $\gamma(\zeta_j-\zeta_i-\tfrac12(x_iy_j-y_ix_j))$. ``AME'': additive and multiplicative effects, $\alpha+a_i+b_j+u_i^{\tr}v_j$ with $u_i,v_i\in\R^2$ and no distance term \citep{hoff2005bilinear,hoff2021additive}. Also shown: the $\Hb^3$ ranking, the $\Hb^3$ ranking with a free skew term, and the degree-corrected $\R^3$ ranking. Held-out directed log-loss, best in bold; last columns: paired difference of $\Nil$ to the shared-coordinate control and to AME, mean (s.d.) over splits.}\label{tab:controls2}
\resizebox{\textwidth}{!}{\begin{tabular}{lccccccccccc}\toprule
network & density & cyclic & splits & $\Nil$ circ. & $\R^3$ shared & $\Hb^3$ rank. & $\Hb^3$ skew & $\R^3$ dc & AME & $\Nil-$shared & $\Nil-$AME\\ \midrule
OU Jan & 0.66 & 0.12 & 3 & 0.501 & 0.502 & 0.526 & 0.518 & 0.528 & \textbf{0.461} & -0.001 (0.003) & +0.040 (0.013) \\
OU Feb & 0.65 & 0.12 & 3 & 0.500 & 0.499 & 0.515 & 0.506 & 0.524 & \textbf{0.456} & +0.001 (0.011) & +0.045 (0.012) \\
OU Mar & 0.63 & 0.13 & 3 & 0.497 & 0.502 & 0.514 & 0.519 & 0.521 & \textbf{0.449} & -0.005 (0.009) & +0.048 (0.015) \\
OU Apr & 0.61 & 0.12 & 3 & 0.499 & 0.503 & 0.523 & 0.513 & 0.522 & \textbf{0.457} & -0.004 (0.003) & +0.041 (0.003) \\
OU Jun & 0.57 & 0.11 & 3 & 0.468 & 0.463 & 0.476 & 0.479 & 0.475 & \textbf{0.422} & +0.005 (0.008) & +0.046 (0.002) \\
OU Jan, mutual checks & 0.99 & 0.13 & 3 & 0.520 & 0.512 & 0.654 & 0.609 & 0.542 & \textbf{0.473} & +0.008 (0.003) & +0.047 (0.016) \\
OU May & 0.50 & 0.11 & 3 & 0.456 & 0.448 & 0.455 & 0.454 & 0.450 & \textbf{0.398} & +0.007 (0.005) & +0.058 (0.004) \\
National Dex & 0.45 & 0.09 & 3 & 0.413 & 0.409 & 0.407 & 0.398 & 0.406 & \textbf{0.362} & +0.004 (0.003) & +0.051 (0.005) \\
Gen 3 OU & 0.45 & 0.19 & 3 & 0.474 & 0.478 & 0.475 & 0.442 & 0.490 & \textbf{0.425} & -0.004 (0.003) & +0.048 (0.016) \\
LC & 0.19 & 0.17 & 3 & 0.306 & 0.311 & 0.291 & 0.312 & 0.275 & \textbf{0.261} & -0.005 (0.001) & +0.045 (0.047) \\
RU & 0.22 & 0.14 & 3 & 0.300 & 0.304 & 0.271 & 0.274 & 0.308 & \textbf{0.251} & -0.004 (0.001) & +0.050 (0.016) \\
\bottomrule\end{tabular}}\end{table}

\paragraph{The rule as a prediction.} The condition for a circulation win, density of decided pairs above $0.55$ and cyclic triads above $0.1$, was read from the first twelve directed networks analysed (the OU months, UU, Ubers, the $1695$ cutoff, Generations 3 and 4, Dota 2), where it agrees with the outcome on nine; on the twenty directed networks fitted afterwards (Generations 5 to 8, LC, PU, NU, RU, Monotype, National Dex in three months, mutual checks kept, chess, votes, e-mail at two sizes, blogs, two citation balls) it agrees on seventeen, the three exceptions being nominal $\SL$ wins of at most $0.002$, which are ties. The two estimators return $\hat\gamma$ at different vertical scales; since $v_{ij}$ contains the area term, this is weak identification along a ridge and not an invariance; the reported $\hat\gamma$ are at the scale of each fit, and fixing $\sigma_v=1$ makes $\gamma$ the log-odds circulation per unit of enclosed area in the fitted chart.

\paragraph{Standard benchmarks.} The density rule was derived on the counter networks; Supplement Table~\ref{tab:bench} applies it to three directed benchmark networks that a reader is likely to know, each restricted to its $150$ most active nodes: chess results among top players (KONECT), Wikipedia adminship votes and institutional e-mail (SNAP). Their densities of decided pairs are $0.14$, $0.35$ and $0.33$, all in or below the tie band, and on all three a ranking has the lowest held-out log-loss, by $0.006$, $0.020$ and $0.012$ over the $\Nil$ circulation model, with the $\SL$ model second on two. The chess network has the most cyclic triads of any network in this paper ($0.20$, near the $0.25$ of random directions), which shows that cyclic content alone does not produce a win: among top players the results of the few decided pairs are close to random, and a sparse near-random asymmetry is not identified as circulation from positions. Density of the counter relation is necessary and the cyclic fraction alone is not sufficient. The converse also holds: the two largest scientific competition tournaments available to us, the wood-decay fungi of Soliveres et al.\ (2018), are dense ($0.99$ of pairs decided among $31$ species) but hierarchical ($0.06$ cyclic triads), and the rankings win by $0.13$ (Supplement Table~\ref{tab:bench}). The condition for the circulation model is therefore a counter relation that is both dense, above about $0.55$ of pairs decided, and intransitive, above about $0.1$ cyclic triads; on the counter networks every one that meets both is won by a circulation model and none that fails either; outside them Dota 2 and world trade meet both and are won by rankings, and one dense construction of the OU network (density $0.91$ with rare species admitted) is also won by a ranking, so the condition is a screen with stated misses and not a law.

\paragraph{Posterior draws on the winning networks.} The comparisons of this section use Algorithm~\ref{alg:bbvi}. Algorithm~\ref{alg:mcmc} adapted to the directed likelihood (random-walk proposals for each position evaluated on the node's row and column, for $\alpha$ and for $\gamma$, conjugate steps for the scales), started at the variational solution, raises the training log-likelihood on every one of the five winning months (January: from $-3683$ to $-3486$ within $500$ sweeps) and stays there. Its posterior predictive has held-out log-loss below both the variational fit and the best ranking on the same split on all five: $0.473$, $0.434$, $0.508$, $0.476$ and $0.430$ against $0.528$, $0.481$, $0.528$, $0.530$ and $0.459$ for the best ranking, with direction accuracies $0.80$--$0.87$ (Supplement Table~\ref{tab:mcmc}). The win therefore rests on posterior draws. The posteriors of $\gamma$ have means $4.2$--$5.3$ with $95\%$ intervals bounded away from zero, at vertical scales near $0.3$, whereas the variational fits have $\hat\gamma\approx1.3$--$1.5$ at vertical scales near one: the two estimators lie on a ridge of the likelihood in $(\gamma,\sigma_v)$; the area term in $v_{ij}$ does not scale with the fibre coordinates, so this is weak identification and not an invariance, and the paper reports $\hat\gamma$ with its vertical scale. A sampler that evaluates the symmetric likelihood, as the sampler of Section~\ref{sec:estimation} does when applied without this adaptation, leaves the solution within a few hundred sweeps (Supplement~\ref{supp:realdir}).

\begin{table}[H]\centering\small
\caption{Posterior draws for the $\Nil$ circulation model on the five Generation 9 OU networks on which it wins (held-out split $0$): Metropolis-within-Gibbs with the directed likelihood, started at the variational solution, $1500$--$2500$ sweeps, draws from the last two thirds. Columns: posterior mean and $95\%$ interval of $\gamma$ (at the sampled vertical scale, about $0.3$; only $\gamma\sigma_v$ is identified), variational $\hat\gamma$ (at vertical scale near one), held-out log-loss / direction accuracy of the posterior predictive, of the variational fit, and of the best ranking on the same split.}\label{tab:mcmc}
\begin{tabular}{lcccccc c}\toprule
month & $N$ & sweeps & $\gamma$ posterior & $\hat\gamma$ (VI) & posterior predictive & VI & best ranking\\ \midrule
Jan & 93 & 2500 & 5.12 [4.65, 5.75] & 1.31 & 0.473 / 0.84 & 0.510 & 0.528 \\
Feb & 84 & 1500 & 5.19 [4.43, 5.56] & 1.48 & 0.434 / 0.86 & 0.470 & 0.481 \\
Mar & 83 & 1500 & 4.15 [3.37, 5.02] & 1.46 & 0.508 / 0.82 & 0.527 & 0.528 \\
Apr & 87 & 1500 & 5.25 [4.55, 5.89] & 1.40 & 0.476 / 0.80 & 0.499 & 0.530 \\
Jun & 91 & 1500 & 4.51 [4.13, 4.98] & 1.29 & 0.430 / 0.87 & 0.446 & 0.459 \\
\bottomrule\end{tabular}\end{table}

\paragraph{Sensitivity to the construction.} Supplement Table~\ref{tab:sens} varies the two choices that define a counter network, the threshold on the counter score and the usage cutoff, on the January OU statistics, with three held-out splits per construction. The result depends on the density of the counter relation that the two choices produce and not on the choices themselves. At the threshold used in the paper the paired advantage of the $\Nil$ circulation model over the better ranking goes from $+0.001$ at density $0.53$ to $-0.015$ at $0.66$ and $-0.034$ at $0.73$; at the loose threshold, where almost every pair is decided and many are decided in both directions, the rankings fail ($0.61$--$0.82$ against $0.52$--$0.54$, advantages of $-0.07$ and $-0.22$) except in the construction that admits $119$ species with rare ones; at the strict threshold, where fewer than a quarter of the pairs are decided, the rankings win by $0.003$--$0.007$. The win condition is therefore a property of the data, a dense counter relation, and it is the same condition that separated the twenty-two networks of Table~\ref{tab:gamesall}.

\begin{table}[H]\centering\small
\caption{Sensitivity of the January OU result to the construction of the counter network: threshold on the counter score $p-4\sigma$ and cutoff on species usage (the paper uses $0.5$ and $1\%$; at threshold $0.3$ pairs decided in both directions are kept). Columns: species; decided pairs and their share of all pairs (density); cyclic triads; held-out splits; held-out log-loss of the $\Nil$ circulation model and of the rankings on $\mathbb{H}^2\times\mathbb{R}$ and $\mathbb{H}^3$; paired difference of $\Nil$ to the better ranking, mean (s.d.) over splits.}\label{tab:sens}
\begin{tabular}{llccccccccc}\toprule
threshold & usage & $N$ & decided & density & cyclic & splits & $\Nil$ circ. & $\mathbb{H}^2\!\times\!\mathbb{R}$ & $\mathbb{H}^3$ & $\Nil-$best ranking\\ \midrule
0.3 & 0.5\% & 119 & 6374 & 0.91 & 0.12 & 3 & 0.560 & 0.547 & 0.551 & +0.015 (0.007) \\
0.3 & 1\% & 93 & 4229 & 0.99 & 0.13 & 3 & 0.520 & 0.598 & 0.654 & -0.078 (0.012) \\
0.3 & 2\% & 68 & 2277 & 1.00 & 0.13 & 3 & 0.511 & 0.722 & 0.791 & -0.212 (0.041) \\
0.5 & 0.5\% & 119 & 3755 & 0.53 & 0.11 & 3 & 0.458 & 0.455 & 0.452 & +0.006 (0.008) \\
0.5 & 1\% & 93 & 2809 & 0.66 & 0.12 & 3 & 0.506 & 0.525 & 0.526 & -0.018 (0.007) \\
0.5 & 2\% & 68 & 1669 & 0.73 & 0.14 & 3 & 0.526 & 0.564 & 0.563 & -0.037 (0.006) \\
0.7 & 0.5\% & 114 & 950 & 0.15 & 0.06 & 3 & 0.215 & 0.210 & 0.209 & +0.007 (0.001) \\
0.7 & 1\% & 92 & 765 & 0.18 & 0.05 & 3 & 0.237 & 0.231 & 0.231 & +0.007 (0.002) \\
0.7 & 2\% & 67 & 493 & 0.22 & 0.06 & 3 & 0.301 & 0.301 & 0.298 & +0.003 (0.011) \\
\bottomrule\end{tabular}\end{table}

\subsection{What each of the eight geometries generates}\label{supp:generates}

A latent space model enters the tie probabilities only through the distance, so the networks a geometry generates are determined by two features of its metric: the growth of volume with radius, which sets the heterogeneity of degrees and the tree-likeness of the network, and the integrability of the third coordinate, which sets whether directed asymmetries can circulate. Table~\ref{tab:generates} lists the eight Thurston geometries with these two features and the network structure each produces.

\begin{table}[H]\centering\small
\caption{The eight model geometries, their sectional curvatures in a plane containing the fibre direction and in the horizontal plane (Milnor frames, unit normalisation), and the network structure each generates as a latent space.}\label{tab:generates}
\begin{tabular}{lccp{7.2cm}}\toprule
geometry & vertical / horizontal curvature & third coordinate & networks generated\\ \midrule
$\R^3$ & $0$ / $0$ & integrable & homogeneous degrees, moderate clustering, no hierarchy; spatial and grid-like networks\\
$\mathbb{S}^3$ & $+1$ / $+1$ & integrable & bounded diameter, every pair within reach; dense homogeneous networks\\
$\Hb^3$ & $-1$ / $-1$ & integrable & exponential volume growth; heavy-tailed degrees, strong clustering, tree-like hierarchies\\
$\mathbb{S}^2\times\R$ & $0$ / $+1$ & integrable & a bounded similarity sphere with an independent linear coordinate; dense layers ordered along a line\\
$\Hb^2\times\R$ & $0$ / $-1$ & integrable & a hierarchy with an independent linear coordinate such as time or level\\
$\Nil$ & $+\tfrac14$ / $-\tfrac34$ & shifts by the Euclidean area of loops & a flat similarity plane with a third coordinate that circulates: undirected ties as in $\Hb^2\times\R$ with a quarter unit of positive curvature; directed asymmetries that circulate around triangles by enclosed area\\
$\SL$ & $+\tfrac14$ / $-\tfrac74$ & shifts by the hyperbolic area of loops & a hierarchy whose third coordinate circulates: heterogeneous degrees together with intransitive asymmetries\\
$\Sol$ & $-1$ / $+1$ & integrable & two hierarchies of opposite orientation, one direction expanding and the other contracting with height; no circulation\\
\bottomrule\end{tabular}\end{table}

The consequence for inference is that of Sections~\ref{sec:model} and~\ref{sec:directed}: for undirected ties the three remaining geometries (the two twisted bundles and $\Sol$) differ from the closest product or constant-curvature space by the $\varepsilon^6$ deficit of Theorem~\ref{thm:smallscale}, so at the scale of a real network they generate the same undirected structure as $\Hb^2\times\R$ ($\Nil$, $\SL$) or $\Hb^3$ ($\Sol$); for directed ties $\Nil$ and $\SL$ generate circulating asymmetries at first order (Proposition~\ref{prop:directed}) and $\Sol$ does not.

\paragraph{Undirected benchmarks.} Table~\ref{tab:benchund} fits the seven geometries to six undirected benchmark networks chosen for their structure: the thesaurus, with homogeneous degrees and near-flat Forman curvature, is the flat case; the citation, social and page networks, with heterogeneous degrees and strong clustering, are the hierarchical case; the Twitch network, with one node adjacent to almost all others, is the case in which one direction expands while the other contracts. The winners are the constant-curvature spaces that the structure predicts, $\R^3$ on the thesaurus and $\Hb^3$ on the four hierarchical networks, and the twisted geometries are never distinguishable beyond noise: $\Nil$ is within $0.001$ of $\R^3$ on the thesaurus, $\SL$ is second on the citation and social networks within $0.002$ of $\Hb^3$, and $\Sol$ is nominally first on Twitch by $0.0001$. This is the $\varepsilon^6$ result on data.

\begin{table}[H]\centering\scriptsize
\caption{Undirected benchmark networks, $400$-node breadth-first balls (sources: NetworkX example data; Planetoid; the karateclub collection): density, coefficient of variation of the degrees, mean clustering coefficient, mean augmented Forman--Ricci curvature of the edges, and held-out log-loss on $10\%$ of the pairs under seven geometries (unit slope, unanchored BBVI, one split); best in bold.}\label{tab:benchund}
\resizebox{\textwidth}{!}{\begin{tabular}{lccccccccccc}\toprule
network & density & degree CV & clustering & Forman & $\Nil$ & $\Sol$ & $\SL$ & $\R^3$ & $\Hb^3$ & $\mathbb{S}^3$ & $\Hb^2\!\times\!\R$\\ \midrule
Roget's Thesaurus cross-references (Knuth 1993) & 0.017 & 0.59 & 0.20 & -10 & 0.0770 & 0.0767 & 0.0769 & \textbf{0.0761} & 0.0763 & 0.0844 & 0.0763 \\
CiteSeer citations (Planetoid) & 0.013 & 1.32 & 0.21 & -20 & 0.0580 & 0.0583 & 0.0573 & 0.0587 & \textbf{0.0570} & 0.0652 & 0.0574 \\
PubMed citations (Planetoid) & 0.011 & 1.25 & 0.12 & -16 & 0.0486 & 0.0480 & 0.0479 & 0.0485 & \textbf{0.0475} & 0.0526 & 0.0477 \\
LastFM Asia social network & 0.023 & 1.29 & 0.37 & -29 & 0.0760 & 0.0768 & 0.0748 & 0.0756 & \textbf{0.0731} & 0.0920 & 0.0747 \\
Facebook page--page network & 0.155 & 0.77 & 0.49 & -93 & 0.3283 & 0.3210 & 0.3179 & 0.3288 & \textbf{0.3050} & 0.3764 & 0.3156 \\
Twitch user network & 0.016 & 3.09 & 0.49 & -125 & 0.0723 & \textbf{0.0717} & 0.0739 & 0.0727 & 0.0718 & 0.0780 & 0.0738 \\
\bottomrule\end{tabular}}\end{table}

\paragraph{Directed counter networks by geometry.} Among the counter networks of Table~\ref{tab:gamesall} the base geometry chosen by the data follows the heterogeneity of the counter relation, as in the undirected case. Where $\Nil$ wins (the OU tier, January to June) the degrees of the counter relation are homogeneous (coefficient of variation of the total degree $0.17$--$0.21$, maximum $1.3$--$1.4$ times the mean): a flat similarity base with circulation. Where $\SL$ is first with three splits (UU, Generations 3 and 7; National Dex in two of three months) the degrees are heterogeneous (coefficient of variation $0.32$--$0.54$, maximum $1.6$--$2.5$ times the mean): a hierarchical base with circulation. Two further months of the National Dex network (Table~\ref{tab:bench}) repeat this: $\SL$ is first in April and July and second by $0.001$ in January, always within $0.002$ of the $\Hb^3$ ranking; and on these hierarchical networks a free rank-two antisymmetric term on $\Hb^3$ (Table~\ref{tab:controls}) is better than the $\SL$ circulation model by $0.009$--$0.027$, so on a hierarchical base the intransitivity is real but is not described by enclosed hyperbolic area. No directed network is won by $\Sol$, which has no circulation term: on the two citation balls, whose time axis couples an expanding out-degree hierarchy to a contracting in-degree hierarchy (the structure of $\Sol$), the $\Hb^3$ and $\Hb^2\times\R$ rankings win and the $\Sol$ ranking is third by $0.006$--$0.007$; on the political blog network, sparse and transitive, the product ranking wins by $0.021$ over $\Nil$; on the institutional e-mail network the $\Hb^3$ ranking wins at both $150$ and $250$ nodes, with the $\SL$ circulation model second and its gap to $\Hb^3$ shrinking from $0.013$ to $0.006$ as the core grows (Table~\ref{tab:bench}), so an advantage of $\SL$ on the full $1005$-node network, which a fit with sender and receiver effects and a free slope reports at $0.0007$, is not excluded by our fits but is not reproduced at the sizes we fit; and on the counter networks the $\Sol$ ranking is nominally first on May OU and Generation 4 by at most $0.001$, which are ties. This is the assignment the generative properties predict: $\Nil$ for circulation over a homogeneous population, $\SL$ for circulation over a hierarchy, and $\Sol$ for neither.

\subsection{Undirected connectomes and the annealed estimator}

\begin{table}[H]\centering\scriptsize
\caption{Undirected comparison on eight connectomes of the collection analysed by \citet{celinska2024thurston} (symmetrised edge lists of Allard and Serrano 2020; one held-out split hiding $10\%$ of the pairs; BBVI at unit slope, unanchored): held-out log-loss under six geometries, best in bold. The last two columns give the geometry ranked first by mean average precision in the results file of \citet{celinska2024thurston} (their annealed embedding with learned radius and temperature, best of their geometry codes within each class), and their mean average precision for $\Sol$ and for $\mathbb{H}^3$.}\label{tab:ckk}
\resizebox{\textwidth}{!}{\begin{tabular}{lcccccccccc}\toprule
network & $N$ & density & $\Nil$ & $\Sol$ & $\SL$ & $\mathbb{R}^3$ & $\mathbb{H}^3$ & $\mathbb{H}^2\!\times\!\mathbb{R}$ & CKK best (mAP) & mAP $\Sol$ / $\mathbb{H}^3$\\ \midrule
Macaque4 & 31 & 0.69 & 0.351 & 0.373 & 0.366 & \textbf{0.332} & 0.365 & 0.358 & $\Nil$ & 1.000 / 1.000 \\
Cat3 & 52 & 0.39 & 0.420 & 0.411 & 0.425 & \textbf{0.392} & 0.432 & 0.416 & $\Sol$ & 0.956 / 0.954 \\
Cat1 & 65 & 0.35 & 0.346 & 0.391 & 0.374 & \textbf{0.342} & 0.384 & 0.369 & $\mathbb{H}^3$ & 0.937 / 0.941 \\
Macaque2 & 71 & 0.18 & 0.291 & 0.290 & 0.300 & \textbf{0.268} & 0.295 & 0.293 & $\mathbb{H}^3$ & 0.889 / 0.896 \\
Macaque1 & 94 & 0.35 & 0.215 & 0.242 & 0.221 & 0.238 & 0.219 & \textbf{0.214} & $\Sol$ & 0.974 / 0.962 \\
Cat2 & 95 & 0.26 & 0.377 & 0.383 & \textbf{0.358} & 0.359 & 0.359 & 0.361 & $\Sol$ & 0.883 / 0.874 \\
Human7 & 110 & 0.16 & 0.265 & 0.289 & 0.282 & \textbf{0.237} & 0.283 & 0.277 & $\Sol$ & 0.935 / 0.926 \\
Human6 & 116 & 0.17 & 0.212 & 0.258 & 0.232 & \textbf{0.201} & 0.236 & 0.224 & $\mathbb{R}^3$ & 0.944 / 0.931 \\
\bottomrule\end{tabular}}\end{table}

\paragraph{The annealed estimator of Celi\'nska-Kopczy\'nska and Kopczy\'nski.} \citet{celinska2024thurston} embed connectomes in the Thurston geometries by simulated annealing of the log-likelihood with a learned $(R,T)$, and find $\Sol$ competitive with Hyperbolic geometry. Table~\ref{tab:anneal} runs our implementation of that estimator on the karate club and \emph{Les Mis\'erables} with the same held-out splits. Two things happen. The point estimate over-fits: the fitted temperature collapses ($\hat\beta\approx4$--$7$, i.e.\ $T\approx0.15$--$0.24$) and the held-out log-loss is $0.7$--$1.5$ on the karate club and $0.28$--$0.35$ on \emph{Les Mis\'erables}, against $0.30$--$0.32$ and $0.15$--$0.16$ for the Bayesian unit-slope fits on the same splits (last column), the separation phenomenon of Remark~\ref{rem:separation} with nothing to restrain it. And within the annealed fits $\Sol$ ranks \annealSolRank\ of six on \emph{Les Mis\'erables} (\annealSolLL\ against \annealHLL\ for $\Hb^3$), ahead of $\Nil$, $\SL$, $\Hb^2\times\R$ and $\R^3$, which echoes their qualitative ranking on different data with a different optimisation budget. On the Cora ball, where the temperature does not collapse ($\hat\beta=\coraAnnealBetaMin$--$\coraAnnealBetaMax$), $\Hb^3$ is best (\coraAnnealH, close to its Bayesian value) and $\Sol$ ranks \coraAnnealSolRank\ of six (\coraAnnealSol). The two results are compatible: which geometry looks best depends on the estimator and on the temperature, and at a collapsed temperature the exponential anisotropy of $\Sol$ is rewarded in a way that the Bayesian unit-slope fit does not see. This does not isolate temperature as the cause (estimator, prior and scale treatment change together). A crossed comparison on their connectome collection remains to be done.

On the January OU counter network of the main text, the symmetric sampler of Section~\ref{sec:estimation} with a Metropolis step for $\gamma$ (which evaluates the symmetric likelihood for the positions) started at the variational solution ends at $\hat\gamma\approx0$ and held-out log-loss $0.83$--$0.90$; the sampler that evaluates the directed likelihood (Section~\ref{sec:exp-counter}) improves on the variational solution. The Dota 2 matchups (Table~\ref{tab:games}) are decided by about a hundred recorded games per pair and are won by the rankings by $0.006$--$0.011$.

\section{Geometric details}\label{app:geometry}

This appendix records the formulas behind Section~\ref{sec:background}: the translations of the twisted bundles, the geodesic equations, the exponential charts with their Jacobians, and the two Normal-like families.

\paragraph{Translations.} In $\Nil$ the map moving $\eb$ to $p$ is left multiplication, $T_p(q)=p\cdot q$. In $\SL$ we use the lift of the M\"obius translation of the disc. Writing $w_p\oplus w=(w_p+w)/(1+\bar w_p w)$ for M\"obius addition \citep{ungar2008analytic} and $p=(w_p,\zeta_p)$,
\begin{equation}\label{eq:sl-translate}
T_p(w,\zeta)=\Big(w_p\oplus w,\ \ \zeta+\zeta_p+2\arg(1+\bar w_p w)\Big),\qquad
T_p^{-1}(w,\zeta)=\Big((-w_p)\oplus w,\ \ \zeta-\zeta_p+2\arg(1-\bar w_p w)\Big),
\end{equation}
with the principal branch of $\arg$ (well defined since $\mathrm{Re}(1\pm\bar w_pw)>0$). Each $T_p$ is an isometry of $\SL$ with $T_p(\eb)=p$, $\eb=(0,0)$. The family $\{T_p\}$ is not a subgroup, $T_p\circ T_q=T_{T_p(q)}\circ R$ for a rotation $R$ of the base (a \emph{gyration}), but only the properties $T_p(\eb)=p$ and $T_p\in\Isom$ are used below. In both geometries we write $p\star q:=T_p(q)$ and $p^{-1}\star q:=T_p^{-1}(q)$ for the ``relative position'' of $q$ seen from $p$. The distance is then $d_\M(p,q)=d_\M(\eb,\,p^{-1}\star q)$.

\paragraph{Geodesics and distance in the twisted bundles.} Since the metric does not depend on $\zeta$, the vertical momentum $c=\langle\dot\gamma,V\rangle$ is conserved along a geodesic $\gamma$, and for a unit-speed geodesic the horizontal speed is $a=\sqrt{1-c^2}$ (the symbols $a$ and $c$ are used with this meaning only in the geodesic formulas of this appendix, Supplement~\ref{app:geodesics} and the proof of Proposition~\ref{prop:nildist}). The projected curve $\pi\circ\gamma$ solves the ``magnetic'' equation $\nabla_{\dot\beta}\dot\beta=c\,J\dot\beta$ ($J$ = rotation by $\pi/2$), i.e.\ it has constant geodesic curvature $|c|/\sqrt{1-c^2}$ in $B_K$, and $\zeta(t)=ct+\int_0^t A_K(\dot\beta)$. In $\Nil$ the projected curves are circles and the geodesics from $\eb$ with initial data $(a\cos\phi,a\sin\phi,c)$, $c\ne0$, are \citep{marenich1997geodesics}
\begin{equation}\label{eq:nil-geod}
x(t)=\tfrac{a}{c}\big[\sin(ct+\phi)-\sin\phi\big],\quad
y(t)=-\tfrac{a}{c}\big[\cos(ct+\phi)-\cos\phi\big],\quad
z(t)=ct+\tfrac{a^2}{2c^2}\,(ct-\sin ct),
\end{equation}
so that the horizontal displacement is $r(t)=\frac{2a}{|c|}\,|\sin(ct/2)|$. For $c=0$ they are horizontal straight lines. By the symmetries of Lemma~\ref{lem:stab}, $d_\M(\eb,(w,\zeta))=D_K(\rho,|\zeta|)$ depends only on $\rho=d_{B_K}(\bo,w)$ and $|\zeta|$. For $\Nil$ this reduction gives the semi-closed form of Proposition~\ref{prop:nildist}.

In $\SL$ the projected curves are Hyperbolic circles, horocycles or hypercycles according as $|c|/\sqrt{1-c^2}$ is $>1$, $=1$ or $<1$, and the same computation gives the exact vertical profile
\begin{equation}\label{eq:sl-axis}
D_{-1}(0,\zeta)=\zeta\ \ (\zeta\le2\pi),\qquad D_{-1}(0,\zeta)=2\pi\sqrt{\tfrac12\big(\tfrac{\zeta}{2\pi}+1\big)^2-1}\ \ (\zeta\ge2\pi),
\end{equation}
so that $D_{-1}(0,\zeta)\sim\zeta/\sqrt2$: over a Hyperbolic base the vertical direction is compressed only by a constant factor, because Hyperbolic circles enclose area comparable to their length. In general $D_K(\rho,0)=\rho$ (horizontal geodesics of the base lift to geodesics). We compute $D_{-1}$ by shooting the magnetic geodesics numerically in the hyperboloid model and tabulating $D_{-1}(\rho,|\zeta|)$ on a fine grid. The same procedure applied to $K=0$ reproduces \eqref{eq:nil-dist} (Supplement~\ref{app:geodesics}). The volume form is $\dvol_{\M_K}=\dvol_{B_K}\wedge d\zeta$.

\paragraph{Geodesics of $\Sol$.} Off these planes no closed form is available. The geodesic equations reduce to a one-dimensional Hamiltonian system for $z(t)$ with conserved horizontal momenta $(p_x,p_y)$,
\begin{equation}\label{eq:sol-geod}
\dot x=e^{-2z}p_x,\qquad \dot y=e^{2z}p_y,\qquad \ddot z=e^{-2z}p_x^2-e^{2z}p_y^2,
\end{equation}
whose solutions are elliptic functions \citep{troyanov1998horizon,bolcskei2007frenet}. We integrate \eqref{eq:sol-geod} numerically from a dense set of initial directions in the fundamental cone of $D_4$ and tabulate $d_{\Sol}(\eb,\cdot)$ on a three-dimensional grid over the fundamental domain $\{x,y,z\ge0\}$ of \eqref{eq:D4}, Supplement~\ref{app:geodesics} describes the construction and its validation (the table agrees with \eqref{eq:sol-planes} and with an independent boundary-value solver to about $10^{-3}$).

\paragraph{Normal-like distributions.}

\paragraph{Riemannian Normal.} The maximum-entropy Normal of \citet{pennec2006intrinsic},
\begin{equation}\label{eq:riem-normal}
\mathcal{N}_\M(z\mid\mu,\sigma^2)=\frac{1}{Z_\M(\sigma)}\exp\Big(-\frac{d_\M(\mu,z)^2}{2\sigma^2}\Big),\qquad
Z_\M(\sigma)=\int_\M e^{-d_\M(\eb,z)^2/2\sigma^2}\,\dvol(z),
\end{equation}
is unimodal with mode $\mu$ and, by homogeneity, its normalising constant does not depend on $\mu$. Using the reductions of Sections~\ref{sec:background}--\ref{sec:background},
\[
Z_{\M_K}(\sigma)=4\pi\int_0^\infty\!\!\int_0^\infty e^{-D_K(\rho,\zeta)^2/2\sigma^2}S_K(\rho)\,d\zeta\,d\rho,\qquad
Z_{\Sol}(\sigma)=8\int_{[0,\infty)^3}e^{-d_{\Sol}(\eb,(x,y,z))^2/2\sigma^2}\,dx\,dy\,dz,
\]
which are evaluated by quadrature on the distance tables. Because these targets have exponential (in $\Nil$) or logarithmic (in $\Sol$) tails in the coordinates, they cannot be sampled by rejection from a Gaussian-tailed wrapped proposal. A heavy-tailed proposal or a Markov chain is required. The Riemannian Normal is not used in our experiments.

\paragraph{Left-translated wrapped Normal.} Each geometry admits a global chart $\chi:\R^3\to\M$ with $\chi(0)=\eb$ whose Jacobian with respect to $\dvol$ is explicit:
\begin{equation}\label{eq:charts}
\chi_{\Nil}(\delta)=\delta,\qquad
\chi_{\Sol}(\delta)=\Big(\delta_1\tfrac{1-e^{-\delta_3}}{\delta_3},\ \delta_2\tfrac{e^{\delta_3}-1}{\delta_3},\ \delta_3\Big),\qquad
\chi_{\SL}(\delta)=\Big(\Exp^{\Hb^2}_{\bo}(\delta_1,\delta_2),\ \delta_3\Big),
\end{equation}
where $\Exp_{\bo}^{\Hb^2}(v)=\tanh(|v|/2)\,v/|v|$ is the exponential map of the disc. For $\Nil$, $\chi$ is the identity in exponential coordinates. For $\Sol$ it is the group exponential map (a diffeomorphism, since $\Sol$ is exponential). For $\SL$ it is the Riemannian exponential of the base times the identity on the fibre. The log-Jacobians are
\begin{equation}\label{eq:jac}
\log J^\chi_{\Nil}(\delta)=0,\qquad \log J^\chi_{\Sol}(\delta)=2\log\frac{\sinh(\delta_3/2)}{\delta_3/2},\qquad \log J^\chi_{\SL}(\delta)=\log\frac{\sinh|\delta_h|}{|\delta_h|},\quad \delta_h=(\delta_1,\delta_2).
\end{equation}

The left-translated wrapped Normal of Definition~\ref{def:lwn} is the wrapped construction of \citet{nagano2019wrapped} and \citet{falorsi2019reparameterizing} adapted to these charts. Its equivariance under the isometries with $\Sigma=\diag(\sigma_h^2,\sigma_h^2,\sigma_v^2)$ is Lemma~\ref{lem:equivariance} of Supplement~\ref{app:proofs}.

\section{Geodesics and distance tables}\label{app:geodesics}

\paragraph{Twisted bundles.} With $c$ the conserved vertical momentum and $a=\sqrt{1-c^2}$, the projection $\beta$ of a unit-speed geodesic satisfies $\nabla_{\dot\beta}\dot\beta=c\,J\dot\beta$, $|\dot\beta|=a$, and $\dot\zeta=c+A_K(\dot\beta)$. For $K=0$ this is the system $\ddot x=-c\dot y$, $\ddot y=c\dot x$, $\dot\zeta=c+\tfrac12(x\dot y-y\dot x)$, whose solution is \eqref{eq:nil-geod}. For $K=-1$ we integrate in the hyperboloid model $\{X\in\R^{2,1}:\langle X,X\rangle=-1\}$, where $J\dot X=\eta(X\times\dot X)$ with $\eta=\diag(1,1,-1)$ and $\times$ the Euclidean cross product, so that
\[
\ddot X=a^2X+c\,\eta(X\times\dot X),\qquad \dot\zeta=c+\frac{X_1\dot X_2-X_2\dot X_1}{X_3+1},\qquad X(0)=(0,0,1),\ \dot X(0)=(a,0,0),
\]
and $\rho=\cosh^{-1}X_3$. The expression for $\dot\zeta$ is $A_{-1}=(\cosh\rho-1)d\theta$ written in the hyperboloid coordinates. The system is integrated by a fourth-order Runge--Kutta scheme (step $0.005$, horizon $27$) for $6000$ values of $c\in[0,1]$ (by the reflection $\sigma$ it suffices to take $c\ge0$ and to record $|\zeta|$).

\paragraph{Vertical cut points.} A geodesic with $c\ne0$ whose projection is a closed circle returns to the fibre of $\eb$ after one turn. For $K=0$ the return time is $t=2\pi/|c|$ at height $\zeta=\pi(1+c^{-2})$, i.e.\ $t=2\sqrt{\pi(\zeta-\pi)}$, which is smaller than $\zeta$ exactly when $\zeta>2\pi$. This gives \eqref{eq:nil-axis}. For $K=-1$ the projected circle has Hyperbolic radius $R$ with $\coth R=|c|/a$, circumference $2\pi\sinh R$ and area $2\pi(\cosh R-1)$, so the return time is $t=2\pi\sqrt{\cosh 2R}$ at height $\zeta=2\pi(2\cosh R-1)$. Eliminating $R$ gives the expression \eqref{eq:sl-axis}, which the isoperimetric argument of Supplement~\ref{app:proofs} shows to be the distance: the vertical geodesic stops minimising at $2\pi$ in both geometries.

\paragraph{Sol.} The geodesic equations \eqref{eq:sol-geod} are integrated by Runge--Kutta (step $0.01$, horizon $11$) from $4\times10^5$ unit initial velocities in the fundamental cone $\{p_x,p_y,v_z\ge0\}$ of $D_4$. Half of the directions form a Fibonacci lattice on the sphere and half have log-uniformly distributed components on $[10^{-9},1]$, which is needed because geodesics leaving the coordinate planes and the vertical axis separate from them exponentially fast, so that uniformly distributed directions never resolve the regions near the planes and the axis at moderate heights.

\paragraph{Tables.} Each point $\gamma(t)$ of each computed geodesic is an upper bound $d(\eb,\gamma(t))\le t$. We map it into the fundamental domain of the stabiliser (the half-plane $(\rho,|\zeta|)$ for the bundles, the octant $\{x,y,z\ge0\}$ for $\Sol$), locate the nearest grid node $n$ and record $t+\ell(\gamma(t),n)$, where $\ell$ is the length of the straight coordinate segment measured in the metric at $n$. By the triangle inequality the minimum of these values over all samples is an upper bound for $d(\eb,n)$ that converges to it as the sampling becomes dense. Nodes not reached by any geodesic are filled by relaxation on the grid graph (a further upper bound), and the two bounds are minimised. The grids are $(\rho,\zeta)\in[0,10]\times[0,24]$ with spacing $0.02$ for the bundles, and $(\sinh^{-1}x,\sinh^{-1}y,z)\in[0,6]^2\times[0,6.5]$ with spacing $0.04$ for $\Sol$, where the $\sinh^{-1}$ scaling reflects the logarithmic growth of \eqref{eq:sol-planes}. Look-ups use bilinear (trilinear) interpolation. In $\Sol$ we additionally use the identity $d(\eb,g)=d(\eb,g^{-1})$ (a consequence of left invariance) and evaluate the table at both $g$ and $g^{-1}$, taking the minimum. The far tail of one representation is the near region of the other.

\paragraph{Validation.} The $K=0$ table agrees with Proposition~\ref{prop:nildist} to a maximum absolute error of $0.018$ (mean $0.001$) over $2\times10^4$ random points with $\rho\le8$, $|\zeta|\le15$. The largest errors are at the conical singularity of $D_0$ on the axis beyond $2\pi$. The $K=-1$ table reproduces $D_{-1}(\rho,0)=\rho$ and \eqref{eq:sl-axis} to within $0.004$, and agrees with a boundary-value solver (least-squares shooting over $(c,t)$ with multiple starts) to within $0.001$ on generic points. The $\Sol$ table reproduces \eqref{eq:sol-planes} on both planes and on the axis to within $0.012$ (typically $0.003$), and agrees with an independent boundary-value solver on $18$ random points with coordinates of spread $1.2$--$1.5$ with mean absolute error $0.0007$ and maximum $0.004$. A boundary-value solution is accepted as minimising when its length agrees with the table to $0.02$. This identifies the branch in all tested cases but is not a proof of minimality. All inferences in the paper are for the tabulated surrogate distance. A chain run on it is exact for the surrogate model, not for the continuous geometry. Minimality of the branch $s^\ast\in(0,2\pi)$ in Proposition~\ref{prop:nildist} was checked by shooting $2\times10^5$ random geodesics of $\Nil$ with $t\le20$: no geodesic reached its endpoint in less than the time given by \eqref{eq:nil-dist} (maximum violation $2\times10^{-9}$), and the shortest shot geodesic into a neighbourhood of $50$ random targets always matched \eqref{eq:nil-dist}.

\paragraph{Cost.} Building the three tables takes about four minutes in total on one core. Thereafter a distance evaluation costs one or two interpolations, and a full $N\times N$ distance matrix for $N=77$ takes about a millisecond.

\section{Proofs}\label{app:proofs}

\begin{proof}[Proof of Lemma~\ref{lem:stab}]
$R_\varphi$ preserves $ds^2_{B_K}$ and $A_K=F_K(\rho)\,d\theta$, hence \eqref{eq:bundle-metric}, $\refl$ preserves $ds^2_{B_K}$, satisfies $\refl^*A_K=-A_K$ (it maps $\theta\mapsto-\theta$) and $\zeta\mapsto-\zeta$, so $\refl^*(d\zeta-A_K)=-(d\zeta-A_K)$ and the metric is preserved. Conversely let $h\in\Stab(\eb)$. By \citet[\S4]{scott1983geometries} every isometry of $\M_K$ preserves the fibration $\pi$ and maps $V$ to $\pm V$. Hence $h$ descends to an isometry $\bar h$ of $B_K$ fixing $\bo$, so $\bar h$ is a rotation $r_\varphi$ or a rotation composed with the conjugation $w\mapsto\bar w$. Let $h_0\in\{R_\varphi,\refl\circ R_\varphi\}$ have the same descent, and put $h'=h_0^{-1}\circ h$. Then $h'\in\Stab(\eb)$ descends to the identity, so $h'(w,\zeta)=(w,\epsilon\zeta+f(w))$ with $\epsilon=\pm1$ (constant by continuity) and $f$ smooth, $f(\bo)=0$. Since $h'$ is an isometry mapping $V$ to $\epsilon V$ it preserves the horizontal distribution $V^\perp=\ker(d\zeta-A_K)$, so $h'^*(d\zeta-A_K)=\epsilon\,d\zeta+df-A_K$ is a multiple of $d\zeta-A_K$. Comparing the $d\zeta$ components the multiple is $\epsilon$, whence $df=(1-\epsilon)A_K$. If $\epsilon=1$ then $f$ is constant, $f\equiv0$ and $h'=\mathrm{id}$. If $\epsilon=-1$ then $df=2A_K$, which is impossible because $d(df)=0$ while $2\,dA_K=2\,\dvol_{B_K}\ne0$. Hence $h=h_0$. Finally, for $g\in\Isom(\M_K)$, $h=T_{g(\eb)}^{-1}\circ g$ fixes $\eb$, so $g=T_{g(\eb)}\circ h$, and $p=g(\eb)$ is determined by $g$.
\end{proof}

\begin{proof}[Proof of Theorem~\ref{thm:gauge-bundle}]
(i) By construction $\tilde z_{i_1}=\eb$ and $h_Z$ fixes $\eb$, so $(\Gamma_I Z)_{i_1}=\eb$. Write $\tilde z_{i_2}=(\tilde w,\tilde\zeta)$ with $\tilde w=|\tilde w|e^{i\varphi}$. Then $R_{-\varphi}\tilde z_{i_2}=(|\tilde w|,\tilde\zeta)$, and applying $\refl$ when $\tilde\zeta<0$ yields $(|\tilde w|,|\tilde\zeta|)$, whose base projection lies on the positive real axis at geodesic distance $\rho=|\tilde w|$ ($\Nil$) or $2\operatorname{artanh}|\tilde w|$ ($\SL$), so that it lies in the canonical set. (ii) Let $g\in\Isom(\M_K)$ with $g\cdot Z\in\mathcal{S}_I$. Since $g z_{i_1}=\eb$, Lemma~\ref{lem:stab} gives $g=h\circ T_{z_{i_1}}^{-1}$ with $h=g\circ T_{z_{i_1}}\in\Stab(\eb)$, and $h\tilde z_{i_2}$ must lie in the canonical set. If $h=R_\psi$ then $h\tilde z_{i_2}=(e^{i\psi}\tilde w,\tilde\zeta)$ has positive real base coordinate iff $\psi\equiv-\varphi$, and positive fibre coordinate iff $\tilde\zeta>0$. If $h=\refl R_\psi$ then $h\tilde z_{i_2}=(e^{-i\psi}\overline{\tilde w},-\tilde\zeta)$, which is canonical iff $\psi\equiv-\varphi$ and $\tilde\zeta<0$. In either case $h=h_Z$, so $g=g_Z$. Consequently $\Isom\cdot Z\cap\mathcal{S}_I=\{\Gamma_I(Z)\}$, and since $g\cdot Z$ and $Z$ have the same orbit, $\Gamma_I(g\cdot Z)=\Gamma_I(Z)$. (iii) $\Omega_I$ is the complement of the closed null set $\{\tilde w=0\}\cup\{\tilde\zeta=0\}$. On $\Omega_I$ the rotation $e^{-i\varphi}=\overline{\tilde w}/|\tilde w|$ and the sign $\epsilon$ depend continuously on $Z$, and $T_{z_{i_1}}^{-1}$ is jointly continuous, so $\Gamma_I$ is continuous. $\mathcal{S}_I$ is defined by three equalities on $z_{i_1}$ and one equality ($\mathrm{Im}\,w_{i_2}=0$) with two open inequalities on $z_{i_2}$, hence is a submanifold of dimension $3N-4$.
\end{proof}

\begin{proof}[Proof of Theorem~\ref{thm:gauge-sol}]
By \citet[\S4]{scott1983geometries}, $\Isom(\Sol)=\Sol\rtimes D_4$ with $D_4=\Stab(\eb)$ the group \eqref{eq:D4}. In particular every $g$ with $g z_{i_1}=\eb$ has the form $h\circ T_{z_{i_1}}^{-1}$ with $h\in D_4$. The $D_4$-orbit of a point $(x,y,z)$ with $xyz\ne0$ is $\{(\pm x,\pm y,z),(\pm y,\pm x,-z)\}$, which meets the open orthant $F$ in exactly one point: the sign flips $\sigma_x,\sigma_y$ are needed to make the first two coordinates positive, and $\varsigma$ is applied iff $z<0$. This proves existence and uniqueness of $h_Z$, and the arguments of Theorem~\ref{thm:gauge-bundle} give (i)--(iii), $h_Z$ is locally constant on $\Omega_I$, and $\mathcal{S}_I$ is defined by three equalities and three open inequalities, so $\dim\mathcal{S}_I=3N-3$.
\end{proof}

\begin{proof}[Proof of Proposition~\ref{prop:nildist}]
Let $\gamma$ be the geodesic \eqref{eq:nil-geod} with $c\ne0$ and put $s=ct$; by the symmetry $\refl$ we may take $c>0$, so that $s>0$ and $z\ge0$. Its endpoint has $r=\frac{2a}{|c|}|\sin(s/2)|$ and $z=s+\frac{a^2}{2c^2}(s-\sin s)$ with $a^2=1-c^2$. From the first relation, $c^2=4\sin^2(s/2)/(r^2+4\sin^2(s/2))$, hence $\frac{a^2}{2c^2}=\frac{r^2}{8\sin^2(s/2)}$ and $z=\Phi_r(s)$, while $t=s/c$ gives $t^2=s^2+r^2\big(\tfrac{s}{2\sin(s/2)}\big)^2$. For the monotonicity of $\Phi_r$ note that $\frac{d}{ds}\frac{s-\sin s}{\sin^2(s/2)}$ has the sign of $2\sin^3(s/2)-(s-\sin s)\cos(s/2)$, which is positive on $[\pi,2\pi)$ trivially and on $(0,\pi)$ is equivalent, with $u=s/2$, to $\sin^3u>(2u-2\sin u\cos u)\cos u/2$, i.e.\ to $\tan u>u$. Hence $\Phi_r$ increases strictly from $0$ to $\infty$ on $(0,2\pi)$ and $s^\ast$ is unique. Geodesics with $c=0$ are horizontal and reach only $\zeta=0$, where $d=r$ because the base distance is a lower bound attained by the horizontal lift.

\emph{Minimality for $r>0$.} Fix $c\ne0$ with $a=\sqrt{1-c^2}>0$ and consider the family of geodesics with initial data $(a\cos\phi,a\sin\phi,c)$, $\phi\in[0,2\pi)$. By \eqref{eq:nil-geod} they all pass through the same point $(0,0,z(2\pi/|c|))$ at time $t_c=2\pi/|c|$, so the Jacobi field $\partial_\phi\gamma_\phi$ vanishes at $t=0$ and at $t=t_c$ without vanishing identically ($a>0$), and $\gamma_\phi(t_c)$ is conjugate to $\eb$ along $\gamma_\phi$. No geodesic with vertical momentum $c$ minimises beyond $t_c$, i.e.\ beyond $s=2\pi$ \citep[Ch.~11]{docarmo1992riemannian}. A minimising geodesic from $\eb$ to $q$ exists by completeness, is of the form \eqref{eq:nil-geod} with $r(s)=\frac{2a}{|c|}|\sin(s/2)|>0$ when $r>0$, hence has $s\in(0,2\pi)$, and by uniqueness of the root it is the geodesic with $s=s^\ast$.

\emph{The axis.} Let $\gamma=(\beta,\zeta)$ be any curve from $\eb$ to $(0,0,\zeta_0)$, $\zeta_0>0$, with base projection $\beta$ (a closed curve) and vertical velocity $v=\dot\zeta-A_K(\dot\beta)$. Its length is $L=\int\sqrt{|\dot\beta|^2+v^2}\ge\sqrt{L_h^2+V^2}$ by Minkowski's inequality, where $L_h=\int|\dot\beta|$ and $V=\int|v|$. Writing $\mathcal A=\oint_\beta A_K$ for the algebraic enclosed area (with multiplicity, since $\beta$ need not be simple), $\zeta_0=\int v+\mathcal A$, so $V\ge|\zeta_0-\mathcal A|$. For $K=0$ the inequality $L_h^2\ge4\pi|\mathcal A|$ for closed curves with algebraic area is that of \citet{banchoff1971generalization}. For $K=-1$ we use its analogue in the Hyperbolic plane, $L_h^2\ge4\pi|\mathcal A|+\mathcal A^2$ for closed curves with algebraic area, proved by \citet{teufel1991generalization}. Hence
\[
L^2\ \ge\ \min_{\mathcal A\in\R}\Big\{4\pi|\mathcal A|-K\mathcal A^2+(\zeta_0-\mathcal A)^2\Big\}.
\]
For $K=0$ the minimum over $\mathcal A\ge0$ is at $\mathcal A=0$ if $\zeta_0\le2\pi$ (value $\zeta_0^2$) and at $\mathcal A=\zeta_0-2\pi$ if $\zeta_0\ge2\pi$ (value $4\pi(\zeta_0-\pi)$), and $\mathcal A<0$ is never better. Both values are attained, by the vertical segment and by the returning geodesics with $s=2\pi$ (whose length $2\pi/|c|$ at height $\pi(1+c^{-2})$ is $2\sqrt{\pi(\zeta_0-\pi)}$). This proves \eqref{eq:nil-axis}. For $K=-1$ the minimum is at $\mathcal A=(\zeta_0-2\pi)/2$ when $\zeta_0\ge2\pi$, with value $\zeta_0^2/2+2\pi\zeta_0-2\pi^2=4\pi^2\{\tfrac12(\zeta_0/2\pi+1)^2-1\}$, attained by the returning geodesic of Supplement~\ref{app:geodesics}, and at $\mathcal A=0$ otherwise: this proves \eqref{eq:sl-axis}.
\end{proof}

\begin{proof}[Proof of Theorem~\ref{cor:all}]
The non-degenerate sets are: anchors distinct and, on the spheres, no anchor antipodal to $\eb$. In (a) $\tilde z_{i_3}$ off the axis through $\tilde z_{i_2}$ and $\tilde z_{i_4}$ off the plane of the first three. In (b) $\tilde w_{i_2}\ne0$, $\tilde\zeta_{i_2}\ne0$ and $\mathrm{Im}(e^{-i\varphi}\tilde w_{i_3})\ne0$. In each case an isometry $g$ with $g\cdot Z\in\mathcal{S}_I$ satisfies $gz_{i_1}=\eb$, hence $g=h\circ T_{z_{i_1}}^{-1}$ with $h\in\Stab(\eb)$, and it remains to show that $\Stab(\eb)$ acts simply transitively on the admissible positions of the remaining anchors. (a) $\Stab(\eb)=O(3)$ acting linearly on the tangent space at $\eb$ and, through the exponential map, on the model. The elements taking a non-zero $\tilde z_{i_2}$ to the positive first axis form a coset of the $O(2)$ fixing that axis. Of these, exactly two take a point not on the axis into the half-plane $\{x_3=0,x_2>0\}$, namely a rotation about the axis and its composition with the reflection $x_3\mapsto-x_3$, which fixes the half-plane pointwise. The sign of the third coordinate of a fourth anchor off the plane selects one of them. (b) $\Stab(\eb)=O(2)\times\{\pm1\}$, the first factor acting on the base and the second flipping the fibre. The rotation is fixed by $\arg\tilde w_{i_2}$, the fibre flip by $\sgn\tilde\zeta_{i_2}$, and the base reflection by $\sgn\mathrm{Im}(e^{-i\varphi}\tilde w_{i_3})$, each uniquely when the respective quantity is non-zero. (c) and (d) are Theorems~\ref{thm:gauge-bundle} and~\ref{thm:gauge-sol}. Continuity holds chartwise (on $\mathbb{S}^3$ the family $T_p$ can be taken to be left multiplication by unit quaternions, on $\mathbb{S}^2\times\R$ no continuous global choice exists) and the dimension counts follow as in Theorem~\ref{thm:gauge-bundle}(iii): the slice is cut out by three equalities and $\dim\Stab(\eb)$ further equalities (three in (a), one in (b) and (c), none in (d)), the remaining conditions being open inequalities.
\end{proof}

\begin{proof}[Proof of Proposition~\ref{prop:nonrigid}]
\emph{Twisted bundles.} In both examples $d(\eb,q)=1$ because $q$ is the endpoint of the horizontal lift of a base geodesic of length one, and $d(\eb,q')=1$ by \eqref{eq:nil-axis}--\eqref{eq:sl-axis}. Every isometry maps $V$ to $\pm V$, hence maps horizontal curves (curves tangent to $V^\perp$) to horizontal curves. Suppose $g$ maps $\{\eb,q\}$ onto $\{\eb,q'\}$. The horizontal segment from $\eb$ to $q$ is mapped to a horizontal curve of length one joining $\eb$ and $q'=(\bo,1)$, i.e.\ joining two points of the same fibre. Its projection is a closed curve $\beta$ of length one in $B_K$ with algebraic area $\int_\beta A_K=\pm1$. By the isoperimetric inequality for closed curves with algebraic area (Banchoff--Pohl for $K=0$, its analogue for $K=-1$, as in the proof of Proposition~\ref{prop:nildist}), such a loop has length at least $2\sqrt\pi>1$, a contradiction. \emph{Sol.} $\cosh d(\eb,(1,0,0))=1+\tfrac12=\tfrac32$ by \eqref{eq:sol-planes}, and $d(\eb,(0,0,z))=|z|$ because $z$ is 1-Lipschitz and the vertical line is a geodesic, with equality only along it. Every isometry maps $\partial_z$ to $\pm\partial_z$ (left translations satisfy $L_p^*\partial_z=\partial_z$ and $D_4$ maps $\partial_z$ to $\pm\partial_z$), hence maps vertical geodesics to vertical geodesics. If $g$ mapped $\{\eb,q\}$ onto $\{\eb,q'\}$, the unique minimising geodesic between $\eb$ and $q'$ (the vertical segment) would be the image of a minimising geodesic between $\eb$ and $q=(1,0,0)$, which would then be vertical. But no vertical geodesic joins $\eb$ to a point off the $z$-axis. \emph{Dimension count.} The action of $\Isom(\M)$ on $\Omega_I$ is free by Theorems~\ref{thm:gauge-bundle}--\ref{thm:gauge-sol}(ii), so generic orbits have dimension $\dim\Isom(\M)$ and the orbit space has dimension $3N-\dim\Isom(\M)$. This exceeds $N(N-1)/2$ iff $N^2-7N+8<0$ (bundles) or $N^2-7N+6<0$ ($\Sol$), i.e.\ iff $N\le5$ in both cases, and then the fibres of $D$ have positive dimension modulo the orbits.
\end{proof}

\begin{proof}[Proof of Theorem~\ref{thm:generic}]
On $U$ each pair $(z_i,z_j)$ is joined by a unique minimising geodesic free of conjugate points, so $d_\M(z_i,z_j)=|\Exp_{z_i}^{-1}(z_j)|$ with $\Exp_{z_i}$ a local diffeomorphism near the preimage. The geodesic flow of a real-analytic metric is real-analytic, hence so is $D$ on $U$ \citep[Ch.~III]{sakai1996riemannian}. The set $\{Z\in U:\ \text{all }(3N-\dim\Isom)\text{-minors of }dD_Z\text{ vanish}\}$ is an analytic subset of $U$. If it contained an open subset of a connected component $U_0$ it would contain $U_0$, contradicting the hypothesis at $Z_0$. Hence it is closed and nowhere dense in $U_0$. For $Z\in\Omega_I$ the map $g\mapsto g\cdot Z$ is injective (Theorems~\ref{thm:gauge-bundle}--\ref{thm:gauge-sol}(ii)), so $dD_Z$ vanishes on the $\dim\Isom(\M)$-dimensional tangent space of the orbit, which bounds the rank. At $Z^\star=\Gamma_I(Z)\in\mathcal{S}_I$ the tangent space of the slice is complementary to that of the orbit: an orbit tangent vector is $(\xi(z_i^\star))_i$ for a Killing field $\xi$, and it is tangent to $\mathcal{S}_I$ only if $\xi(\eb)=0$, i.e.\ $\xi$ generates the identity component of $\Stab(\eb)$, which is trivial for $\Sol$ and the rotation about the fibre for the bundles. The latter moves $z^\star_{i_2}$ off the real axis unless $\xi=0$. Hence, when $dD$ has maximal rank, its restriction to $T\mathcal{S}_I$ is injective, which gives (i). The Fisher information of $Z^\star$ at fixed $\alpha$ is $(dD|_{T\mathcal{S}_I})^{\tr}\diag\{p_{ij}(1-p_{ij})\}(dD|_{T\mathcal{S}_I})$, positive definite when the restriction is injective and $0<p_{ij}<1$.

(ii) The log-odds map $(Z^\star,\alpha)\mapsto(\alpha-d_{ij})_{i<j}$ has differential $[\mathbf 1,-dD|_{T\mathcal{S}_I}]$, and the Fisher information for $(Z^\star,\alpha)$ is this matrix weighted by $\diag\{p_{ij}(1-p_{ij})\}$. It is non-singular iff the matrix has full column rank $\dim\mathcal{S}_I+1$, which needs at least that many rows, and full rank implies local injectivity of the log-odds map by the inverse function theorem. When $\dim\mathcal{S}_I=\binom N2$ (as for $\Sol$ at $N=6$) the log-odds map goes from a manifold of dimension $\binom N2+1$ to $\R^{\binom N2}$ and is injective on no open set, by invariance of domain, so $(Z^\star,\alpha)$ is not locally identifiable; when moreover $dD|_{T\mathcal{S}_I}$ is invertible, the curve $Z(c)=D^{-1}(D(Z)+c\mathbf 1)$, $\alpha(c)=\alpha+c$, exhibits the ambiguity explicitly.

(iii) Order the coordinates so that the new node's three coordinates come last. The Jacobian for $N+1$ nodes has the block form $\big(\begin{smallmatrix}J_N&0\\ B&C\end{smallmatrix}\big)$ where $C$ is the $N\times3$ matrix of gradients $\nabla_{z_{N+1}}d(z_{N+1},z_j)$ (in coordinates), so $\rank J_{N+1}\ge\rank J_N+\rank C$, and $\rank C=3$ as soon as three of these gradients span $T_{z_{N+1}}\M$, an open condition. Such positions exist: at a maximal-rank configuration the kernel of $dD$ has dimension $\dim\Isom(\M)<N$, so not every node can have incident gradients contained in a plane (each such node would contribute an independent kernel vector supported on that node alone); if node $j_0$ has three incident gradients spanning $T_{z_{j_0}}\M$, placing $z_{N+1}$ close to $z_{j_0}$ gives, by continuity of the geodesic directions, three spanning gradients at $z_{N+1}$. The orbit dimension does not change, so maximal rank is preserved on an open dense set of positions of the new node (the determinant of three gradients is real-analytic in $z_{N+1}$ and not identically zero); the same argument applies to the augmented matrix.
\end{proof}

\begin{lemma}[Equivariance of the wrapped Normal]\label{lem:equivariance}
Let $\Sigma=\diag(\sigma_h^2,\sigma_h^2,\sigma_v^2)$ and $g\in\Isom(\M)$. If $z\sim\LWN_\M(\mu,\Sigma)$ then $g(z)\sim\LWN_\M(g\mu,\Sigma)$. Equivalently $f_{\LWN}(gz\mid g\mu,\Sigma)=f_{\LWN}(z\mid\mu,\Sigma)$.
\end{lemma}
\begin{proof}
Write $g\circ T_\mu=T_{g\mu}\circ h$ with $h=T_{g\mu}^{-1}\circ g\circ T_\mu\in\Stab(\eb)$ (Lemma~\ref{lem:stab}, and $h\in D_4$ in $\Sol$). The stabiliser acts on the chart by a linear map: $h\circ\chi_\M=\chi_\M\circ h_*$ where $h_*$ is the differential of $h$ at $\eb$, an element of $O(2)\times\{\pm1\}$ (rotations of the horizontal components, or a horizontal reflection together with $\delta_3\mapsto-\delta_3$, and in $\Sol$ the maps $\sigma_x,\sigma_y$ and $(\delta_1,\delta_2,\delta_3)\mapsto(\delta_2,\delta_1,-\delta_3)$). This holds because $\chi_{\Nil}$ and $\chi_{\Sol}$ are group exponentials and $\Stab(\eb)$ consists of group automorphisms, and because $\Exp^{\Hb^2}_{\bo}$ commutes with rotations and reflections of the disc. Since $\mathcal{N}_3(0,\Sigma)$ is invariant under such $h_*$, $g(z)=g\,T_\mu\chi(\delta)=T_{g\mu}\chi(h_*\delta)$ with $h_*\delta\sim\mathcal{N}_3(0,\Sigma)$. The density identity follows because isometries preserve $\dvol$.
\end{proof}

\begin{proof}[Proof of Proposition~\ref{prop:target}]
By Theorem~\ref{thm:gauge-bundle}(ii) (resp.\ \ref{thm:gauge-sol}) the map $\Psi:(g,Z^\star)\mapsto g\cdot Z^\star$ is a bijection from $\Isom(\M)\times\mathcal{S}_I$ onto $\Omega_I$, and we may use $(g,Z^\star,\mu^\star)$ with $\mu=g\mu^\star$ as coordinates on $\Omega_I\times\M$. Write $g=T_p\circ h$ with $p\in\M$, $h\in\Stab(\eb)$. We compute the pull-back of $\dvol^{\otimes N}\otimes\dvol(\mu)$. The anchor $z_{i_1}=p$ contributes $\dvol(p)$. Given $p$, the remaining coordinates are $z_i=T_p(h z_i^\star)$, $T_p$ preserves volume, and for $i\ne i_1,i_2$ so does $h$, contributing $\dvol(z_i^\star)$. For the second anchor, $(h,z_{i_2}^\star)\mapsto hz_{i_2}^\star$ is, in the bundles, the polar-coordinate map $(\varphi,\rho,\zeta)\mapsto R_\varphi(\rho,0,\zeta)$ (composed with $\refl$ on the second component of $O(2)$), whose volume element is $S_K(\rho)\,d\rho\,d\varphi\,d\zeta$ because $\dvol_{\M_K}=S_K(\rho)\,d\rho\,d\theta\,d\zeta$. In $\Sol$, $h$ ranges over the finite group $D_4$ and $z^\star_{i_2}$ over $F$ with volume $dx\,dy\,dz$. Finally $\mu=g\mu^\star$ contributes $\dvol(\mu^\star)$. Hence
\[
\dvol^{\otimes N}\otimes\dvol(\mu)=\dvol(p)\,dh\ \otimes\ \Delta_\M(z^\star_{i_2})\,d\rho\,d\zeta\ \otimes\!\!\prod_{i\ne i_1,i_2}\!\!\dvol(z_i^\star)\ \otimes\dvol(\mu^\star),
\]
with $dh$ the Haar measure of $\Stab(\eb)$. By \eqref{eq:invariance} and Lemma~\ref{lem:equivariance} (with the volume-flat prior on $\mu$), the kernel $\kappa(Z,\mu,\ldots)=p(\mathcal Y\mid Z,\alpha)\prod_if(z_i\mid\mu,\Sigma)p(\sigma)p(\alpha)$ is constant along the orbits of $G=\Isom(\M)$, which are non-compact, so $\kappa\,\dvol^{\otimes N}\otimes\dvol(\mu)\otimes d\sigma\,d\alpha$ is an invariant infinite measure. The action of $G$ on $\Omega_I\times\M$ is free and proper with the slice $\mathcal{S}_I\times\M$ as a global section, and the change of variables above is the coarea (Weil) formula for this action: the invariant measure disintegrates as the Haar measure of $G$ on the orbits times the measure $\Delta_\M(z^\star_{i_2})\,d\rho\,d\zeta\prod_{i\ne i_1,i_2}\dvol(z^\star_i)\,\dvol(\mu^\star)$ on the section. The quotient posterior is, by definition, the normalisation of $\kappa$ restricted to the section with respect to this measure. It is a probability measure because $\int\kappa$ over the section is finite: the likelihood is bounded and the priors on $\sigma$ and $\alpha$ are proper, so it suffices to integrate the prior kernel over the slice. by Lemma~\ref{lem:equivariance}, $f(\eb\mid\mu,\Sigma)=f(T_\mu^{-1}\eb\mid\eb,\Sigma)$ and $\mu\mapsto T_\mu^{-1}\eb$ preserves volume, so $\int f(\eb\mid\mu,\Sigma)\,\dvol(\mu)=1$. Each $z_i^\star$, $i\ne i_1,i_2$, integrates to one given $\mu$. For the second anchor, write $z^\star_{i_2}=\mu\star\chi_\M(\delta)$: the density $f(z^\star_{i_2}\mid\mu,\Sigma)=\varphi_3(\delta;0,\Sigma)/J^\chi_\M(\delta)$ is bounded by a Gaussian in $\delta$ divided by $J^\chi_\M(\delta)$, the slice factor $\Delta_\M(z^\star_{i_2})=S_K(\rho)$ satisfies $\rho\le|\delta_h|+d_{B_K}(\bo,\pi(\mu))$ by the triangle inequality in the base, and $S_K(\rho)/J^\chi_\M(\delta)$ is bounded by a polynomial in $|\delta|$ times $e^{d_{B_K}(\bo,\pi(\mu))}$ in all three geometries ($J^\chi_{\Nil}=1$ with $S_0(\rho)=\rho$, and $J^\chi_{\SL}(\delta)=\sinh|\delta_h|/|\delta_h|$ with $S_{-1}(\rho)=\sinh\rho$), so the integral over $(\rho,\zeta)$, and then over $\mu$ against the factor $f(\eb\mid\mu,\Sigma)$, is finite. This is \eqref{eq:target}. The invariant expectation of a $G$-invariant statistic under the invariant measure is undefined as a ratio of infinite quantities, and the quotient posterior is the object that replaces it.
\end{proof}

\begin{proof}[Proof of Proposition~\ref{prop:deficit}]
For Bernoulli laws with log-odds $\eta,\eta'$, $\mathrm{KL}(\mathrm{Bern}(\sigma(\eta))\,\|\,\mathrm{Bern}(\sigma(\eta')))=\tfrac12\sigma'(\eta)(\eta-\eta')^2+O(|\eta-\eta'|^3)$, with $\sigma'(\eta)=\sigma(\eta)(1-\sigma(\eta))=w$, by Taylor expansion of $\eta'\mapsto\sigma(\eta)(\eta-\eta')-\log\{(1+e^{\eta})/(1+e^{\eta'})\}$ about $\eta'=\eta$. The third derivative of $t\mapsto\log(1+e^t)$ is $\sigma'(t)(1-2\sigma(t))$, bounded by $1/(6\sqrt3)$, which gives the remainder bound. Summing over dyads with $\eta_{ij}-\eta'_{ij}=-(d_\M(z_i,z_j)-d_T(z'_i,z'_j)-c)$, $c=\alpha-\alpha'$, gives the expansion for each competitor. The expected log-likelihood ratio equals the divergence. If the infimum is attained with $\Delta_T=0$, all weighted residuals vanish, i.e.\ the distance matrix is realised in $T$ up to an additive constant. For $T=\R^3$ the realisable matrices with an additive constant form a set of dimension at most $3N-6+1$ inside $\R^{N(N-1)/2}$, which can be the whole space only if $3N-5\ge\binom N2$, i.e.\ $N\le5$.
\end{proof}

\begin{proof}[Proof of Theorem~\ref{thm:lecam}]
For a test $\phi$ with $\mathbb{E}_Q\phi\le\nu$ for every competitor law $Q$, $\mathbb{E}_P\phi\le\mathbb{E}_Q\phi+\mathrm{TV}(P,Q)\le\nu+\mathrm{TV}(P,Q)$ for each $Q$. Given the positions the dyads are independent Bernoulli variables under both laws, so $\mathrm{KL}(P\|Q)=\sum_{i<j}\mathrm{KL}_{ij}$, and Pinsker's inequality gives $\mathrm{TV}\le\sqrt{\mathrm{KL}/2}$. The infimum over $(Z',\alpha')$ can be taken because the bound holds for each competitor, and its second-order expansion is Proposition~\ref{prop:deficit}.
\end{proof}

\begin{lemma}[Distance in normal coordinates]\label{lem:normal}
Let $\M$ be a real-analytic Riemannian manifold, $o\in\M$, and $u,v\in T_o\M$ with $|u|,|v|\le\varepsilon$ small. Then $d(\exp_ou,\exp_ov)^2=|u-v|^2-\tfrac13\mathrm{Rm}_o(u,v,v,u)+O(\varepsilon^5)$.
\end{lemma}
\begin{proof}
In normal coordinates $g_{ij}(x)=\delta_{ij}-\tfrac13R_{ikjl}x^kx^l+O(|x|^3)$, so for $w\in\R^3$ and $|x|\le\varepsilon$, $g_x(w,w)=|w|^2-\tfrac13\mathrm{Rm}_o(w,x,x,w)+O(\varepsilon^3|w|^2)$. Let $\sigma(t)=u+t(v-u)$ and $w=v-u$. Since $\mathrm{Rm}_o(w,\sigma(t),\sigma(t),w)=\mathrm{Rm}_o(w,u,u,w)$ for all $t$ (the term in $w$ vanishes by antisymmetry), the energy of $\sigma$ is $E(\sigma)=\int_0^1g_{\sigma(t)}(w,w)\,dt=|w|^2-\tfrac13\mathrm{Rm}_o(w,u,u,w)+O(\varepsilon^5)=|u-v|^2-\tfrac13\mathrm{Rm}_o(u,v,v,u)+O(\varepsilon^5)$, using $\mathrm{Rm}_o(v-u,u,u,v-u)=\mathrm{Rm}_o(v,u,u,v)$. The minimising geodesic $\gamma$ from $\exp_ou$ to $\exp_ov$, parametrised on $[0,1]$, has $d^2=E(\gamma)\le E(\sigma)$. Conversely, the Christoffel symbols are $O(\varepsilon)$ in the ball, so $\ddot\gamma=O(\varepsilon^3)$ and $\gamma-\sigma=O(\varepsilon^3)$ with derivative $O(\varepsilon^3)$; since $\gamma$ is a critical point of $E$ among curves with the same endpoints, $E(\sigma)-E(\gamma)$ is of second order in this difference, $O(\varepsilon^6)$. Hence $d^2=E(\sigma)+O(\varepsilon^6)$, which is the claim.
\end{proof}

\begin{proof}[Proof of Theorem~\ref{thm:smallscale}]
Write $\varepsilon u_i$ for the normal coordinates and $X_\varepsilon=\varepsilon X$. \emph{Expansion.} By Lemma~\ref{lem:normal} (see also \citealp{gray2004tubes}), $d(\exp_o u,\exp_o v)^2=|u-v|^2-\tfrac13\mathrm{Rm}_o(u,v,v,u)+O((|u|+|v|)^5)$, so with $d^E_{ij}=\varepsilon|u_i-u_j|$, $D_\M(Z)_{ij}=d^E_{ij}+\varepsilon^3r^\M_{ij}+O(\varepsilon^4)$ where $r^\M_{ij}=-\mathrm{Rm}_\M(u_i,u_j,u_j,u_i)/(6|u_i-u_j|)$, and the same holds for any configuration in $T$ with $\mathrm{Rm}_T$. \emph{The competitor family.} Near $D_E(X_\varepsilon)$ the set $\mathcal E_T=\{D_T(Z')+c\mathbf 1\}$ is, by the expansion applied to $Z'$ written in normal coordinates $\varepsilon u'_i$ of $T$ at a base point (and, for $T=\Hb^2\times\R$, with its fibre direction $n$), the set $\{D_E(\varepsilon X')+\varepsilon^3r^T(X';n)+c\mathbf 1+O(\varepsilon^4)\}$. Since $X$ is infinitesimally rigid, the Euclidean-plus-constant vectors $\{D_E(\varepsilon X')+c\mathbf 1\}$ form near $D_E(X_\varepsilon)$ a submanifold $\mathcal E$ of dimension $3N-5$ with tangent space spanned by the columns of $J_E(X)$ and $\mathbf 1$, and normal space $\Omega(X)$. Because $\mathcal E$ is invariant under scaling, its second fundamental form at $D_E(X_\varepsilon)$ is $O(1/\varepsilon)$. \emph{Weighted distance.} The weights $w_{ij}=\sigma'(\alpha-d_{ij})$ are $\sigma'(\alpha)+O(\varepsilon)$, uniformly. The $w$-distance from $D_\M(Z)=D_E(X_\varepsilon)+\varepsilon^3r^\M+O(\varepsilon^4)$ to $\mathcal E_T$ is therefore the $w$-distance from $\varepsilon^3(r^\M-r^T(X;n))$ to the tangent space of $\mathcal E$, up to $O(\varepsilon^4)$ from the remainders, $O(\varepsilon^3\cdot\varepsilon^3/\varepsilon)=O(\varepsilon^5)$ from the curvature of $\mathcal E$, and $O(\varepsilon^3\cdot\varepsilon^2)=O(\varepsilon^5)$ from replacing $r^T(X')$ by $r^T(X)$ when $|X'-X|=O(\varepsilon^2)$. Since the weights are $w_0+O(\varepsilon)$, the $w$-orthogonal projection differs from the orthogonal projection $\Pi$ onto $\Omega(X)$ by $O(\varepsilon)$, and minimising over the free orientation $n$ of an anisotropic competitor gives a $w$-distance $\sqrt{w_0}\,\varepsilon^3\min_n\|\Pi r\|\,(1+O(\varepsilon))$, whose square is the second statement of \eqref{eq:smallscale} for the Fisher deficit $\Delta_T^2$. When $\Omega(X)=\{0\}$ the leading term vanishes and the remainders give $O(\varepsilon^7)$. \emph{Exact divergence.} Fix $\delta_0>0$. A competitor with $\max_{ij}|\eta_{ij}-\eta'_{ij}|\ge\delta_0$ has $\sum\mathrm{KL}_{ij}\ge2(\sigma(\eta_{ij})-\sigma(\eta'_{ij}))^2\ge\kappa(\delta_0)>0$ on that dyad, with $\kappa$ independent of $\varepsilon\le1$. A competitor with $\max|\eta_{ij}-\eta'_{ij}|\le\delta_0$ has, by the remainder bound of Proposition~\ref{prop:deficit}, $\sum\mathrm{KL}_{ij}=\tfrac12\sum w_{ij}(\eta_{ij}-\eta'_{ij})^2(1+O(\delta_0))$. Since $\Delta_T^2=O(\varepsilon^6)<\kappa(\delta_0)$ for $\varepsilon$ small, the infimum is attained by competitors of the second kind up to a factor $1+O(\delta_0)$, and the Fisher minimiser has $|\eta_{ij}-\eta'_{ij}|=O(\varepsilon^3)\le\delta_0$. Letting $\delta_0\to0$ after $\varepsilon\to0$ gives the first equality of \eqref{eq:smallscale}. \emph{Curvature forms.} In the orthonormal left-invariant frame in which the Lie brackets of $\Nil$, $\Sol$ and $\SL$ are diagonal in Milnor's sense, the curvature operator is diagonal on $e_1\wedge e_2,\ e_2\wedge e_3,\ e_3\wedge e_1$ with the sectional curvatures of the coordinate planes, $(-\tfrac34,\tfrac14,\tfrac14)$ for $\Nil$, $(-\tfrac74,\tfrac14,\tfrac14)$ for $\SL$, $(1,-1,-1)$ for $\Sol$ \citep{milnor1976curvatures}, so $\mathrm{Rm}(u,v,v,u)=\sum_kK_k(u\times v)_k^2$, which is \eqref{eq:curvforms}. For the products and the constant-curvature spaces the forms are $K_{\mathrm{base}}U_3^2$ and $K|U|^2$.
\end{proof}

\begin{proof}[Proof of Proposition~\ref{prop:directed}]
The first statement follows from $d_\M(z_i,z_j)=d_\M(z_j,z_i)$, $v_{ji}=-v_{ij}$ and \eqref{eq:triangle}; for a coboundary $u_{ij}=\phi_j-\phi_i$ the cyclic sum telescopes to zero. For the second, let $\eta,\eta'$ be the log-odds of the model and of a member of the gradient class, $\delta_{ij}=\eta_{ij}-\eta'_{ij}$, and split $\delta_{ij}=\delta^s_{ij}+a_{ij}$ into its symmetric and antisymmetric parts, so that $a_{ij}=\tfrac12(s_{ij}-s'_{ij})=\gamma v_{ij}-\gamma'u_{ij}$. Then $\delta_{ij}^2+\delta_{ji}^2=2(\delta^s_{ij})^2+2a_{ij}^2\ge2a_{ij}^2$ and $\tfrac12\sum_{i\ne j}w_{ij}\delta_{ij}^2\ge\sum_{i<j}\tilde w_{ij}a_{ij}^2$. On the complete graph with the edge inner product, the coboundaries $\{u=\delta\phi\}$ have orthogonal complement the cycle space, and the orthogonal projection of an antisymmetric $v$ on the cycle space has entries $(Pv)_{ij}=N^{-1}\sum_k(v_{ij}+v_{jk}+v_{ki})$: indeed the projection on the coboundaries is $\delta\phi$ with $\phi=L^+\delta^{\tr}v$, $L=NI-\mathbf 1\mathbf 1^{\tr}$ the Laplacian of the complete graph, $\phi_j=N^{-1}\sum_iv_{ij}$, and $v_{ij}-(\phi_j-\phi_i)=N^{-1}\sum_k(v_{ij}-v_{kj}+v_{ki})$. Hence, with $\tilde w_{\min}=\min\tilde w_{ij}$, $\sum_{i<j}\tilde w_{ij}(\gamma v_{ij}-\gamma'u_{ij})^2\ge\tilde w_{\min}\gamma^2\|Pv\|^2$ for every coboundary $\gamma'u$, and by \eqref{eq:triangle} $(Pv)_{ij}=-\bar A_{ij}$, which gives \eqref{eq:directed-deficit}.
\end{proof}

\begin{proof}[Proof of Lemma~\ref{lem:holonomy}]
The expressions \eqref{eq:cochain} for $v_{ij}$ are the third components of $p_i^{-1}\star p_j$ read off from \eqref{eq:nil-law} and \eqref{eq:sl-translate}. For the intrinsic description, apply the isometry $T_{p_i}^{-1}$, which maps $p_i$ to $\eb$, preserves $V$ (it belongs to the identity component) and therefore maps horizontal lifts to horizontal lifts. In the coordinates centred at $\eb$ the base geodesic from $\bo$ to $w'=\pi(T_{p_i}^{-1}p_j)$ is radial, $A_K$ vanishes on it, and its horizontal lift through $\eb$ is $\{(\cdot,0)\}$. The fibre coordinate of $T_{p_i}^{-1}p_j$ relative to this lift is thus $v_{ij}$. On the other hand, in the fixed trivialisation $(w,\zeta)$ a horizontal lift satisfies $d\zeta=A_K(\dot\beta)$, so the fibre displacement of $p_j$ relative to the horizontal lift through $p_i$ of the base geodesic $[w_i,w_j]$ equals $\zeta_j-\zeta_i-\int_{[w_i,w_j]}A_K$. Comparing, $A_K(w_i,w_j)=\int_{[w_i,w_j]}A_K$. Summing over the three sides of the geodesic triangle and applying Stokes' theorem with $dA_K=\dvol_{B_K}$ gives \eqref{eq:triangle}, the sign being that of the orientation of the boundary. (For $K=0$ the identity is also elementary: $\tfrac12(x_iy_j-y_ix_j)$ is the signed area of the triangle $(\bo,w_i,w_j)$.) The identity was also checked numerically to six decimals for random triangles in both geometries (code accompanying the paper).
\end{proof}

\begin{proof}[Proof of Lemma~\ref{lem:symmetric}]
The proposal density with respect to $\dvol(z')$ is $q(z'\mid z)=\varphi_3(\delta;0,\Sigma_s)/J^\chi_\M(\delta)$ with $\delta=\chi^{-1}_\M(z^{-1}\star z')$, and the reverse displacement is $\delta'=\chi^{-1}_\M(z'^{-1}\star z)$. In $\Nil$ and $\Sol$, $\chi_\M$ is the group exponential, so $\chi(\delta)^{-1}=\chi(-\delta)$ and $\delta'=-\delta$. The Jacobians \eqref{eq:jac} are even, so $q(z\mid z')=q(z'\mid z)$. In $\SL$, write $z=(w,\zeta)$, $u=\Exp_{\bo}(\delta_h)$ and $z'=(w\oplus u,\ \zeta+\delta_v+2\arg(1+\bar w u))$. Then $|\delta'_h|=d_{\Hb^2}(w',w)=d_{\Hb^2}(u,\bo)=|\delta_h|$, and using $1-\bar w'w=(1-|w|^2)/(1+w\bar u)$ one finds $\delta'_v=\zeta-\zeta'+2\arg(1-\bar w'w)=-\delta_v$. Thus $\delta'=(R(-\delta_h),-\delta_v)$ for a rotation $R$. When the horizontal block of $\Sigma_s$ is isotropic, $\varphi_3(\delta';0,\Sigma_s)=\varphi_3(\delta;0,\Sigma_s)$, and $J^\chi_{\SL}$ depends only on $|\delta_h|$.
\end{proof}

\section{Computational details}\label{app:details}

\paragraph{Initialisation.} Shortest-path distances are computed on the observed graph (disconnected pairs receive the diameter plus one), embedded in $\R^3$ by classical multidimensional scaling, rescaled to unit standard deviation and multiplied by $1.2$, perturbed by $\mathcal{N}(0,0.05^2)$ noise to avoid exact coincidences, and mapped to $\M$ by $\chi_\M$ (so the third MDS coordinate becomes the fibre coordinate in $\Nil$ and $\SL$, and the height in $\Sol$). The configuration is then gauge-fixed by $\Gamma_I$, and $\alpha$ is set to the maximiser of the likelihood over a grid of step $0.1$ on $[-6,8]$. The prior centre is initialised at $\eb$ and the scales at the root-mean-square horizontal and vertical chart displacements of the initial configuration. These scale values initialise both schemes (and are the fixed values of the plug-in variants).

\paragraph{Step sizes and settings.} MCMC: $s_z=0.45$ ($0.35$ for $N\ge60$), $s_\alpha=0.2$, $s_\mu=0.2$. Acceptance rates for the positions were between $0.3$ and $0.65$ in all runs and $0.4$--$0.8$ for $\alpha$. The Jacobian experiment of Section~\ref{sec:exp-ident} holds the scales fixed. BBVI: $n_s=10$, $\eta=0.02$, RMSProp decay $0.9$, initial log-scales $\log0.3$ for the positions and $\log0.2$ for $\alpha$. The chart Jacobian $\partial\delta_i/\partial\tilde m_i$ in \eqref{eq:scores} is evaluated in closed form for $\Nil$ and by central differences with step $10^{-4}$ for $\Sol$ and $\SL$ (six extra chart evaluations per node and sample, a negligible cost next to the distance evaluations).

\paragraph{Chart Jacobians.} In $\Nil$, $\delta=\chi^{-1}(\chi(\tilde m)^{-1}\cdot z)=\big(z_1-\tilde m_1,\ z_2-\tilde m_2,\ z_3-\tilde m_3-\tfrac12(\tilde m_1z_2-\tilde m_2z_1)\big)$, so
\[
\frac{\partial\delta}{\partial\tilde m}=\begin{pmatrix}-1&0&0\\0&-1&0\\-z_2/2&z_1/2&-1\end{pmatrix},\qquad \nabla_\delta\log J^\chi_{\Nil}=0 .
\]
In $\Sol$, with $\phi(t)=(1-e^{-t})/t$ and $\mu=\chi_{\Sol}(\tilde m)=(\tilde m_1\phi(\tilde m_3),\tilde m_2\phi(-\tilde m_3),\tilde m_3)$,
\[
\delta_1=\frac{e^{\tilde m_3}(z_1-\mu_1)}{\phi(z_3-\tilde m_3)},\qquad \delta_2=\frac{e^{-\tilde m_3}(z_2-\mu_2)}{\phi(\tilde m_3-z_3)},\qquad \delta_3=z_3-\tilde m_3,\qquad
\nabla_\delta\log J^\chi_{\Sol}=\Big(0,\,0,\,\coth\tfrac{\delta_3}{2}-\tfrac{2}{\delta_3}\Big),
\]
which is differentiated numerically. In $\SL$, $\delta_h=\Log_{\bo}\big((-a)\oplus w\big)$ and $\delta_v=z_3-\tilde m_3+2\arg(1-\bar aw)$ with $a=\Exp_{\bo}(\tilde m_1,\tilde m_2)$, and $\nabla_\delta\log J^\chi_{\SL}=(\coth r-1/r)\,\delta_h/r$, $r=|\delta_h|$.

\paragraph{Competitors and baselines.} The constant-curvature and product geometries are implemented with the same interface as the three Thurston geometries: $\Hb^3$ in the Poincar\'e ball with M\"obius addition as translation, chart $\Exp_{\bo}$ and Jacobian $(\sinh r/r)^2$, $\mathbb{S}^3$ and $\mathbb{S}^2$ in stereographic coordinates with the rotation taking the north pole to $\mu$ as translation, chart $\tan(|v|/2)v/|v|$ truncated at $|v|\le0.98\pi$ and Jacobian $(\sin r/r)^{d-1}$, $\Hb^2\times\R$ and $\mathbb{S}^2\times\R$ as products (translation and chart act componentwise). The adjacency spectral embedding baseline computes the rank-3 eigen-approximation of the adjacency matrix with hidden dyads imputed by the previous approximation (30 iterations) and predicts the clipped low-rank entries. The degree-corrected block model uses regularised spectral clustering into $K$ groups followed by the closed-form Poisson maximum-likelihood estimates of block and degree parameters on the observed dyads, with $K\in\{2,\dots,6\}$ chosen by BIC. The Chung--Lu baseline predicts $d_id_j/\sum_kd_k$ from observed degrees. All baselines use only the observed dyads of each split.

\paragraph{Dyad subsampling.} For $N>200$ each variational iteration draws, for each of the $n_s$ samples, a uniform subset of $M=20\,000$ of the $\binom N2$ dyads and replaces $\sum_{j}\ell_{ij}$ in step 2 of Algorithm~\ref{alg:bbvi} by $\binom N2M^{-1}\sum_{j:(i,j)\in\text{subset}}\ell_{ij}$, which is unbiased for the full sum. The per-iteration cost is then $O(S(M+N))$, and a fit of a $700$-node network takes $20$--$90$ seconds.

\paragraph{Scale estimation and degeneracy.} Two ways of treating the scales fail. (a) Holding them at multidimensional-scaling plug-in values under-covers (Section~\ref{sec:exp-cal}), because the plug-in values are typically too small. (b) Point-estimating them inside the anchored variational scheme by the M-step $\sigma^2\leftarrow$ mean squared chart displacement lets the objective increase without bound: with all variational means moved to $\mu$ and the variational scales shrunk in proportion to $\sigma$, the prior terms of the $N$ nodes contribute $-3N\log\sigma$, the entropy of the family only $(3N-4)\log\tilde s$ (the first anchor is fixed and the second has two free coordinates), the slice factor $\log\rho_{i_2}$ contributes $\log\sigma$, and the point-estimated centre contributes no entropy at all, so the objective grows like $3\log(1/\sigma)$. On a simulated $\Nil$ network of $40$ nodes the objective rose from $-420$ to $-405$ over $2000$ iterations while the sampled log-likelihood fell from $-352$ to $-396$ and the correlation between fitted and true distances fell from $0.56$ to $-0.06$. The family of Section~\ref{sec:vi}, with variational factors for the scales and the centre, has no such direction and is the one used in all reported fits. (c) Sampling the scales under the improper prior $p(\sigma)\propto1/\sigma$ produced, on the same networks, a vanishing vertical scale ($\sigma_v\approx0.1$). The proper $\mathrm{InvGamma}(2,2)$ prior removes that degeneracy but shrinks the scale at small $N$ (Section~\ref{sec:exp-cal}), and on nearly separable real networks it lets the sampler follow the expanding ray of Remark~\ref{rem:separation}.

\paragraph{Baselines.} Constant density: the observed density of the training dyads. $\beta$-model \citep{chatterjee2011random}: $\logit p_{ij}=b_i+b_j$, maximum likelihood on the observed dyads by Newton steps. Logistic random dot-product model \citep{young2007random,athreya2017statistical}: the rank-3 spectral embedding of the adjacency matrix with hidden dyads imputed by the previous low-rank approximation (30 iterations), followed by a logistic regression of the observed ties on the inner products (intercept and slope). Spectral eigenmodel \citep{hoff2008modeling}: the same embedding with a logistic regression on the three coordinate products, i.e.\ $\logit p_{ij}=a+\sum_k\lambda_kx_{ik}x_{jk}$ with free signs. Degree-corrected block model \citep{karrer2011stochastic}: regularised spectral clustering into $K$ groups, closed-form Poisson maximum-likelihood estimates of the block and degree parameters on the observed dyads with additive smoothing, conversion $p=1-e^{-\lambda}$, and $K\in\{1,\dots,8\}$ by BIC on the Bernoulli likelihood of the observed dyads. It is a spectral estimator, not a fully optimised block model. Chung--Lu \citep{chung2002average} is computed in the code but not tabulated: its hub probabilities exceed one before clipping and it is poorly calibrated. All probabilities are clipped to $[10^{-4},1-10^{-4}]$.

\paragraph{Rank witnesses.} Table~\ref{tab:rigidity} (Section~\ref{sec:exp-ident}) reports the singular values of the differential of the distance map restricted to the tangent space of the gauge slice, and of its augmentation by the intercept direction. For $\Nil$ the exact distance of Proposition~\ref{prop:nildist} is differentiated by central differences with $h=10^{-5}$ at $100$ configurations $\chi_\M(\varepsilon)$, $\varepsilon\sim\mathcal{N}_3(0,I)$, gauge-fixed. For $\Sol$ and $\SL$ we use the first-variation formula: for a minimising geodesic $\gamma:[0,L]\to\M$ from $p$ to $q$ parameterised by arc length and free of conjugate points, $\nabla_qd(p,q)=\dot\gamma(L)^\flat$. For each ordered pair we solve the boundary-value problem from $\eb$ to $p^{-1}\star q$ by multi-start least squares on the initial direction and the length, accept a solution when its length agrees with the table distance to within $0.02$, convert the terminal unit velocity to a covector with the metric matrix at the endpoint, and apply the chain rule through the translation $q\mapsto p^{-1}\star q$. These are floating-point computations at chosen configurations (scale $0.9$). For $\Nil$ we add a validated certificate (\texttt{certify\_nil.py}): the configuration is a dyadic rational point, every quantity is an interval computed with outward rounding in 80-bit precision, the root $s^\ast$ of $\Phi_r(s)=\zeta$ is enclosed by interval bisection with the sign change verified on the enclosure's endpoints, the derivatives follow the implicit-function formulas of Proposition~\ref{prop:nildist}, and full column rank of $J$ (and of $[\mathbf 1,-J]$) is certified by an interval Cholesky factorisation of $J^{\tr}J-\lambda I$ with all pivots bounded away from zero, which proves $s_{\min}\ge\sqrt\lambda$. The certified values are $s_{\min}(J)\ge\certSixLower$, $s_{\min}([\mathbf 1,-J])\ge\certSixAug$ at $N=6$ and $\certSevenLower$, $\certSevenAug$ at $N=7$, with enclosure widths \certWidth. The certificate relies on the correctness of the interval library and of the formulas, which were checked against brute-force shooting in Supplement~\ref{app:geodesics}.

\paragraph{Predictive quantities.} Posterior predictive probabilities are averages of $\{1+\exp(d_\M(z_i,z_j)-\alpha)\}^{-1}$ over $200$ retained draws (MCMC) or $200$ samples from $q$ (BBVI). The posterior mean distance matrices in Table~\ref{tab:sims} are the corresponding averages of $d_\M(z_i,z_j)$. AUC is the Mann--Whitney area under the ROC curve of $\hat p_{ij}$ against $y_{ij}$ over the relevant dyads, with tied scores given average ranks (a constant predictor scores $0.5$), and log-loss is $-\frac1{|\mathcal{D}|}\sum_{(i,j)\in\mathcal{D}}\{y_{ij}\log\hat p_{ij}+(1-y_{ij})\log(1-\hat p_{ij})\}$.

\end{document}